\documentclass[11pt,1p,authoryear]{elsarticle}

\usepackage{setspace}
\usepackage{amsmath, amssymb, amsfonts} 
\usepackage{graphicx} 
\usepackage{booktabs} 
\usepackage{multirow} 
\usepackage{mathrsfs}
\usepackage{caption}  
\usepackage{subcaption} 
\usepackage{hyperref}  
\usepackage{geometry}
\usepackage{amssymb,commath,graphics,amsmath,amsthm,amsopn,amstext,
amsfonts,hyperref,fullpage,graphicx,fancybox,verbatim,
changepage,framed,booktabs,tabularx}
\usepackage[utf8]{inputenc}
\usepackage{algorithm,algorithmicx,mathtools}
\usepackage{multirow, array}
\usepackage{xcolor}
\usepackage{bm}
\usepackage[colorinlistoftodos,prependcaption,textsize=small]{todonotes}
\usepackage{cancel}
\usepackage{subcaption}
\usepackage{blkarray}
\usepackage{xurl}

\newtheorem{theorem}{Theorem}
\newtheorem{proposition}[theorem]{Proposition}
\newtheorem{lemma}[theorem]{Lemma}

\newtheorem{definition}[theorem]{Definition}

\newtheorem{remark}[theorem]{Remark}

\newcommand{\epc}{\hspace{1pc}}
\newcommand{\gap}{\vspace{0.1in}}

\newcommand{\onebld}{{\mbox{\boldmath $1$}}}

\newcommand{\wt}{\widetilde}
\newcommand{\thalf}{\tfrac{1}{2}}

\allowdisplaybreaks

\makeatletter
\def\ps@pprintTitle{%
  \let\@oddhead\@empty
  \let\@evenhead\@empty
  \let\@oddfoot\@empty
  \let\@evenfoot\@empty
}
\makeatother

\begin{document}

\begin{frontmatter}

\title{Traffic Equilibrium in Mixed-Autonomy 
Network \\
with Capped Customer Waiting}
% Continuous Customer Waiting Functions}
% 

\author[inst1]{Jiaxin Hou}
\ead{jiaxinho@usc.edu}

\author[inst2]{Kexin Wang}
\ead{kwang255@usc.edu}

\author[inst2]{Ruolin Li\corref{cor1}}
\ead{ruolinl@usc.edu}

\author[inst1]{Jong-shi Pang}
\ead{jongship@usc.edu}

\cortext[cor1]{Corresponding author}

\address[inst1]{Daniel J. Epstein Department of Industrial and Systems Engineering, University of Southern California, Los Angeles, California
90089-0193, U.S.A.}
\address[inst2]{Sonny Astani
Department of Civil and
Environmental Engineering, University of Southern California, Los Angeles, California
90089-2531, U.S.A.}

\begin{abstract}
This paper develops a unified modeling framework to capture the equilibrium-state interactions among 
ride-hailing companies, travelers, and traffic of mixed-autonomy transportation networks.
Our framework integrates four interrelated sub-modules: (i) the operational behavior of representative ride-hailing Mixed-Fleet Traffic Network Companies (MiFleet TNCs) managing autonomous vehicle (AV) and human-driven vehicle (HV) fleets, (ii) traveler mode-choice decisions taking into account travel costs and waiting time, (iii) 
capped customer waiting times 
% as truncations of queue-based models 
to reflect
the option available to travelers not to wait for
TNCs' service beyond his/her patience and to 
resort to existing travel modes,
% of flows, demands, and trip times as 
% illustrated by a
% truncated, congestion-dependent, queue-based model},
and (iv) a flow-dependent traffic congestion model for travel times. A key modeling feature distinguishes AVs and HVs across the pickup and service (customer-on-board) stages: AVs follow Wardrop pickup routes but may deviate during service under company coordination, whereas HVs operate in the reverse manner. The overall framework is formulated as a Nonlinear Complementarity Problem (NCP), which is equivalent to a Variational Inequality (VI) formulation based
on which the existence of a variational equilibrium solution
to the traffic model is established.  Numerical experiments examine how AV penetration and 
Wardrop relaxation factors, which bound route deviation, affect company, traveler, and system performance to various degrees.
The results provide actionable insights for 
policymakers on regulating AV adoption and company 
vehicle deviation behavior in modern-day 
traffic systems that are fast changing due to 
the advances in
technology and information accessibility.

\end{abstract}

\end{frontmatter}

\newpage
%-------------------------------
% MAIN TEXT
%-------------------------------
\section{Introduction}
% As autonomous vehicle (AV) technologies continue to mature, leading Traffic Network Companies (TNCs), e.g., Uber and Lyft, are starting to integrate AVs into their service offerings. 
% Major consulting firms project a significant rise in AV adoption in the coming decades. By 2030, 12\% of new passenger cars are projected to feature Level 3+ (L3+) autonomous capabilities—vehicles capable of fully autonomous operations while requiring human intervention when necessary. This figure is expected to reach 37\% by 2035, generating \$300 to \$400 billion in the passenger car market, according to a 2023 McKinsey \& Company report \cite{deichmann2023autonomous}. Meanwhile, Boston Consulting Group (BCG) \cite{mosquet2015revolution} forecasts that by 2035, autonomous features could be present up to 25\% of new vehicles, representing a \$77 billion market. 

The rapid advancement of autonomous vehicle (AV) technology is fundamentally
transforming transportation systems. Recent projections suggest that by 2030, approximately 12\% of new
passenger cars will feature Level 3+ autonomous capabilities, with this figure potentially rising to 37\% by
2035 and generating \$300-\$400 billion in market value \cite{deichmann2023autonomous}.
As this technological transition accelerates, a diverse set of firms has emerged or is actively redefining their roles in the emerging AV ecosystem, each pursuing distinct pathways toward automation. Technology-driven pioneers, such as Waymo, operate fully autonomous ride-hailing services across several U.S. cities, providing over 
250,000 paid trips per week~\cite{WaymoScaling2025}. 
% and planning to expand operations to additional cities by 2026 
% https://waymo.com/blog/2025/05/scaling-our-fleet-through-us-manufacturing
% and reporting substantial safety gains compared with human drivers.
% https://waymo.com/safety/
Tesla, as a new entrant, 
% adopts a vertically integrated model, developing both its vehicles and autonomous software in-house and deploying a purpose-built Cybercab for its robotaxi fleet \cite{Tesla10K2024}.
% https://ir.tesla.com/_flysystem/s3/sec/000110465925042659/tm252787d2_10ka-gen.pdf
 launched robotaxi services with safety drivers in Austin in 2025 and intends to expand to additional cities while progressing toward full autonomous operations\cite{TeslaQ22025}.
 % https://www.tesla.com/sites/default/files/downloads/TSLA-Q2-2025-Update.pdf
%  vertically = in-house vehicle production + in-house software + direct service operation.
Meanwhile, platform incumbents such as Uber and Lyft are approaching automation through strategic partnerships. Uber began integrating Waymo's autonomous vehicles into its ride-hailing platform in Atlanta in 2025 \cite{UberWaymo2025}, while Lyft plans to introduce Waymo's fully autonomous vehicles to its service and expand operations to Nashville in 2026 \cite{LyftWaymo2025}. To examine how traffic network companies coexist in the transportation system, we introduce the concept of \textbf{Mixed-Fleet Traffic Network Companies (MiFleet-TNCs)}, which are service providers that operate both AVs and human-driven vehicles (HVs), and propose a framework in which multiple heterogeneous MiFleet TNCs interact with travelers and traffic. This framework allows us to study the effects of AV adoption on company profitability and system performance, providing insights for MiFleet TNCs on fleet planning and pricing strategies, for travelers on traveling choices, and for regulators on policy design.

The involvement of AVs introduces substantial complexity to both the overall performance of the transportation system and its constituent components. While the advanced coordination and routing capabilities of AVs can enhance vehicle distribution efficiency, improve safety, reduce fuel consumption, and alleviate—or at least avoid exacerbating—traffic congestion \cite{fagnant2015preparing, yang2017impact, olia2016assessing, stern2018dissipation, rossi2018routing}, their profit-driven deployment may intensify traffic imbalances, particularly when AVs are disproportionately allocated to high-demand and high-congestion areas. This behavior can reinforce congestion, reduce fleet efficiency, and ultimately lower long-term profitability \cite{li2020game,li2021employing,mehr2021game}. These dynamics become even more complex in mixed-autonomy environments, where AVs and HVs coexist under differing levels of control and interact with self-interested travelers.
Consequently, the system-wide impacts of automation depend critically on the AV penetration rate. Limited adoption may result in modest efficiency gains and coordination frictions with HVs \cite{zeng2025modeling, zheng2020smoothing, li2020leveraging, huang2019stabilizing}, whereas higher penetration can improve network performance and reduce delays \cite{abdeen2022evaluating, obaid2022autonomous, li2024managing}. However, excessive automation may induce additional travel demand and deadheading, leading to increased vehicle miles traveled (VMT) and vehicle hours traveled (VHT) \cite{childress2015using, auld2017analysis, horl2019fleet, levin2015effects}. Developing a unified framework to understand AV penetration rate is therefore essential for planners and policymakers to determine when adoption delivers net benefits and when it may introduce new challenges.

The objective of this project is to develop a macroscopic unified modeling framework to characterize the \textbf{Mixed-Autonomy General Equilibrium with Customer Waiting Functions
(MAGE-CW)}
% Capped Queuing (MAGE-CQ)}, 
assessing the multifaceted impacts of AV adoption on various stakeholders, including MiFleet TNCs with AVs and HVs, travelers, and traffic congestion. The model aims
to describe a transportation system in which profit-making companies operate AV and HV fleets and coordinate their routing strategies, while solo drivers, acting as Nash players, selfishly choose their travel routes.
% The competitions among companies impact the traffic network performance by fleet pricing
% \cite{kaplan2024modeling}, deployment strategies \cite{ao2024control}, and routing and operational strategies
% \cite{chen2020path}. 
% Private vehicles, acting as Nash players (or Wardrop players), choose
% their routes minimizing individual travel delays. In contrast, AV fleets are coordinated and serve to optimize firm-specific profit objectives. Although
% companies also own human-driven fleets, a key distinction is that AV fleets are entirely subject to
% company directives, enabling coordinated dispatch, routing, and rebalancing, whereas HV fleets retain a
% certain level of individual autonomy, preventing full centralized control. Our goal is to model these coupled strategic interactions across heterogeneous firms, travelers, and traffic dynamics, thereby revealing the equilibrium mechanisms underlying mixed-autonomy transportation systems. In parallel, our framework provides insights for system planners to develop adaptive regulatory measures that guide the sustainable evolution of a mixed-autonomy system. 
% TNCs have such incentives, to balance demand and supply, to motivate drivers to travel unoccupied to reposition or respond to customer requests \cite{yan2020dynamic, ma2022spatio}. However, the extent to which AVs exacerbate or alleviate these dynamics critically depends on the overall AV penetration rate.
% Some studies suggest that introducing autonomous vehicles (AVs) can improve the efficiency and overall performance of transportation systems.
We develop a unified framework building upon prior works \cite{ban2019general,gu5461575generalized} in e-hailing services, advancing the models therein
to address the operations of multiple competing 
AV and HV fleets 
under traffic congestion, demand-side user preferences, and regulatory constraints within a consistent equilibrium structure.
Our framework is structured into four interrelated sub-modules: (1) \textbf{MiFleet-TNC operation module:} characterizes the operational and economic behavior of ride-hailing companies that manage both AV and HV fleets. (2) \textbf{Traveler choice module:} represents travelers’ decisions in response to service attributes such as travel distance-based, travel time-based, and waiting time cost. (3) \textbf{% Truncated queue-based 
Customer waiting time module:} describes 
travelers' waiting times for e-hailing service
that are capped by travelers' option of solo drive
in lieu of the latter service.
% captures the interaction between supplies and demands of the traffic market through an abstract waiting-time function queue-based approximation that links fleet availability and customer demand to expected customer waiting time. 
(4) \textbf{Traffic congestion module:} describes the macroscopic relationship between vehicular flows and travel times, providing feedback from network conditions to both companies and travelers. Our contributions can be summarized as follows.

\gap

\noindent $\bullet $ \textbf{Modeling:} we develop a unified equilibrium model that captures interactions among MiFleet TNCs, travelers, and traffic. This framework is flexible and can be extended to accommodate multiple heterogeneous fleet types, providing a practical and forward-looking tool for future mixed autonomy transportation system. Key modeling features are as follows:

\noindent --- we explicitly differentiate AV and HV behaviors across pickup and service stages: in the pickup phase, AVs follow Wardropian routing while HVs may deviate due to individual preferences; in the service phase, AVs may deviate under coordinated company control, whereas HVs adhere to shortest-path behavior;

\gap

\noindent --- recognizing the challenge in
prescribing waiting times, we employ an abstraction of these times as continuous functions
of the model variables and illustrate one
such function by an explicit, truncated, queue-based formulation, in which the waiting duration is computed using endogenously determined travel times, dispatch rates, and demand rates;
    
\gap

\noindent $\bullet $ \textbf{Analysis:} we apply
fundamental results from complementarity 
and degree theory to establish the existence 
of an equilibrium solution for the overall traffic model under the sole requirement of continuity of the model functions.
Most importantly, our analysis removes a previous restriction on the availability of TNCs' fleets (the key Lemma~3 in \cite{ban2019general}) and the setting of fixed travel times \cite{gu5461575generalized} in which traffic
congestion is absent.

\gap

\noindent $\bullet $ \textbf{Competitive insights:} we design detailed numerical studies to uncover how AV penetration, pricing strategy, and routing strategy influence equilibrium outcomes. The studies aim to provide insights into different control levels over AVs and deliver actionable guidance for MiFleet operators and regulators on profitability, congestion, and efficiency trade-offs.
    
\gap

\noindent $\bullet $ \textbf{System framework tool:} we provide a general, computational implementable framework for practical scenario analysis, policy evaluation, and adaptive regulation in mixed autonomy transportation networks.

\gap

The remainder of this paper is organized as follows. Section~\ref{sec:related work} reviews the related literature. Section~\ref{sec:model_setting} presents the overall model setup, including the modeling foundation, notation, and the model functions. Section~\ref{sec:math_formulation} presents the mathematical formulation of the submodules. Section~\ref{sec:model_analysis} establishes the existence of an equilibrium solution. Sections~\ref{sec:benchmark_numerical_results} and~\ref{sec:case_study} report numerical experiments that evaluate the effectiveness and reliability of the proposed model. Section~\ref{sec:conclusion} concludes the paper.
\section{Related Work} \label{sec:related work}

Recent studies related to our work can be broadly categorized into three main strands: models of TNC operations, analyses of AV impacts, and multi-agent equilibrium formulations. 

\subsection{TNC Operations}
Extensive research has analyzed the role of TNCs in shaping transportation systems through pricing, fleet management, and matching mechanisms \cite{ban2019general, li2021spatial,ni2021modeling,zha2018geometric,lai2023spatiotemporal,ke2020pricing}. These studies typically represent the interactions among TNCs, travelers, and the traffic network through network equilibrium formulations, where fleet operations, traveler decisions, and congestion outcomes are jointly determined \cite{ban2019general,chen2024network,xu2021equilibrium}. Complementary to this stream, another line of research explicitly captures vehicle-customer matching using queuing-based formulations \cite{braverman2019empty, feng2022approximating,gu5461575generalized,iglesias2019bcmp}.
Most existing studies, however, have focused on homogeneous fleet structures, assuming that all vehicles operate under identical behavioral and coordination rules, with limited attention to the behavioral and operational asymmetries between AV and HV services. This operational-level simplification raises an important question: when fleets become mixed, how does automation reshape the broader transportation system? We address these questions by developing an equilibrium model including competing MiFleet TNCs that operate mixed fleets, incorporating explicit AV-HV differentiation in operational behaviors.

\subsection{Mixed-Autonomy Transportation Systems}

A growing body of literature has explored the implications of AV deployment on transportation system performance. Studies have demonstrated that AVs can improve road safety, dampen wave-and-go phenomenon, and mitigate congestion \cite{fagnant2015preparing,stern2018dissipation,zheng2020smoothing,wu2021flow}. Recent studies on coordinated fleets indicate that these vehicles have the potential to enhance overall network efficiency and lead to system optimum \cite{chen2020path,Battifarano2023TheIO}. Other works, however, suggest potential unintended consequences, including induced travel demand, increased VMT, and spatial imbalances in vehicle distributions \cite{auld2017analysis,childress2015using, castro2024autonomous,chen2024distributional}. 
These contrasting findings highlight the complexity of automation’s system-level impacts. 
While these studies yield valuable insights into local dynamics and AV coordination, they often isolate such effects from traveler demand, fleet heterogeneity, and network congestion. To address this limitation, we develop a mixed-autonomy system that distinguishes AVs from HVs while integrating these critical elements.
% At system level, mixed-autonomy studies have advanced understanding of AV-HV coexistence, spanning from analyses of highway on-ramp merging and intersection coordination to platform operating strategies.\cite{fang2024cooperative,zhu2021safety, chen2024operations}. 

% While these studies yield valuable insights into local dynamics or firm-level strategies, they often isolate these effects rather than integrating traveler demand, fleet heterogeneity, and network congestion feedbacks within one coherent system. To capture these interdependent processes, we formulate a unified network-equilibrium model representing multiple interacting agents and their equilibrium behaviors.

%While these studies provide valuable insights into specific operational mechanisms and firm-level decisions, a unified-analytical framework that integrates heterogeneous TNC fleets, traveler choices, and congestion feedbacks remain lacking. Our work fills this gap by formulating a network-equilibrium model that coherently integrates these components within a single mixed-autonomy system.

\subsection{Multi-Behavior Modeling}
%Beyond application-oriented studies, a parallel line of research focuses on the methodological foundations of multi-agent equilibrium modeling.

% Such integration draws on the methodological line of multi-agent equilibrium modeling.
Early research into heterogeneous behavioral models can be traced back to Harker's seminal work~\cite{harker1988multiple}. His model incorporates both price-making and price-taking agents within the same system, allowing different origin-destination (OD) pairs to exhibit diverse strategic behaviors under a unified network structure. 
Subsequent research has acknowledged the complexity inherent in such equilibrium formulations. To manage this, some researchers introduce specific structural restrictions \cite{yang2007stackelberg,yang2017mixed}, while others have adopted simplified assumptions, such as treating demand as exogenously fixed or limiting the structure
of competing fleets. Although some models effectively capture how different travel modes influence congestion and yield useful policy insights, they remain constrained by assuming a single fleet structure across all operators \cite{di2019unified, ban2019general}. 
Moreover, to model traffic delays, prior studies often adopt an assumption that all vehicles on the road behave as Wardrop users, making route choices independently to minimize their own travel costs \cite{ban2019general,di2019unified,xu2015complementarity,yang2011equilibrium}. However, such assumptions lead to inefficiency in a mixed-autonomy system \cite{lazar2020routing,lazar2018price}.
% More recent advancements incorporate end-to-end learning with deep neural
% networks to enhance route choice modeling and equilibrium estimation, but they still assume fixed
% demand, limiting their ability to reflect demand elasticity and dynamic trip generation \cite{liu2023end}. 
AVs may deviate under coordinated company dispatching and HVs may not follow the shortest path due to personal preferences \cite{feng2022understanding, liu2010uncovering, sirisoma2010empirical,shou2020optimal}.
A related separation of behavioral roles appears in dual-sourcing ride-hailing models, where freelance drivers behave as Wardrop players, while idle contracted drivers follow platform-directed repositioning \cite{dong2024strategic}. In a similar spirit, we model HVs as freelancers who act based on personal preferences and AVs as contractors who always follow coordinated guidance.
These behavioral differences are embedded directly into our framework and reflected across both pickup and in-service phases, allowing the heterogeneity among HVs, AVs, and solo-driving vehicles (SVs) to be fully represented.  [Remark:
We assume SVs are all human
driven and reserve the term HVs for TNCs' human fleets.  We assume that every traveler has access to a 
private car.]  
% {\color{cyan}For travelers without their own vehicles, SVs refer to their existing travel modes.}

%While these approaches provide a solid foundation for transportation equilibrium analysis, they cannot be directly applied to a mixed-autonomy systems. Our work extends this line of research by incorporating autonomous vehicles into the equilibrium framework and explicitly differentiating vehicle behaviors, including solo drivers, AVs, and HVs, within a unified mixed-autonomy setting.
\section{Problem Statement} \label{sec:model_setting}
 Representing transportation networks that comprise both autonomous and human-driven fleets entails distinct modeling and methodological challenges. 
The heterogeneity in operational coordination and behavioral responses between AVs and HVs gives rise to complex interactions among MiFleet TNCs, travelers, and the traffic environment. To ensure tractability and behavioral consistency, the framework strikes a balance between abstracting key operational features of mixed fleets and retaining a structure suitable for equilibrium analysis and computation.

\subsection{Modeling Foundations}\label{subsec:model foundation}

Each MiFleet-TNC is modeled as a self-interested agent that maximizes its own profit, determined by fare structure, cost decomposition, and vehicle deployment. MiFleet-TNCs make dispatching decisions to match vehicles with customer requests while satisfying flow balance constraints, i.e., vehicle inflows and outflows must balance at every destination node. Vehicle allocation is further subject to regulated AV shares and given total fleet size.
On the demand side, travelers decide between requesting MiFleet-TNC service or driving alone, based on perceived disutility incorporating travel time, fare, and waiting time. Travelers’ mode choices determine the demand faced by each MiFleet-TNC, which influences companies' operational decisions including vehicle dispatching.
% These interactions are mediated through customer waiting times. Some prior models approximate the vehicle–customer matching time as the marginal price of the demand satisfaction constraint~\cite{ban2019general}. In contrast, we explicitly compute the matching cost using an M/M/1 queuing formulation, extending the model in \cite{gu5461575generalized} by incorporating variable congestion effects.

Given the coordinated nature of AVs, MiFleet TNCs can strategically route AVs to avoid congestion, improve spatial fleet distribution, or comply with regulatory requirements. AVs do not necessarily follow the shortest paths when serving customers, yet they adhere to the Wardrop principle during the service phase to maintain the shortest possible customer pickup times. In contrast, HVs, behave more like traditional taxis, lack such centralized coordination. Empirical evidence suggests they often deviate from Wardrop behavior during the pickup phase \cite{liu2010uncovering, sirisoma2010empirical,shou2020optimal}. As part of the novelty of our modeling, these deviation behaviors are detailed in~\eqref{eq:traffic model}, which 
%  The modeling framework is built upon the interaction among profit-maximizing MiFleet TNCs, traveler, and traffic within the mixed-autonomy transportation system. 
% Each company aims to maximize its profit by strategically matching customers to vehicles and dispatching those vehicles across the network.
%  \textbf{AV fleets} are controlled by the company. AVs remain at their last drop-off locations until reassigned and may be detoured during service to avoid congestion, improve spatial balance, or increase revenue. They follow the Wardrop principle during pickup to ensure efficient customer matching but may not follow the shortest paths while serving customers due to coordinated company control.
% \textbf{HV fleets} retain partial autonomy. HV drivers may cruise or wait for new requests, and they may deviate from shortest paths during pickup due to individual preferences. However, when serving customers, HVs are assumed to follow the shortest path in line with Wardrop behavior.
captures congestion effects by linking path flows of SVs, AVs, and HVs to network conditions, thereby determining equilibrium travel times. Table~\ref{tab:vehicle_behavior}  provides a summary of vehicle behaviors.
\begin{table}[H]
\centering
\begin{tabular}{lcc}
\toprule
\textbf{Vehicle type} & \textbf{Pickup phase} & \textbf{Service phase} \\
\midrule
SVs & Follow Wardrop principle & Follow Wardrop principle \\[0.15cm]
AVs & Follow Wardrop principle & May deviate under company control \\[0.15cm]
HVs & May deviate & Follow Wardrop principle \\[0.15cm]
\bottomrule
\end{tabular}
\caption{Driving Behaviors across Vehicle Types and Trip Stages}
\label{tab:vehicle_behavior}
\end{table}

The principal system interactions are depicted schematically in Figure \ref{fig:system_interaction}.

\begin{figure}[H]
\centering
\includegraphics[width=5in]{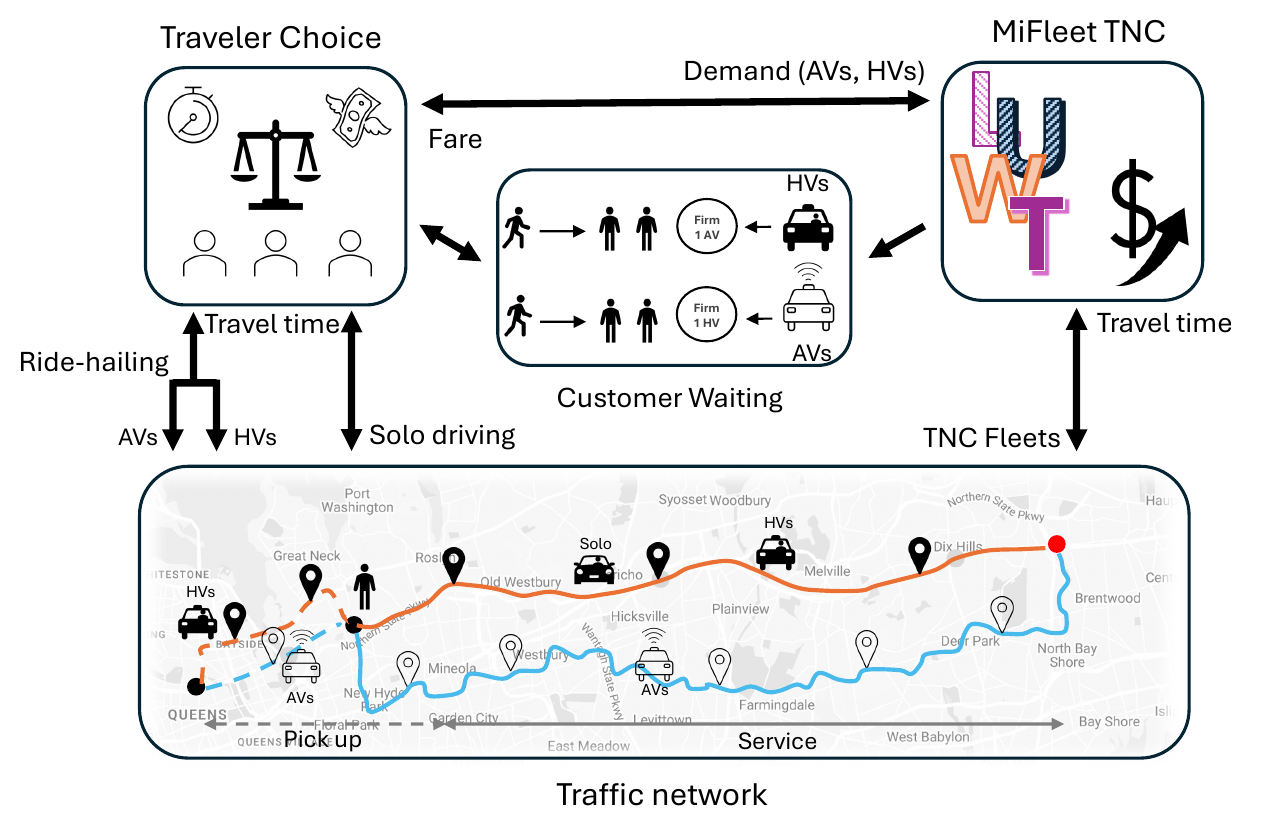}
\caption{Mixed-Autonomy System Overview} 
\label{fig:system_interaction}
\end{figure}

\subsection{Notation}
With the major foundations established, we next define the mathematical notations used throughout this paper. These notations provide a formal representation of indices, parameters, and variables.

\noindent {\bf Sets associated with 
the system} 

\vspace{1\baselineskip}

\noindent \begin{tabular}{ll} 
\(\mathcal{N}\) & Set of nodes in the network 
\\ [0.08cm] 
\(\mathcal{A}\) & Set of arcs in the network 
\\ [0.08cm] 
\(\mathcal{W}\) & Set of Origin-Destination (OD) pairs, 
a subset of \(\mathcal{N}\times\mathcal{N}\)
\\ [0.08cm] 
\(\mathcal{O}\) & Set of origin nodes, 
\(i \in \mathcal{O}\) 
if \(\exists \) some \(j\in\mathcal{N}\) 
such that \((i,j)\in {\cal W}\) \\[0.08cm] 
\(\mathcal{D}\) &  Set of destination nodes, 
\(j \in \mathcal{D}\)
if \(\exists \) some \(i\in\mathcal{N}\) such 
that \((i,j)\in {\cal W}\)\\[0.08cm] 
\(\mathcal{K}\) & Set of Traffic Network Companies (TNCs)\\[0.08cm] 
\(\mathcal{X}\) & \(\triangleq 
\{{\rm AV,HV}\}\). Set of TNC vehicle 
types \\[0.08cm] 
\({\cal W}^{k,x}\) & Set of OD pairs served
by vehicle type
\(x\in\mathcal{X}\) of company 
\(k\in\mathcal{K}\) \\[0.08cm] 
\(\mathcal{O}^{k,x}\) & Set of origin nodes served
by vehicle type
\(x\in\mathcal{X}\) of company 
\(k\in\mathcal{K}\),\\[0.08cm] 
&\(i \in \mathcal{O}^{k,x}\) if \(\exists \) 
some \(s\in\mathcal{N}\) 
such that \((i,s)\in {\cal W}^{k,x}\) 
\\ [0.08cm] 
\(\mathcal{D}^{k,x}\) & Set of destination
nodes served
by vehicle type
\(x\in\mathcal{X}\) of company 
\(k\in\mathcal{K}\),
\\[0.08cm]
&\(s \in \mathcal{D}^{k,x}\) if 
\(\exists \) some \(i\in\mathcal{N}\) 
such that \((i,s)\in {\cal W}^{k,x}\) 
\\ [0.08cm]  
\(\mathcal{K}_{ij}^{x}\) & Set of TNCs that 
provide vehicle type \(x\) to serve OD pair 
\((i,j) \in {\cal W}\) \\ [0.08cm] 
% \(\mathcal{K}_{s,i}^{x}\) & Set of TNCs that 
% provide vehicle type
% \(x\), which recently completed serving \((\ell,s)\in{\cal W}^{k,x}\)
% \\
% & for some 
% \(\ell \in \mathcal{O}^{k,x}\), to serve OD pair \((i,j) \in {\cal W}^{k,x}\) for some \(j\in\mathcal{D}^{k,x}\)\\[0.08cm] 
\(\mathcal{P}\) & Set of all paths in the 
network \\ [0.08cm] 
\(\mathcal{P}_{ij}\) & Set of paths 
connecting node \(i\in\mathcal{N}\) to 
node \(j\in\mathcal{N}\) \\
% $\mathcal{P}_{s,i}$ & Set of paths 
% connecting destination node $s
% \in \mathcal{D}$ to origin node
% $i \in \mathcal{O}$ \\
% & for pick up.
\end{tabular}

\vspace{1\baselineskip}

\noindent {\bf Model Parameters} (all 
positive):

\vspace{0.5 \baselineskip}
\noindent \begin{tabular}{ll} 
\(F_{ij}^{k,x}\) & Fixed fare charged by 
vehicle type \(x\) of 
company \(k\) serving OD pair \((i,j)\in\mathcal{W}^{k,x}\)\\ [0.08cm] 
\(\alpha_1^{k,x}\) & Travel time-based 
fare rates for vehicle type \(x\) of 
company \(k\) \\[0.08cm] 
\(\alpha_2^{k,x}\) & Travel distance-based 
fare rates for vehicle type \(x\) of 
company \(k\) \\[0.2cm]
\(\beta_1^{k,x}\) & Travel time conversion 
factor to monetary costs for vehicle type 
\(x\) of company \(k\)  \\ [0.2cm]
\(\beta_2^{k,x}\) & Travel distance conversion
factor to monetary costs for vehicle type 
\(x\) of company \(k\)  \\ [0.2cm]
\(\beta_3^{k,x}\) & Waiting time conversion 
factor to monetary costs for vehicle type 
\(x\) of company \(k\) \\ [0.08cm] 
$\mu^{k,{\rm AV}}$ & Relaxation 
factor ($\geq 1$) of the Wardrop principle 
for AV operated by company \(k\) in serving \\ [0.2cm]
& an OD pair in ${\cal W}^{k,\rm AV}$; this non-Wardropian 
behavior may be due to regulations that \\ [0.2cm]
& prohibit AVs to take certain 
routes \\ [0.08cm] 
$\mu^{k,{\rm HV}}$ & Relaxation 
factor ($\geq 1$) of the Wardrop principle 
for HV operated by company $k$ in \\ [0.2cm]
& pre-service pick up at a node pair 
in ${\cal D}^{k,\rm HV}$; this
reflects HV's permitted \\ [0.2cm]
& flexibility in their driving behavior 
in empty vehicles \\ [0.08cm] 
\(t_{ij}^0\) & Free-flow travel time 
of the shortest path (in terms of free-flow 
traveling time) \\ [0.2cm]
& from node \( i \in \mathcal{O} \) to node 
\( j \in \mathcal{D} \) \\ [0.2cm]
\end{tabular}

\noindent \begin{tabular}{ll} 
% $\mu_{s,i}^{k,{\rm AV}}$ & $= 1$; signifying
% that the Wardrop principle is obeyed by AV
% of company $k$ in \\ [0.2cm]
% & pre-service pick up at the node pair
% $(s,i) \in {\cal D} \times {\cal O}$; this
% reflects that companies \\ [0.2cm]
% & want AVs to pick up 
% customers as quickly as possible \\ [0.2cm]
% $\mu_{ij}^{k,{\rm HV}}$ & $= 1$; signifying
% that the Wardrop principle is obeyed by HV
% of company $k$ in \\ [0.2cm]
% & serving the OD pair $(i,j) \in {\cal W}$
% \\ [0.08cm] 
\(t_{s,i}^0\) & Free-flow travel time 
of the shortest path (in terms of free-flow 
traveling time) \\ [0.2cm]
& from node \( s \in \mathcal{D} \) to node 
\( i \in \mathcal{O} \) \\ [0.2cm]
% written as $t_{s,i}^{\, 0}$ for a pair 
% $(s,i) \in {\cal D} \times {\cal O}$ 
\(d_{ij}^{\, 0}\) & Free-flow travel 
distance of the shortest path (in terms of
free-flow traveling time) \\ [0.2cm]
& from node \( i \in \mathcal{N} \) to node 
\( j \in \mathcal{N} \)
% ; written as 
% $d_{s,i}^{\, 0}$ for a pair 
% $(s,i) \in {\cal D} \times {\cal O}$ 
\\[0.08cm] 
\(\mu^{\rm cap}_{\rm AV}\)  &  Maximum allowed fraction of 
AVs in TNC's fleet \\ [0.15cm] 
% \(\mu^{\rm cap}_{\rm HV}\) & 
% $= 1 - \mu^{\rm cap}_{\rm AV}$:
% fraction of HV vehicles in TNC's fleet 
% \\ [0.15cm]
\(N^k\)  &  Total fleet size of company 
\(k\) \\ [0.08cm] 
\(\gamma_1^{k,x}\) & Travel time conversion 
factor to monetary cost 
for travelers requesting \\ [0.08cm] 
& vehicle \(x\) from company \(k\)\\[0.2cm]
\(\gamma_2^{k,x}\) & Waiting time conversion factor
to monetary cost for 
travelers requesting \\ [0.08cm] 
& vehicle \(x\) form company \(k\)\\[0.08cm] 
% $\gamma_3^{k,x}$ & A multiplicative (proportional) factor in the normalized equilibrium corresponding to \\ [0.2cm]
% &
% the shared fleet-demand constraint of vehicle type \(x\) of company \(k\)
% \\ [0.2cm]
\(\alpha_{1}^{\rm SV}\) & Travel time 
conversion factor to monetary cost for solo driving vehicles\\ [0.2cm]
\(\alpha_{2}^{\rm SV}\) & Travel distance 
conversion factor to monetary cost for 
solo driving vehicles\\ [0.2cm]
\(D_{ij}\) & Total demand rate of OD 
pair \((i,j)\in\mathcal{W}\)
\end{tabular}

\vspace{1\baselineskip}
\iffalse
\begin{remark}
    We could allow
$\mu^{k,{\rm AV}}$ to depend on the OD pair 
$(i,j) \in {\cal W}^{k,{\rm AV}}$,  
$\mu^{k,{\rm HV}}$ to depend on the pair 
$(s,i) \in {\cal D}^{k,{\rm HV}} \times {\cal O}^{k,{\rm HV}}$,
and $\mu^{\rm cap}_{\rm AV}$ 
% and $\mu^{\rm cap}_{\rm HV}$ 
to depend on company $k$.  
Our model and its analysis can easily
accommodate these dependencies, albeit at the
expense of complicating the notations and 
analysis.  Interestingly, the positivity
of the TNCs' fleet sizes $N^k$ plays an
important role in the proof of the main
existence 
Theorem~\ref{th:equivalence and existence}.
Practically, a company with zero fleet size
can be dropped from consideration; this is
reflected in the fleet capacity constraint 
in the 
TNC module; see (\ref{eq:TNC module}) in
conjunction with 
(\ref{eq:fleet demand equality}).  Thus
the positivity of $N^k$ is justified 
practically and supported by the
mathematical model.
\end{remark}
\fi

\noindent {\bf Remark:}  We could allow
$\mu^{k,{\rm AV}}$ to depend on the OD pair 
$(i,j) \in {\cal W}^{k,{\rm AV}}$,  
$\mu^{k,{\rm HV}}$ to depend on the pair 
$(s,i) \in {\cal D}^{k,{\rm HV}} \times {\cal O}^{k,{\rm HV}}$,
and $\mu^{\rm cap}_{\rm AV}$ 
% and $\mu^{\rm cap}_{\rm HV}$ 
to depend on company $k$.  
Our model and its analysis can easily
accommodate these dependencies, albeit at the
expense of complicating the notations and 
analysis.  Interestingly, the positivity
of the TNCs' fleet sizes $N^k$ plays an
important role in the proof of the main
existence 
Theorem~\ref{th:equivalence and existence}.
Practically, a company with zero fleet size
can be dropped from consideration; this is
reflected in the fleet capacity constraint 
in the 
TNC module; see (\ref{eq:TNC module}) in
conjunction with 
(\ref{eq:fleet demand equality}).  Thus
the positivity of $N^k$ is justified 
practically and supported by the
mathematical model.  \hfill $\Box$

% \vspace{1\baselineskip}

\noindent {\bf Primary model variables} 

\vspace{1\baselineskip}

{
\noindent \begin{tabular}{ll}
\(z_{s,ij}^{k,x}\) & Dispatch rate of vehicle 
type \(x\) of 
company \(k\), originating at destination  \(s \in \mathcal{D}^{k,x}\),  
\\ [0.2cm]
&
assigned to serve OD pair 
\((i,j)\in {\cal W}^{k,x}\) \\ [0.08cm] 
\(D_{ij}^{\rm SV}\) & Demand rate of OD pair 
\((i,j)\) with solo driving \\ [0.15cm] 
\(D_{ij}^{k,x}\) & Demand rate of OD pair 
\((i,j)\) requesting 
vehicle type \(x\) from company \(k\) 
\\ [0.15cm] 
\(h_p^{\rm SV}\) & Traffic flow of solo 
driving vehicles on path 
\(p \in \mathcal{P}_{ij}\), for OD pair 
\((i,j)\in {\cal W}\) \\ [0.2cm] 
\(h_p^{k,x}\) & Traffic flow of 
vehicle \(x\) of company \(k\) on path
\(p \in \mathcal{P}_{ij}\), 
for OD pair \((i,j)\in {\cal W}^{k,x}\)
\\ [0.2cm] 
\(h_p^{k,x}\) & Traffic flow of vehicle 
\(x\) of company \(k\) on path 
\(p \in \mathcal{P}_{si}\), for 
\(i\in\mathcal{O}^{k,x}\) and 
\(s\in\mathcal{D}^{k,x}\) \\[0.2cm] 
\(t_{ij}^{\rm SV}\) & Congestion dependent 
travel time of solo 
driving vehicles from \(i\) to \(j\), where 
\((i,j)\in {\cal W}\) \\ [0.2cm] 
\(t_{ij}^{k,x}\) & Congestion dependent travel 
time of vehicle type \(x\) of company \(k\) 
from \(i\) to \(j\)  \\ [0.1in]
& where \((i,j)\in {\cal W}^{k,x}\) \\ [0.08cm] 
\(t_{s,i}^{k,x}\) & Congestion dependent 
travel time of vehicle type \(x\) of company 
\(k\) from \(s\) to \(i\) \\ [0.2cm]
& where \(i\in\mathcal{O}^{k,x}\) and 
\(s\in\mathcal{D}^{k,x}\)
\end{tabular}
}

% \vspace{1\baselineskip}

\noindent {\bf Induced model variables}

\vspace{1\baselineskip}

\noindent \begin{tabular}{ll}
\(R_{s,ij}^{k,x}\) & Per trip revenue for vehicle \(x\) of company \(k\) serving OD pair \((i,j)\in\mathcal{W}^{k,x}\) from \(s\in\mathcal{D}^{k,x}\)\\ [0.15cm] 
\(w_{ij}^{k,x}\) & Waiting time for customers 
requesting
vehicle \(x\) from company \(k\) \\ [0.15cm] 
\( \widehat{w}_{ij}^{k,x}\) & Waiting time for 
TNC $k$'s vehicle type $x$ currently at node 
$s \in {\cal D}^{k,x}$ to \\ [0.2cm]
& serve OD pair $(i,j)\in\mathcal{W}^{k,x}$ \\ [0.08cm] 
\(\phi_{s}^{k,x}\) & Shadow price of the 
flow conservation constraint \\ [0.2cm]
\( \lambda_{ij}^{k,x}\) & Marginal 
price of OD demand \((i,j)\) 
of vehicle type \(x\), perceived by 
company \(k\), \\ [0.2cm]
\( \widehat{\lambda}_{ij}^{k,x}\) & Marginal price
of OD demand \((i,j)\) of vehicle type 
\(x\) of company \(k\), \\ [0.2cm]
& perceived by customer; assumed to be 
proportional to 
\(\lambda_{ij}^{k,x}\) \\ [0.2cm]
$\nu^k_{\rm AV}$ & Marginal price of company 
$k$'s AV capacity  \\ [0.2cm]
$\nu^k$ & Marginal price of company 
$k$'s fleet capacity \\ [0.2cm]
\(\sigma_{ij}\) & Shadow price of the total 
demand satisfaction constraint
\end{tabular}

\vspace{0.5\baselineskip}
\noindent {\bf Auxiliary model variables}

\vspace{1\baselineskip}

\noindent \begin{tabular}{ll}
\(\theta_{s,ij}^{k,x}\), 
\(\zeta_{s,ij}^{k,x}\) & Artificial 
variables employed to handle the ambiguity 
of the undefined \\ [0.2cm]
& fraction 0/0 in customers’ waiting costs.
\end{tabular}

\vspace{1\baselineskip}

\noindent In the model formulation, to be presented
momentarily, we let $\boldsymbol{h}$ be the
tuple consisting of all the path flows:
\[ \begin{array}{l}
\left\{ \, h_p^{\rm SV} \, : \,
p \in {\cal P}_{ij} \, , \, (i,j)\in\mathcal{W} \, \right\} \ 
\mbox{for the solo vehicles}, 
\\ [0.1in]
\left\{ \, h_p^{k,x} \, : \, 
(k,x) \in {\cal K} \times {\cal X}, \, 
(i,j) \in {\cal W}^{k,x},
\, 
p\in\mathcal{P}_{ij}\, \right\}\
 \epc \mbox{and}
\\ [0.1in]
\left\{ \, h_p^{k,x} \, : \, 
(k,x) \in {\cal K} \times {\cal X}, \, 
(s,i) \in {\cal D}^{k,x}\times\mathcal{O}^{k,x},
\,
p\in\mathcal{P}_{si}
\, \right\} \mbox{for the TNC's vehicles}.
\end{array} \]
We also let $C_p( \boldsymbol{h} )$
be the cost on path $p$, which we assume is
continuous and 
satisfies the natural condition:
\[
C_p( \boldsymbol{h}) \, \geq \, C_p(0) 
\, \geq \, 0, \epc \forall \,
p \, \in \, {\cal P} \ \mbox{ and all } 
\ \boldsymbol{h} \geq 0.
\]
Additionally, these path costs are required 
to satisfy three weak positivity conditions
(see \cite[Proposition~1.4.6]{facchinei2003finite}
for background):

\gap

\noindent $\bullet $ for all OD pairs
$(i,j) \in {\cal W}$:
\begin{equation} 
\label{eq:path cost conditions I}
\left[ \, \displaystyle{
\sum_{p \in {\cal P}_{ij}}
} \, h_p^{\rm SV} \, C_p( \boldsymbol{h} ) 
\, = \, 0; \, 
h_p^{\rm SV} \, \geq \, 0 \ \forall \, p \in 
{\cal P}_{ij}\, \right] \ \Rightarrow \ 
h_p^{\rm SV} \, = \, 0 \ \forall \, 
p \, \in \, {\cal P}_{ ij};
\end{equation}
\noindent $\bullet $ for all 
$(k,x) \in {\cal K} \times {\cal X}$, 
and all $(i,j) \in {\cal W}^{k,x}$:
\begin{equation} 
\label{eq:path cost conditions II}
\left[ \, \displaystyle{
\sum_{p \in {\cal P}_{ij}}
} \, h_p^{k,{\rm x}} C_p( \boldsymbol{h} ) 
\, = \, 0; \, 
h_p^{k,{\rm x}} \, \geq \, 0, \ \forall \, 
p \, \in \, {\cal P}_{ij} \, \right] 
\ \Rightarrow \ 
h_p^{k,{\rm AV}} \, = \, 0 \ \forall \, 
p \, \in \, {\cal P}_{ij}; 
\end{equation}
$\bullet $ for all $(k,x) \in {\cal K} \times
{\cal X}$ and all $(s,i) \in 
{\cal D}^{k,x} \times {\cal O}^{k,x}$:
\begin{equation} 
\label{eq:path cost conditions III}
\left[ \, \displaystyle{
\sum_{p \in {\cal P}_{si}}
} \, h_p^{k,x} \, C_p( \boldsymbol{h} ) 
\, = \, 0; \, 
h_p^{k,x} \, \geq \, 0, \ \forall \, 
p \, \in \, {\cal P}_{si} \, \right] 
\ \Rightarrow \ 
h_p^{k,x} \, = \, 0 \ \forall \, 
p \, \in \, {\cal P}_{si}.
\end{equation}
These conditions are trivially satisfied when the 
path cost
functions are positive.  The model and its
analysis do not assume that the path costs
are necessarily derived from an additive model
of the link costs; in particular, the BPR link
cost functions are not needed.  Lastly, 
we define the free-flow path costs:
\[ \begin{array}{l}
t_{ij}^{\, 0} \, \triangleq \, \displaystyle{
\min_{p \in {\cal P}_{ij}}
} \, C_p(0), \epc \forall \, (i,j)  \in {\cal W}
 \ \mbox{ and } p \, \in \, {\cal P}_{ij}
; \\ [0.2in]
t_{s,i}^{\, 0} \, \triangleq \, \displaystyle{
\min_{p \in {\cal P}_{si}}
} \, C_p(0), \epc \forall \, (s,i) \, \in \,
{\cal D}^{k,x} \times {\cal O}^{k,x} \text{ for some }k \in\mathcal{K},\, x\in\mathcal{X}.
\end{array} \]
\section{Mathematical Formulation}\label{sec:math_formulation}
Four interrelated submodules constitute the model: the profit-making TNCs
operating the AVs and HVs, the traveler decision model, the customer waiting model, and traffic conditions. Their interactions are depicted in Figure~\ref{fig:submodule_interaction}.

\begin{figure}[H]
\centering
\includegraphics[width=5in]{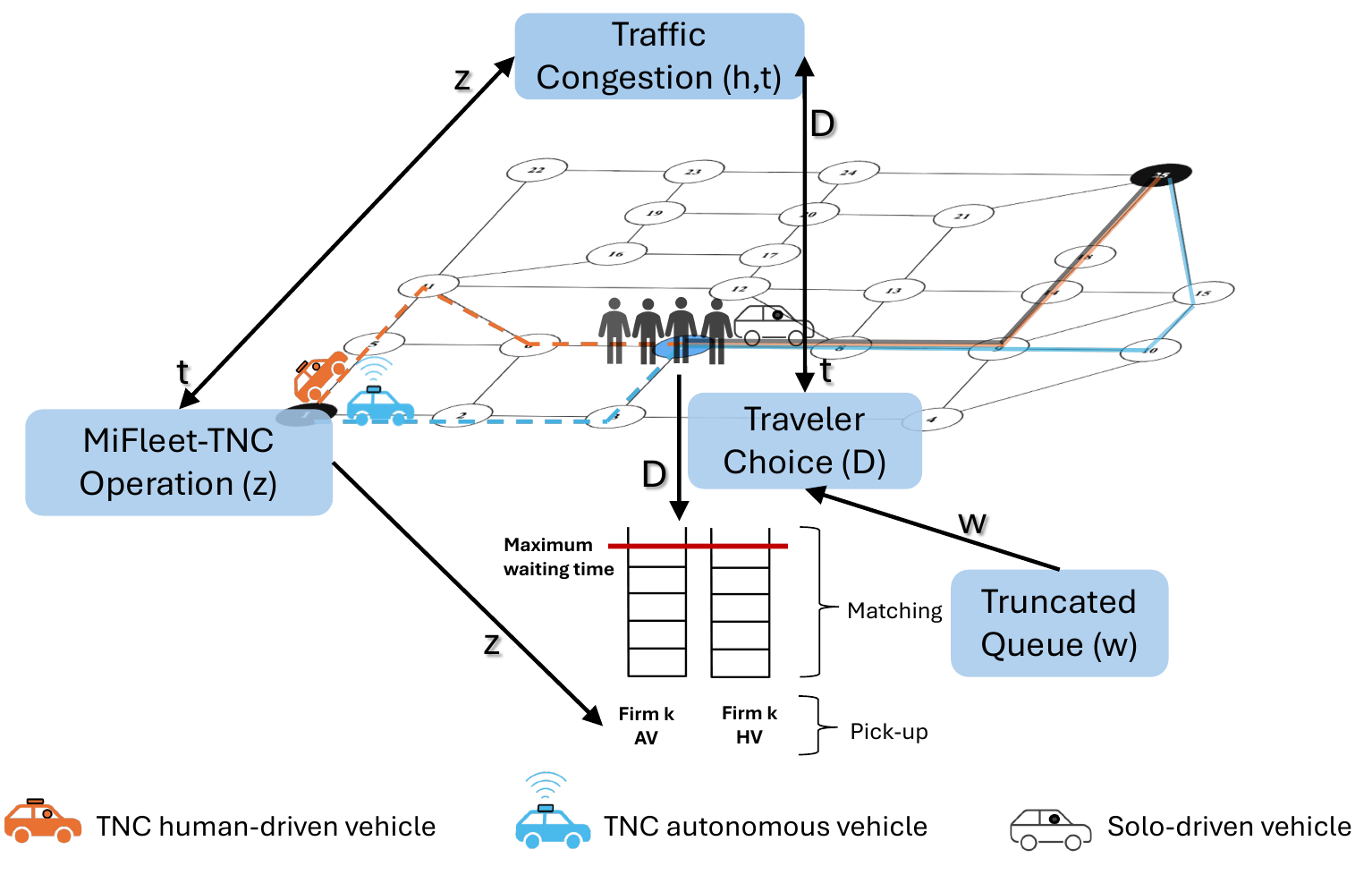}
\caption{System Interactions} 
\label{fig:submodule_interaction}
\end{figure}

{
\begin{center}
\fbox{
\parbox{6in}{
\noindent {\bf The MiFleet TNC operational module.} 
(See \cite{ban2019general} for more details.)
The profit of TNC $k \in {\cal K}$ derived
from vehicle type $x \in {\cal X}$ that 
is currently located at node $s \in {\cal D}^{k,x}$ 
and assigned to serve OD pair 
$(i,j) \in {\cal W}^{k,x}$ is given by
$R_{s,ij}^{k,x}$, which is equal to 
\[ \begin{array}{l}
% R_{s,ij}^{k,x} \triangleq 
F_{ij}^{k,x} +
\underbrace{\alpha_2^{k,x} \, 
d_{ij}^{\, 0}}_{\mbox{
dist.\ based revenue}} +
% \, + \\ [0.35in]
% \hspace{0.2in} 
\underbrace{\alpha_1^{k,x} \,
( t_{ij}^{k,x} - t_{ij}^{\, 0} )}_{
\mbox{time-based revenue}} -
\underbrace{\beta_1^{k,x} \, ( t_{ij}^{k,x} 
+ t_{s,i}^{k,x} )}_{\mbox{time-based cost}} - 
\underbrace{\beta_2^{k,x} \, 
( d_{ij}^{\, 0} + 
d_{si}^{\, 0} )}_{\mbox{dist.\ based cost}}
\\[0.3in] 
= \, \underbrace{F_{ij}^{k,x} - 
\alpha_1^{k,x} 
t_{ij}^{\, 0} + \alpha_2^{k,x} d_{ij}^{\, 0}  
- \beta_2^{k,x} ( d_{ij}^{k,x} + d_{si}^{\, 0} 
)}_{\mbox{fixed part, denoted 
$\wt{R}_{s,ij}^{k,x}$}} 
\, + \underbrace{\alpha_1^{k,x} t_{ij}^{k,x}
- \beta_{1}^{k,x} ( t_{s,i}^{k,x} + 
t_{ij}^{k,x} )}_{\mbox{variable part}}.
\end{array}
\]
The TNC vehicles' waiting times satisfy 
the following balancing equation:
\[
\displaystyle{
\sum_{(i,j) \in {\cal W}^{k,x}}
} \, \left[ \, \displaystyle{
\sum_{s \in \mathcal{D}^{k,x}}
} \, z^{k,x}_{sij} \, \right]
\widehat{w}_{ij}^{k,x} \, = \, 
N^{k,x} - \displaystyle{
\sum_{(i,j) \in {\cal W}^{k,x}}
} \, \displaystyle{
\sum_{s\in\mathcal{D}^{k,x}}
} \, z^{k,x}_{s,ij} t_{s,i}^{k,x} - 
\displaystyle{
\sum_{(i,j)\in {\cal W}^{k,x}}
} \, D^{k,x}_{ij}t_{ij}^{k,x}.
\]
Substituting the right-hand side for the left-hand 
side into the objective function,
TNC $k$'s optimization problem is the 
following profit-maximization linear program:
\begin{equation} \label{eq:TNC module}
\left\{ \begin{array}{l}
\displaystyle{
\operatornamewithlimits{
\mbox{\bf maximize}}_{z_{s,ij}^{k,{\rm AV}}, 
\, z_{s,ij}^{k,{\rm HV}}}
} \\
\textcolor{blue}{\mbox{\begin{tabular}{l}
company \\
profit
\end{tabular}}} \ \displaystyle{  
\sum_{x \in \mathcal{X}}
} \, \left\{ \, \displaystyle{
\sum_{(i,j) \in {\cal W}^{k,x}}
} \, \displaystyle{
\sum_{s \in \mathcal{D}^{k,x}}  
} \, \left( \, 
\underbrace{R_{s,ij}^{k,x}}_{
\text{revenue}} - \underbrace{\beta_{3}^{k,x} 
\widehat{w}_{ij}^{k,x}}_{\text{
monetary waiting cost}} \, \right) 
z_{s,ij}^{k,x} \, \right\} \\ [0.3in]
\hspace{0.2in} = \, \displaystyle{  
\sum_{x \in \mathcal{X}}
} \, \displaystyle{
\sum_{(i,j) \in {\cal W}^{k,x}}
} \, \displaystyle{
\sum_{s \in \mathcal{D}^{k,x}}  
} \, \left[ \, \wt{R}_{s,ij}^{\, k,x} +
\alpha_1^{k,x} t_{ij}^{k,x} - 
\beta_{1}^{k,x} ( t_{s,i}^{k,x} + 
t_{ij}^{k,x} ) \, \right] z_{s,ij}^{k,x}
\, + \\ [0.3in]
\hspace{0.3in} \beta_3^{k,x} \, 
\displaystyle{  
\sum_{x \in \mathcal{X}}
} \,  \displaystyle{
\sum_{(i,j) \in {\cal W}^{k,x}}
} \, \displaystyle{
\sum_{s \in \mathcal{D}^{k,x}}  
} \, z^{k,x}_{s,ij} t_{s,i}^{k,x} +
\mbox{constant exogeneous to module} 
\\ [0.2in]
\mbox{\bf subject to}: \\ [0.2in]
\text{\textcolor{blue}{\begin{tabular}{l}
flow \\
conservation
\end{tabular}}} \ \displaystyle{
\sum_{(i,j)\in {\cal W}^{k,x}}
} \, z_{s,ij}^{k,x} \, = \, \displaystyle{
\sum_{i\in\mathcal{O}^{k,x}\mid (i,s)\in\mathcal{W}^{k,x}}
} \, D_{is}^{k,x}, \quad x \in \mathcal{X}, \ 
s \in \mathcal{D}^{k,x} \\ 
\text{\textcolor{blue}{\begin{tabular}{l}
fleet-demand \\
constraint
\end{tabular}}} \epc \displaystyle{
\sum_{s \in \mathcal{D}^{k,x}}
} \, z_{s,ij}^{k,x} \, \geq \, D_{ij}^{k,x}, 
\quad  x \in \mathcal{X}, \ 
(i,j) \in {\cal W}^{k,x} \\ [0.2in]
\text{\textcolor{blue}{AV capacity}} 
\ \displaystyle{
\sum_{(i,j)\in {\cal W}^{k,{\rm AV}}}
} \, \left\{ \, \displaystyle{
\sum_{s\in\mathcal{D}^{k,{\rm AV}}} 
} \, t_{s,i}^{k,{\rm AV}}
z^{k,{\rm AV}}_{s,ij} + 
t_{ij}^{k,{\rm AV}}D^{k,{\rm AV}}_{ij} 
\, \right\} \, \leq \, \mu^{\rm cap}_{\rm AV} 
\, N^k \\ 
\text{\textcolor{blue}{\begin{tabular}{l}
fleet \\
capacity
\end{tabular}
}} \ \underbrace{\displaystyle{
\sum_{x \in {\cal X}}
} \ \displaystyle{
\sum_{(i,j)\in {\cal W}^{k,x}}
} \ \displaystyle{
\sum_{s\in\mathcal{D}^{k,x}}
} \, t_{s,i}^{k,x} z^{k,x}_{s,ij}}_{
\begin{tabular}{l}
vehicles \\
en route to service calls
\end{tabular}} \, + \underbrace{\displaystyle{
\sum_{x \in {\cal X}}
} \, \displaystyle{
\sum_{(i,j)\in {\cal W}^{k,x}}
} \, t_{ij}^{k,x} D^{k,x}_{ij}}_{
\begin{tabular}{l}
vehicles \\
serving travel demands
\end{tabular}} \leq N^k \\ [0.4in]
% t_{ij}^{k,x} D^{k,x}_{ij}} \, \right\} 
\text{\textcolor{blue}{nonnegativity}} \
z_{s,ij}^{k,x} \geq 0, \epc \mbox{for } \, 
x \in {\cal X}, \, s \in {\cal D}^{k,x}, \,
(i,j) \in {\cal W}^{k,x}
\end{array} \right\}
\end{equation}
}}
\end{center}
}

{
\begin{center}
\fbox{
\parbox{6in}{
\noindent {\bf The traveler choice module.}
(Similar to \cite{ban2019general})
Each traveler chooses a single travel mode per 
trip over all travel modes that
include solo driving, HV or AV service
from a TNC. For each OD pair 
$(i,j) \in {\cal W}$, let \(V_{ij}^{k,x}\) 
be the traveler's 
disutility for choosing vehicle type \(x\) 
of company 
\(k\) and let \(V_{ij}^0\) be solo 
driver's disutility 
traveling from \(i\) to \(j\).  We have
\[ \begin{array}{l}
V_{ij}^{k,x} \, \triangleq \, F_{ij}^{k,x} + 
\underbrace{\alpha_{1}^{k,x} 
( t_{ij}^{k,x} - t_{ij}^0 )}_{\mbox{
\begin{tabular}{l}
travel time \\
based disutility
\end{tabular}
}} + \underbrace{\alpha_{2}^{k,x} d_{ij}^{\, 0}}
_{\mbox{\begin{tabular}{l}
distance  \\
disutility
\end{tabular}
}} + \underbrace{\gamma_{1}^{k,x} 
t_{ij}^{k,x}}_{\mbox{\begin{tabular}{l}
travel time \\
disutility
\end{tabular}
}} + \underbrace{
\gamma_{2}^{k,x} w_{ij}^{k,x}}_{\mbox{
\begin{tabular}{l} 
waiting disutility \\
to be picked up
\end{tabular}
}} \\ [0.6in]
V^{\rm SV}_{ij} \, \triangleq \, 
\underbrace{\alpha_1^{SV}t_{ij}^{SV}}_{\mbox{
\begin{tabular}{l}
travel time \\
based disutility
\end{tabular}
}} + 
\underbrace{\alpha^{\rm SV}_2 d_{ij}^{\, 0}}_{
\mbox{\begin{tabular}{l}
travel distance \\
based disutility
\end{tabular}
}}.
\end{array} \]
Note that unlike $V_{ij}^{k,x}$ whose sign
is not predetermined, the solo driver's
disutility $V_{ij}^{\rm SD}$ is always 
positive.
Taking the waiting times $w_{ij}^{k,x}$ as 
exogenous variables (see the customer-waiting
module), the traveler choice optimization 
problem is the following disutility 
minimization linear program:
\begin{equation} \label{eq:traveler choice}
\left\{ \begin{array}{l}
\displaystyle{
\operatornamewithlimits{
\mbox{\bf minimize}}_{D_{ij}^{k,x}, \
\, D_{ij}^{\rm SV}}
} \\
\textcolor{blue}{\mbox{traveler disutility}} 
\ \displaystyle{
\sum_{(i,j)\in {\cal W}}
} \, \left( V_{ij}^{\rm SV} D_{ij}^{\rm SV} 
+ \displaystyle{ 
\sum_{x\in\mathcal{X}}
} \, \displaystyle{
\sum_{k\in \mathcal{K}^x_{ij}} 
} \, V_{ij}^{k,x} D_{ij}^{k,x} \right) 
\\ [0.2in]
\mbox{\bf subject to} \\ [0.2in]
\text{\textcolor{blue}{\begin{tabular}{l}
Total demand \\
satisfaction:
\end{tabular}}} \epc 
D_{ij}^{\rm SV} + \displaystyle{
\sum_{x\in\mathcal{X}}
} \, \displaystyle{
\sum_{k\in \mathcal{K}^x_{ij}}
} \ D_{ij}^{k,x} \, = \, D_{ij}, 
\quad \forall \,
(i,j) \in {\cal W}  \\ [0.3in]
\text{\textcolor{blue}{\begin{tabular}{l}
Fleet dictated \\
demand constraint
\end{tabular}}} \ \displaystyle{  
\sum_{s \in \mathcal{D}^{k,x}} 
} \, z_{s,ij}^{k,x} \, \geq \, D_{ij}^{k,x}, 
\quad \forall \,
k \in \mathcal{K}, \, x \, \in\mathcal{X}, 
\, (i, j) \in {\cal W}^{k,x} \\ [0.3in]
\text{\textcolor{blue}{Nonnegativity:}} 
\hspace{0.3in}
\begin{array}{l}
D_{ij}^{\rm SV} \, \geq \, 0, \quad \forall \,(i, j) \in {\cal W}
\\[0.1in]
D_{ij}^{k,x} 
\, \geq \, 0, \quad \forall \,
k \in \mathcal{K}, \, x \, \in\mathcal{X}, 
\, (i, j) \in {\cal W}^{k,x}
\end{array}
\end{array} \right\} 
\end{equation}
}}
\end{center}
}

Pertaining to an optimal solution of the
problem (\ref{eq:traveler choice}), the
following lemma helps to establish a 
reasonable upper bound for customer waiting 
time, serving as a preparatory result for 
the formula to be followed.

\begin{lemma}\label{lemma:demand shifting} \rm
Suppose that $V_{ij}^{k,x} > V_{ij}^{\rm SV}$ 
for some OD pair $(i,j) \in \mathcal{W}$
and company-vehicle type pair 
$(k,x) \in \mathcal{K} \times \mathcal{X}$,
then there exists an optimal solution to the 
traveler choice 
problem (\ref{eq:traveler choice})
satisfying \(D_{ij}^{k,x} = 0\).
\end{lemma}

\begin{proof}  It suffices to note that if
$D_{ij}^{k,x} > 0$ in an optimal 
solution, then shifting \(D_{ij}^{k,x}\) to 
\(D_{ij}^{\rm SV}\) preserves feasibility 
while strictly improving the objective value. 
\end{proof}

\noindent {\bf Customer waiting.}  Similar to the 
previous approach \cite{ban2019general}, we model
the customer waiting for TNCs' service to be
composed of two components, company's matching
(or dispatching)
plus the mean pickup time, both as perceived 
by the customer.  In the reference,
the former matching time is described by a multiple of
the marginal price of the fleet demand constraint:
$\displaystyle{
\sum_{s\in\mathcal{D}^{k,x}}
} \, z_{s,ij}^{k,x} \geq D_{ij}^{k,x}$.  Since
the multiplier is complementary to the slack of
this constraint, an issue
with this definition of matching is that when
equality holds, the matching time is not well
defined.  An alternative definition of the matching
time is given in~\cite{gu5461575generalized} under
the assumption of fixed travel times that facilitated
the employment of 
a queuing model.  Our matching time below extends
this previous model to incorporate congestion. 
To capture realistic customer behavior under high 
demand scenarios, we incorporate a truncation to the 
queue-based waiting time to reflect limited customer 
patience.

% {
% \begin{center}
% \fbox{
% \parbox{6in}{
% \noindent {\bf Customer waiting.}
In general, we may model the
matching time as an extended-value function of the
tuple $\{ z_{ij}^{k,x},D_{ij}^{k,x} \}$ of 
company's vehicle allocation and OD demands.
An example of this function is that derived from 
a queuing model such as an M/M/1 queue.  According to 
the this model, the customer waiting for company 
$k$'s vehicle type $x$ 
traveling between OD pair 
\((i,j)\in {\cal W}^{k,x}\) is considered 
as a queue, in particular an M/M/1 model, 
where vehicles are 
the servers and passengers are customers.  
Company $k$'s AVs and HVs
arrive exponentially with mean 
$\left( \displaystyle{
\sum_{s\in\mathcal{D}^{k,x}}z_{{sij}}^{k,x}
} \right)^{-1}$ for \(x \in\mathcal{X}\).  
The interarrival times of customers 
requesting AVs and HVs from company 
\(k\) follow an exponential 
distribution with mean 
\(1/D_{ij}^{k,{\rm AV}}\) and 
\(1/D_{ij}^{k,{\rm HV}}\), respectively.  
Mean matching time is 
modeled as the mean waiting time in the 
queuing system.  Thus we have 
the queuing-based steady-state matching time:
\begin{equation} \label{eq:dispatching}
\displaystyle{
\frac{1}{\displaystyle{
\sum_{s\in\mathcal{D}^{k,x}}
} \, z_{s,ij}^{k,x} - D_{ij}^{k,x}}
}, \quad \forall \, k\in\mathcal{K}, \, 
x \in \mathcal{X}, \, 
(i,j) \in {\cal W}^{k,x},
\end{equation}
where the denominator is nonnegative due 
to the fleet-demand 
constraint; it can equal to zero, leading 
to a infinite matching
time that will be resolved by capping.
To formalize the capping, we note
since the travelers have the option of not waiting
for a TNC service, it is reasonable to postulate
that there is a maximum time
beyond which the traveler will decide not to wait, and 
therefore resort to solo driving.  This maximum
time is in turn derived from the travel choice
module based on the smaller of the disutility 
of solo driving and TNC service.  Based on this
consideration, Lemma~\ref{lemma:demand shifting}
therefore suggests the maximum waiting time to be
\[ \begin{array}{l}
w_{ij}^{\max} \, \triangleq \\ [0.1in] 
\max\left\{ \, 0, \, \displaystyle{
\frac{\alpha_1^{\rm SV}t_{ij}^{\rm SV} 
+ \alpha^{\rm SV}_2 d_{ij}^0 - \displaystyle{
\min_{k\in\mathcal{K}, x\in\mathcal{X}}
} \, \left\{ \, F_{ij}^{k,x} + 
\alpha_{1}^{k,x} (t_{ij}^{k,x}- t_{ij}^0) 
+ \alpha_{2}^{k,x} d_{ij}^{\, 0} +
\gamma_{1}^kt_{ij}^{k,x} \, 
\right\}}{\displaystyle{
\min_{k\in\mathcal{K}, x\in\mathcal{X}}
} \, \gamma_{2}^{k,x}
}} \, \right\},
\end{array} \]
which is a measure of customer's patience.
The mean pickup time is equal to the product:
\[
\left( \, \displaystyle{
\sum_{s\in\mathcal{D}^{k,x}}
} \, t_{s,i}^{k,x} \, \right) \, \left( \, 
\displaystyle{
\frac{z_{s,ij}^{k,x}}{\displaystyle{
\sum_{s^{\prime} \in\mathcal{D}^{k,x}}
} \, z_{s{\prime},ij}^{k,x}}
} \, \right), \quad \forall \, 
k \in \mathcal{K}, \, x \in \mathcal{X}, 
\, (i,j) \in {\cal W}^{k,x}.
\]
Summarizing the above derivations, we therefore
arrive at the waiting time definition:
\begin{equation} \label{eq:customer waiting}
w_{ij}^{k,x} \, \triangleq \, 
\min\left\{ \, w_{ij}^{\max}, \, \wt{w}_{ij}^{k,x}
+ \left( \, \displaystyle{
\sum_{s\in\mathcal{D}^{k,x}}
} \, t_{s,i}^{k,x} \, \right) \, 
\left( \, \displaystyle{
\frac{z_{s,ij}^{k,x}}{\displaystyle{
\sum_{s^{\prime} \in \mathcal{D}^{k,x}}
} \, z_{s^{\prime},ij}^{k,x}} 
} \, \right) \, \right\},
\end{equation}
where $\wt{w}_{ij}^{k,x}$ is company $k$'s 
matching/dispatching time, an example of which is
the (extended-valued) queuing-based 
steady-state matching time (\ref{eq:dispatching}).
The expression (\ref{eq:customer waiting})
essentially stipulates that a traveler will choose 
to drive solo if the disutility due to 
waiting to be picked up
exceeds the disutility of solo driving.  
% }}
% \end{center}
% }

In order to address the
ambiguity of the fraction $0/0$, we
define $\theta_{s,ij}^{k,x}$ 
% \triangleq \displaystyle
% \frac{z_{s,ij}^{k,x}}{\displaystyle{
% \sum_{s^{\prime} \in \mathcal{D}^{k,x}}
% } \, z_{s^{\prime}ij}^{k,x}}
as a minimizer of 
\begin{equation}\label{eq:customer_waiting}
\displaystyle{
\operatornamewithlimits{\mbox{minimize}}_{
\theta \in [0,1]}
} \, \left\{ -z_{s,ij}^{k,x} \theta + 
\thalf \, \left\{ \displaystyle{
\sum_{s^{\prime} \in \mathcal{D}^{k,x}}
} \, z_{s^{\prime},ij}^{k,x} \right\} 
\theta^2 \right\},
\end{equation} which
is equal to the fraction $\displaystyle{
\frac{z_{s,ij}^{k,x}}{\displaystyle{
\sum_{s^{\prime} \in \mathcal{D}^{k,x}}
} \, z_{s^{\prime},ij}^{k,x}}
}$ when the 
denominator is positive and
is an arbitrary scalar between 0 and 1 
if the denominator is zero. 
In the latter case, by the demand constraint 
in the traveler choice model, 
\(D_{ij}^{k,x} = 0\) thus the total waiting 
cost of the 
OD pair \((i,j)\in {\cal W}^{k,x}\) of vehicle 
type \(x\) of company \(k\) is zero.  

\begin{center}
\fbox{
\parbox{6in}{
\noindent {\bf The customer waiting time module.}
In the rest of the paper, we take the customer's
waiting time to be
\begin{equation} \label{eq:capped customer waiting}
w_{ij}^{k,x} \, \triangleq \, 
\min\left\{ \, w_{ij}^{\max}, \, \wt{w}_{ij}^{k,x}
+ \left( \, \displaystyle{
\sum_{s\in\mathcal{D}^{k,x}}
} \, t_{s,i}^{k,x} \, \right) \theta_{s,ij}^{k,x}
\, \right\}
\end{equation}
where $\wt{w}_{ij}^{k,x}$ is a nonnegative,
possibly extended-valued function
of the tuple $\{ z_{ij}^{k,x},D_{ij}^{k,x} \}$
that is continuous on its domain of finiteness,
and $\theta_{s,ij}^{k,x}$ is an optimal solution
of (\ref{eq:customer_waiting}).
For the purpose of analysis, the explicit form of
the waiting time function is not important, it is
the continuity of the function in the arguments
$\{ \, z_{ij}^{k,x},D_{ij}^{k,x}, t_{s,i}^{k,x},
\theta_{s, ij}^{k,x} \, \}$ that is needed. 
Nevertheless, for the purpose of computation,
an explicit form of the matching time
$\wt{w}_{ij}^{k,x}$, such as 
(\ref{eq:dispatching}), is needed.
}}
\end{center}

{
\begin{center}
\fbox{
\parbox{6in}{
\noindent {\bf The traffic congestion module.} 
This is 
an extension of the basic Wardrop model; 
the one most distinctive
feature of this extension is that the 
Wardrop's shortest-path principle
is imposed for the AVs when they travel 
to serve the next OD demand but is relaxed 
when they are serving the demand; this
stipulation is reversed for the HVs.  
{\small
\begin{equation}\label{eq:traffic model}
\begin{array}{lll}
\displaystyle{
\sum_{p \in \mathcal{P}_{ij}}
} \, h_p^{\rm SV} & = & D_{ij} - 
\displaystyle{
\sum_{x\in\mathcal{X}}
} \, \displaystyle{
\sum_{k\in\mathcal{K}_{ij}^x}
} \, D_{ij}^{k,x} \, 
( \, = \, D_{ij}^{\rm SV} \, ), \epc 
t_{ij}^{\rm SV} \, \geq \, 0, 
\quad  \forall \, (i,j) \in {\cal W} 
\\ [0.25in]
\displaystyle{
\sum_{p \in\mathcal{P}_{ij}}
} \, h_p^{k,x} & = & D_{ij}^{k,x}, \epc
t_{ij}^{k,x} \, \geq \, 0, \hspace{0.7in} 
\forall \, x \in \mathcal{X}, \, 
k \in \mathcal{K}, 
\, (i,j) \in {\cal W}^{k,x} \\ [0.25in]
\displaystyle{
\sum_{p \in\mathcal{P}_{si}}
} \, h_p^{k,x} & = & \displaystyle{
\sum_{j\in\mathcal{D}^{k,x}}
} \, z_{s,ij}^{k,x}, \epc 
t_{s,i}^{k,x} \, \geq \, 0,
\quad \forall \, 
x \in \mathcal{X}, \, k \in \mathcal{K}, \, 
(s,i) \in \mathcal{D}^{k,x} \times 
\mathcal{O}^{k,x} \\ [0.25in]
0 \leq h_p^{\rm SV} & \perp & 
C_p( \boldsymbol{h} ) - 
t_{ij}^{\rm SV} \geq 0, \hspace{0.5in} 
\forall \, (i,j) \in {\cal W}, \,
\forall \, p \in\mathcal{P}_{ij} \\ [0.1in]
0 \leq h_p^{k,{\rm AV}} & \perp & 
\mu^{k,{\rm AV}} C_p( \boldsymbol{h} ) - 
t_{ij}^{k,{\rm AV}} \geq 0, \hspace{0.3in} 
\forall \, k \in \mathcal{K},\,(i,j) \in {\cal W}^{k,\rm AV}, \, 
p \in \mathcal{P}_{ij}
\\ [0.1in]
0 \, \leq \, h_p^{k,{\rm AV}} & \perp & 
C_p( \boldsymbol{h} ) - 
t_{s,i}^{k,{\rm AV}} \geq 0, \hspace{0.5in}
\forall \, k \in \mathcal{K},(s,i) \in 
\mathcal{D}^{k,\rm AV} \times \mathcal{O}^{k,\rm AV}, 
 p \in\mathcal{P}_{si} \\ [0.1in]
0 \leq h_p^{k,{\rm HV}} & \perp & 
C_p( \boldsymbol{h} ) - t_{ij}^{k,{\rm HV}} 
\, \geq \, 0, \hspace{0.5in} 
\forall \, k \in \mathcal{K},\,(i,j) \in {\cal W}^{k,\rm HV}, \, 
p \in \mathcal{P}_{ij}
\\ [0.1in]
0 \leq h_p^{k,{\rm HV}} & \perp & 
\mu^{k,{\rm HV}} C_p( \boldsymbol{h} ) - 
t_{s,i}^{k,{\rm HV}} \geq 0, \epc
\forall \, k \in \mathcal{K},
(s,i) \in \mathcal{D}^{k,\rm HV} \times 
\mathcal{O}^{k,\rm HV}, 
p \in\mathcal{P}_{si}
\end{array}
\end{equation}}
}}
\end{center}
}

Overall, the mixed-fleet traffic 
equilibrium problem with the coexistence 
of HV, AV, and SV 
is to solve for the primary model variables 
satisfying the conditions in the 4 modules
introduced above.  
The problem can be studied as a 
noncooperative game where the primary agents 
(the TNCs and the travelers) are competing 
for the use of 
the traffic network subject to the Wardrop 
principle (with extensions to allow for their 
relaxations for the TNCs' fleets) and the 
customer waiting times for pick-up vehicles.
The game is of the generalized type as the
fleet-demand constraints appear in both 
the TNC and customer modules and the flow
conservation constraint is of the coupled but
un-shared type.

\gap

\noindent {\bf Variations of the model} 
(per discussion with Yueyue Fan).  To avoid the 
functional approach 
to the modeling of customer waiting times, which
invariably requires some behavioral assumption of
customers' disutility to waiting, one may treat these
waiting times as endogenous variables that induce a 
set of travel demands as a result of waiting.  One then
seeks a set of waiting times so that the induced
travel demands match the given fixed demands. 
Alternatively, one may also postulate that there is
a given demand function of customer waiting and employ
the inverse of this function as the customer waiting
time.  These alternatives are analogous to the 
fixed or elastic models in the 
basic traffic equilibrium problem;
see \cite[Subsection~1.4.5]{facchinei2003finite}.
Details of these variations are not considered in
the rest of the paper.

\section{Existence of
a generalized equilibrium} \label{sec:model_analysis}
Our goal is to establish
the existence of a generalized equilibrium~\ref{defn:mage-cq} of 
the mixed-fleet transportation system
by examining the details and interactions 
of its four submodules.
% In this section, we analyze the integrated system by examining the interaction of its four submodules:
%  \textbf{MiFleet-TNC operation module:} each MiFleet TNC maximize its overall profit given by \eqref{company model: obejctive}, optimality conditions in \eqref{OptCond: Company model} (details in Section \ref{sec:company}); \textbf{Traveler choice module:} travelers minimize the aggregate disutility given by \eqref{traveler model: obejctive}, optimality conditions in \eqref{OptCond: traveler model};
%  \textbf{Truncated queue-based customer waiting module:} customer waiting times (for matching and pickup) are modeled as functions of vehicle allocation and demand, with \(w_{ij}^{k,x}\) plugged into the traveler problem's objective in the integrated system (Section \ref{sec:customer waiting});
% \textbf{Traffic congestion module:} represents travel times and congestion effects across the network as in \eqref{OptCond/comp: traffic model} (Section \ref{sec:traffic}).
\begin{framed}
    \begin{definition}[MAGE-CW]\label{defn:mage-cq}
    \rm
    An MAGE-CW equilibrium is a tuple \(\{\mathbf{z}, \mathbf{D}, \boldsymbol{\theta}, \mathbf{h}, \mathbf{t}\}\) such that 
    \begin{enumerate}
        \item For each \(k\in\mathcal{K}\), the tuple \(\mathbf{z}\triangleq\{z_{s,ij}^{k,x}: x\in\mathcal{X}, s\in\mathcal{D}^{k,x}, (i,j)\in\mathcal{W}^{k,x}\}\) solves Problem \eqref{eq:TNC module} for given 
        \(\{\mathbf{D}, \boldsymbol{\theta}, \mathbf{h}, \mathbf{t}\}\);
        \item
        Tuple \(\mathbf{D}\triangleq\{D_{ij}^{k,x}, D_{ij}^{\rm SV}: (i,j)\in\mathcal{W},x\in\mathcal{X}, k\in\mathcal{K}_{ij}^{x}\}\) solves Problem \eqref{eq:traveler choice} for given \(\{\mathbf{z}, \boldsymbol{\theta}, \mathbf{h}, \mathbf{t}\}\), where each $w_{ij}^{k,x}$ 
        is a continuous function of
        $\{ \, z_{ij}^{k,x},D_{ij}^{k,x}, t_{s,i}^{k,x},\theta_{s, ij}^{k,x} \, \}$
        such as that given by 
        (\ref{eq:capped customer waiting});
        \item 
        Tuple \(\boldsymbol{\theta}\triangleq\{\theta_{s,ij}^{k,x}: x\in\mathcal{X}, s\in\mathcal{D}^{k,x}, (i,j)\in\mathcal{W}^{k,x}\}\) solves \eqref{eq:customer_waiting} for given \(\{\mathbf{z}, \mathbf{D}, \mathbf{h}, \mathbf{t}\}\);
        \item 
        Tuple \(\mathbf{h}\triangleq\{h_{p}^{\rm SV}, t_{ij}^{\rm SV}, h_{p}^{k,x}, t_{ij}^{k,x}, t_{s,i}^{k,x}:(i,j)\in\mathcal{W}, x\in\mathcal{X}, k\in\mathcal{K}_{ij}^{x}, (s,i)\in \mathcal{D}^{k,x} \times \mathcal{O}^{k,x}, 
        \mbox{ and } p \in \mathcal{P}_{ij} \cup \mathcal{P}_{si}\}\) satisfies \eqref{eq:traffic model} for given \(\{\mathbf{z}, \mathbf{D}\}\).
    \end{enumerate}
    \end{definition}  
\end{framed}

\noindent Similar to \cite{ban2019general}, 
our analysis 
addresses a special kind of equilibrium
solution known as a {\sl normalized equilibrium}
in which the multipliers of the
shared fleet-demand constraint:
$\displaystyle{  
\sum_{s \in \mathcal{D}^{k,x}} 
} \, z_{s,ij}^{k,x} \, \geq \, D_{ij}^{k,x}$ 
in the TNC and the traveler choice modules
are postulated to be proportional.  This follows Rosen’s \cite{rosen1965existence} classical framework for generalized Nash games with shared constraints.   In what follows, 
we further specialize the normalized equilibrium
to a {\sl variational equilibrium} in which 
the said multiplier are equal.  This specialization
simplifies the notation without affecting the 
analysis and leads to a nonlinear complementarity 
problem (NCP) formulation of the model. The NCP
is then shown to be equivalent a variational 
inequality (VI), to which we apply a 
fundamental theorem---Proposition~\ref{pr:existence VI for traffic}---to 
prove the existence of a variational equilibrium, 
thus the existence of an MAGE-CQ.
Two preliminary steps are needed for this
purpose: one, we derive several lemmas
pertaining to the equality constraints in
the three modules: MiFleet TNC, traveler, and 
traffic congestion, based on
which we write down the said NCP formulation for 
this game in which all equalities are replaced by
their respective inequalities.
Next, we establish some bounds on the primary
variables of the model, which allow us
to obtain the desired variational inequality 
formulation for the
problem where some complementarity conditions
are not explicitly stated. Throughout the 
analysis, we note that the customer waiting
time, including the special case of
(\ref{eq:customer waiting})
is a continuous function of the model's 
primary variables $z_{s,ij}^{k,x}$, $D_{ij}^{k,x}$
and $t_{s,i}^{k,x}$.  As it turns out, 
the denominator in the above expression is
always (in particular, at equilibrium)
equal to zero when all vehicle-trips
are balanced
(see Lemma~\ref{lm:TNC inequalities});
rendering $w_{ij}^{k,x} = w_{ij}^{\max}$ 
and obscuring the effect of matching.

\gap

\noindent {\bf First step--two conversion 
lemmas and
the complementarity formulation:}  Similar
to the basic traffic equilibrium model, (see \cite[Proposition~1.4.6]{facchinei2003finite}), 
we prove a key lemma that shows that some equalities in the MAGE-CQ model
are equivalent to inequalities; in particular, the equations in the
traffic congestion module 
(\ref{eq:traffic model}) are equivalent
to complementarity conditions.  

\begin{lemma} \label{lm:remaining equalities} \rm
Suppose the path cost functions are
nonnegative and satisfy the three conditions:
(\ref{eq:path cost conditions I}),
(\ref{eq:path cost conditions II}), and 
(\ref{eq:path cost conditions III}).  The
following two statements hold.

\gap

\noindent (A) If
$D_{ij}^{\rm SV} + \displaystyle{
\sum_{x\in\mathcal{X}}
} \, \displaystyle{
\sum_{k\in\mathcal{K}_{ij}^x}
} \, D_{ij}^{k,x} \geq D_{ij} > 0$ for all 
$(i,j) \in {\cal W}$, all
$D_{ij}^{\rm SV} \geq 0$, 
$D_{ij}^{k,x} \geq 0$, then the conditions in 
(\ref{eq:traffic model})
are equivalent to 
{\small
\begin{equation}
\label{eq:modified traffic model}
\begin{array}{l}
0 \, \leq \, t_{ij}^{\rm SV} \, \perp \,
\displaystyle{
\sum_{p \in \mathcal{P}_{ij}}
} \, h_p^{\rm SV} - D_{ij}^{\rm SV} 
\, \geq \, 0, \hspace{0.8in}  
\forall \, (i,j) \in {\cal W} \\ [0.25in]
0 \, \leq \, t_{ij}^{k,x} \, \perp \,
\displaystyle{
\sum_{p \in\mathcal{P}_{ij}}
} \, h_p^{k,x} - D_{ij}^{k,x} \, \geq \, 0,
\hspace{0.4in} \forall \, 
x \in \mathcal{X}, \, k \in \mathcal{K}, 
\, (i,j) \in {\cal W}^{k,x} \\ [0.25in]
0 \, \leq \, t_{s,i}^{k,x} \, \perp \,
\displaystyle{
\sum_{p \in\mathcal{P}_{si}}
} \, h_p^{k,x} - \displaystyle{
\sum_{j\in\mathcal{D}^{k,x} |
(i,j)\in\mathcal{W}^{k,x}
} \, z_{s,ij}^{k,x} \, \geq \, 0, 
\quad \forall \, x \in \mathcal{X}, \, 
k \in \mathcal{K}, \, 
(s,i) \in \mathcal{D}^{k,x} \times 
\mathcal{O}^{k,x}} \\ [0.25in]
0 \leq h_p^{\rm SV} \, \perp \, 
C_p( \boldsymbol{h} ) - t_{ij}^{\rm SV} \geq 0, 
\hspace{0.8in} \forall 
\, (i,j) \in {\cal W}, \,
\forall \, p \in\mathcal{P}_{ij} \\ [0.1in]
0 \leq h_p^{k,{\rm AV}} \, \perp \, 
\mu^{k,{\rm AV}} C_p(\boldsymbol{h}) - 
t_{ij}^{k,{\rm AV}} \geq 0,
\quad \forall \, (i,j) \in {\cal W}, \, 
p \in \mathcal{P}_{ij}, \, k\in\mathcal{K} 
\\ [0.1in]
0 \, \leq \, h_p^{k,{\rm AV}} \, \perp \, 
C_p( \boldsymbol{h} ) - 
t_{s,i}^{k,{\rm AV}} \geq 0, \quad \forall \, 
(s,i) \in \mathcal{D}\times\mathcal{O}, 
\, p \in\mathcal{P}_{si}, 
\, k \in \mathcal{K} \\ [0.1in]
0 \leq h_p^{k,{\rm HV}} \, \perp \, 
C_p( \boldsymbol{h} ) - t_{ij}^{k,{\rm HV}} 
\, \geq \, 0, \hspace{0.2in} 
\forall \, (i,j) \in {\cal W}, \, 
p \in \mathcal{P}_{ij}, \, k\in\mathcal{K} 
\\ [0.1in]
0 \leq h_p^{k,{\rm HV}} \, \perp \, 
\mu^{k,{\rm HV}} C_p( \boldsymbol{h} ) - 
t_{s,i}^{k,{\rm HV}} \geq 0, \quad \forall \, 
(s,i) \in \mathcal{D}\times\mathcal{O}, 
\, p \in\mathcal{P}_{si}, \, 
k \in \mathcal{K}.
\end{array}
\end{equation}
}

\noindent (B) Conversely, under the conditions 
in (\ref{eq:modified traffic model}), the 
optimization problem 
(\ref{eq:traveler choice}) in the traveler
choice module is equivalent 
to one in which the demand satisfaction 
equality is formulated as an inequality:
\begin{equation} \label{eq:demand inequality}
D_{ij}^{\rm SV} + \displaystyle{
\sum_{x\in\mathcal{X}}
} \, \displaystyle{
\sum_{k\in \mathcal{K}^x_{ij}}
} \ D_{ij}^{k,x} \, \geq \, D_{ij}.
\end{equation}
\end{lemma}

\begin{proof} (A) The implication
(\ref{eq:traffic model}) $\Rightarrow$ 
(\ref{eq:modified traffic model}) is obvious.
To prove the converse
implication: (\ref{eq:modified traffic model})
$\Rightarrow$ (\ref{eq:traffic model}).  we 
need to show that any solution to 
(\ref{eq:modified traffic model}) will satisfy
\[ \begin{array}{l}
\displaystyle{
\sum_{p \in \mathcal{P}_{ij}}
} \, h_p^{\rm SV} \, = \, D_{ij}^{\rm SV}, 
\ \displaystyle{
\sum_{p \in\mathcal{P}_{ij}}
} \, h_p^{k,x} \, = \, D_{ij}^{k,x}, 
\ \mbox{ and } \ \displaystyle{
\sum_{p \in\mathcal{P}_{si}}
} \, h_p^{k,x} \, = \, \displaystyle
\sum_{j\in\mathcal{D}^{k,x} | (i,j)\in\mathcal{W}^{k,x}
} \, z_{s,ij}^{k,x}.
\end{array} \]
We prove only the first equality.  The proof of
the other two equalities is similar.
Suppose otherwise; then $\displaystyle{
\sum_{p \in \mathcal{P}_{ij}}
} \, h_p^{\rm SV} > D_{ij}^{\rm SV}$ for
some $(i,j)$ in ${\cal W}$.  By
complementarity, it follows that
$t_{ij}^{\rm SV} = 0$.  We then have
$\displaystyle{
\sum_{p \in {\cal P}_{ij}}
} \, h_p^{\rm SV} C_p( \boldsymbol{h} ) = 0$.  
By the
assumption (\ref{eq:path cost conditions I}),
we deduce $h_p^{\rm SV} = 0$ for all 
$p \in {\cal P}_{ij}$; this implies
$D_{ij}^{\rm SV} < 0$, which is a 
contradiction.

\gap

\noindent (B)  It suffices to show that an 
optimal solution of the 
problem (\ref{eq:traveler choice}) with the
demand inequality 
(\ref{eq:demand inequality}) must satisfy
this inequality as an equality.  Assume by
way of contradiction that this is false;
i.e., $D_{ij}^{\rm SV} + \displaystyle{
\sum_{x\in\mathcal{X}}
} \, \displaystyle{
\sum_{k\in \mathcal{K}^x_{ij}}
} \ D_{ij}^{k,x} > D_{ij} > 0$.  Then either
$D_{ij}^{\rm SV} > 0$ for some 
$(i,j) \in {\cal W}$ or $D_{ij}^{k,x} > 0$
for some $(k,x) \in {\cal K} \times {\cal X}$
and some $(i,j)$ in ${\cal W}^{k,x}$.  Assume
the former.  Then, by optimality to the
problem (\ref{eq:traveler choice}), 
we must have 
$V_{ij}^{\rm SV} \leq 0$, which contradicts
the positivity of this quantity by its 
definition.  Assume now that $D_{ij}^{k,x} > 0$
for some $(k,x) \in {\cal K} \times {\cal X}$
and some $(i,j) \in {\cal W}^{k,x}$.  Then 
similarly, we must have $V_{ij}^{k,x} \leq 0$,
which implies $t_{ij}^{k,x} < t_{ij}^0$.  By
the definition of the latter free-flow time,
we have $t_{ij}^{k,x} < C_p(\boldsymbol{h})$ for
all $p \in {\cal W}^{k,x}$.  This implies
by complementarity that $h_p^{k,x} = 0$
for all such $p$.  But this contradicts
the inequality 
$\displaystyle{
\sum_{p \in\mathcal{P}_{ij}}
} \, h_p^{k,x} \geq D_{ij}^{k,x} > 0$.
\end{proof}

Next is the flow conservation equality in
the TNC module. The lemma below also shows 
that
there is an equivalent formulation of the
other constraints in (\ref{eq:TNC module})
in which the travel demands $D_{ij}^{k,x}$
in the fleet capacity constraint,
which are exogenous to this module, can be
replaced by the module's primary variables
$z_{s,ij}^{k,x}$.  This substitution
turns out to be
an important maneuver for our existence 
proof of an equilibrium solution to the 
mixed-fleet transportation system.  

\begin{lemma} \label{lm:TNC inequalities} \rm
Given \(D_{ij}^{k,x}\geq 0 \), for any TNC $k$,
the constraints of the optimization problem 
(\ref{eq:TNC module}) are equivalent 
to the following:
\begin{equation} 
\label{eq:TNC constraints inequality}
\left. \begin{array}{l}
\text{\textcolor{blue}{\begin{tabular}{l}
revised flow \\
conservation
\end{tabular}}} \ \displaystyle{
\sum_{(i,j)\in {\cal W}^{k,x}}
} \, z_{s,ij}^{k,x} \, \leq \, \displaystyle{
\sum_{i\in\mathcal{O}^{k,x}|(i,s)\in\mathcal{W}^{k,x}}
} \, D_{is}^{k,x}, \quad x \in \mathcal{X}, \ 
s \in \mathcal{D}^{k,x} \\ [0.2in]
\text{\textcolor{blue}{\begin{tabular}{l}
fleet-demand \\
constraint
\end{tabular}}} \epc \displaystyle{
\sum_{s \in \mathcal{D}^{k,x}}
} \, z_{s,ij}^{k,x} \, \geq \, D_{ij}^{k,x}, 
\quad  x \in \mathcal{X}, \ 
(i,j) \in {\cal W}^{k,x} \\ [0.25in]
\text{\textcolor{blue}{AV capacity}}
\ \displaystyle{
\sum_{(i,j)\in {\cal W}^{k,{\rm AV}}}
} \ \displaystyle{
\sum_{s\in\mathcal{D}^{k,{\rm AV}}} 
} \, ( \, t_{s,i}^{k,{\rm AV}} +
t_{ij}^{k,{\rm AV}} \, ) \,
z^{k,{\rm AV}}_{s,ij} 
% + t_{ij}^{k,{\rm AV}}D^{k,{\rm AV}}_{ij} 
% \right\} 
\leq \mu^{\rm cap}_{\rm AV} N^k \\ [0.25in]
\text{\textcolor{blue}{\begin{tabular}{l}
fleet \\
capacity
\end{tabular}
}} \ \displaystyle{
\sum_{x \in {\cal X}}
} \ \displaystyle{
\sum_{(i,j)\in {\cal W}^{k,x}}
} \  \displaystyle{
\sum_{s\in\mathcal{D}^{k,x}}
} \, ( t_{s,i}^{k,x} + t_{ij}^{k,x} ) \,
z^{k,x}_{s,ij} 
% + t_{ij}^{k,x} D^{k,x}_{ij}} 
\, \leq \, N^k \\ [0.3in]
\text{\textcolor{blue}{nonnegativity}} \epc 
z_{s,ij}^{k,x} \, \geq \, 0, \quad 
\text{for }x \in \mathcal{X}, \ 
s \in \mathcal{D}^{k,x}, \ 
(i,j) \in {\cal W}^{k,x}
\end{array} \right\}.
\end{equation}
Moreover under either (\ref{eq:TNC module}) 
or (\ref{eq:TNC constraints inequality}),
it must hold that 
\begin{equation} 
\label{eq:fleet demand equality}
\displaystyle{
\sum_{s \in {\cal D}^{k,x}}
} \, z_{s,ij}^{k,x} = D_{ij}^{k,x}, \epc
\forall \, x \in {\cal X}, \,  
(i,j) \in {\cal W}^{k,x}.
\end{equation}
\end{lemma}

\begin{proof}  
% For the first assertion about
% the equivalence between 
(\ref{eq:TNC constraints inequality})
$\Rightarrow$ (\ref{eq:TNC module}).
% it suffices to show that the flow 
% conservation
% inequalities must hold as equalities under
% (\ref{eq:TNC constraints inequality}).  
On one hand, summing up the revised flow conservation 
inequalities 
over all $s \in {\cal D}^{k,x}$ yields:
\begin{equation} 
\label{eq:sum flow constraints}
\displaystyle{
\sum_{s \in {\cal D}^{k,x}}
} \, \displaystyle{
\sum_{(i,j)\in {\cal W}^{k,x}}
} \, z_{s,ij}^{k,x} \, \leq \, \displaystyle{
\sum_{s \in {\cal D}^{k,x}}
} \, \displaystyle{
\sum_{i\in\mathcal{O}^{k,x}|(i,s)\in\mathcal{W}^{k,x}}} \, D_{is}^{k,x}.
\end{equation}
On the other hand, 
summing up the fleet constraint over all 
$(i,j) \in {\cal W}^{k,x}$, we have
\begin{equation} 
\label{eq:sum fleet constraints}
\displaystyle{
\sum_{(i,j) \in {\cal W}^{k,x}}
} \, \displaystyle{
\sum_{s \in \mathcal{D}^{k,x}}
} \, z_{s,ij}^{k,x} \, \geq \, 
\displaystyle{
\sum_{(i,j) \in {\cal W}^{k,x}}
} \, D_{ij}^{k,x} \, = \,  \displaystyle{
\sum_{s \in {\cal D}^{k,x}}
} \, \displaystyle{
\sum_{i\in\mathcal{O}^{k,x}|(i,s)\in\mathcal{W}^{k,x}}
} \, D_{is}^{k,x}.
\end{equation}
% where the equality holds because the demands
% $D_{ij}^{k,x}$ over all the OD pairs $(i,j)$
% must be served by some vehicles
% of the same mode that come from some previous
% destinations; i.e., $\displaystyle{
% \sum_{(i,j) \in {\cal W}^{k,x}}
% } \, = \, \displaystyle{
% \sum_{s \in {\cal D}^{k,x}}
% } \, \displaystyle{
% \sum_{i \in {\cal O}^{k,x}}
% }$. 
Combining (\ref{eq:sum flow constraints}) 
and (\ref{eq:sum fleet constraints}), it
follows that the flow conservation and
fleet-demand inequalities in 
(\ref{eq:TNC constraints inequality}) must
hold as equalities. In particular,
(\ref{eq:fleet demand equality}) holds.
Thus, the AV capacity and fleet capacity 
in (\ref{eq:TNC module}) both hold.

\gap

\noindent (\ref{eq:TNC module}) $\Rightarrow$
(\ref{eq:TNC constraints inequality}).  This
can be proved similarly.
\end{proof}

Next, we introduce multipliers for
all the constraints in the MAGE-CQ model and
write down the optimality conditions for each 
of the optimization problems in the modules.
By the above two lemmas, all the constraints
can be formulated as inequalities; thus
their multipliers are all nonnegative.  Most
importantly, we invoke the postulate
of proportional multipliers (denoted 
\textcolor{red}{$\lambda_{ij}^{k,x}$})
for the shared fleet-demand constraint:
$\displaystyle{  
\sum_{s \in \mathcal{D}^{k,x}} 
} \, z_{s,ij}^{k,x} \, \geq \, D_{ij}^{k,x}$
to define the variational equilibrium.  We then
concatenate all the optimality
conditions into a large-scale 
nonlinear complementarity problem, which we
denote as (NCP)$_{\rm main}$:
% \[
% 0 \, \leq \, \boldsymbol{x} \, \perp \, 
% \Phi(\boldsymbol{x}) \, \geq \, 0,
% \]
% where the vector $\mathbf{x}$ contains all 
% the model
% variables and the constraint multipliers and
% $\Phi(\boldsymbol{x})$
% is a corresponding vector function 
% that define the inequality constraints, 
% Specifically, the resulting NCP is:
{\small
\[ \begin{array}{l}
0 \, \leq \, z_{s,ij}^{k,{\rm AV}} 
\, \perp \, - \wt{R}_{s,ij}^{\, k,{\rm AV}} -
\alpha_1^{k,x} t_{ij}^{k,{\rm AV}} + 
\beta_{1}^{k,x} ( t_{s,i}^{k,{\rm AV}} + 
t_{ij}^{k,{\rm AV}} ) - 
\beta_3^{k,{\rm AV}}
t_{s,i}^{k,{\rm AV}}\, + \\ [0.1in]
\epc \phi_s^{k,{\rm AV}} - 
\textcolor{red}{
\lambda_{ij}^{k,{\rm AV}}} + 
( t_{s,i}^{k,{\rm AV}} + 
t_{ij}^{k,{\rm AV}} ) \,
( \nu^k + \nu^k_{\rm AV} )  
\geq 0, \epc \forall \, k \in \mathcal{K}, \, 
s \in {\cal D}^{k,{\rm AV}}, 
\, (i,j) \in \mathcal{W}^{k,{\rm AV}}  
\\ [0.1in]
0 \, \leq \, z_{s,ij}^{k,{\rm HV}} 
\, \perp \, - \wt{R}_{s,ij}^{\, k,{\rm HV}} -
\alpha_1^{k,x} t_{ij}^{k,{\rm HV}} + 
\beta_{1}^{k,x} ( t_{s,i}^{k,{\rm HV}} + 
t_{ij}^{k,{\rm HV}} ) - 
\beta_3^{k,{\rm HV}}
t_{s,i}^{k,{\rm HV}} + \\ [0.1in]
\epc \phi_s^{k,{\rm HV}} - 
\textcolor{red}{
\lambda_{ij}^{k,{\rm HV}}} + 
( t_{s,i}^{k,{\rm HV}} + 
t_{ij}^{k,{\rm HV}} ) \, \nu^k \geq 0,
\epc \forall \, k \in \mathcal{K}, \, 
s \in {\cal D}^{k,{\rm HV}}, 
\, (i,j) \in \mathcal{W}^{k,{\rm HV}}  
\\[0.1in]
0 \, \leq \, \phi_s^{k,x} \, \perp \, 
\displaystyle{
\sum_{i\in\mathcal{O}^{k,x}|(i,s)\in\mathcal{W}^{k,x}}
} \, D_{is}^{k,x} - \displaystyle{
\sum_{(i,j)\in \mathcal{W}^{k,x}}
} \, z_{s,ij}^{k,x} \, \geq \, 0, \ 
\forall \, k\in\mathcal{K}, \, 
x \in \mathcal{X}, 
\, s \in \mathcal{D}^{k,x} \\ [0.3in]
0 \, \leq \, \lambda_{ij}^{k,x} \, \perp \, 
\displaystyle{
\sum_{s \in \mathcal{D}^{k,x}}
} \, z_{s,ij}^{k,x} - D_{ij}^{k,x} 
\, \geq \, 0, \hspace{0.2in} 
\forall \, k \,\in \mathcal{K}, \, 
x \in \mathcal{X}, 
\, (i, j) \in \mathcal{W}^{k,x} \\ [0.3in]
0 \, \leq \, \nu^k \, \perp \, N^k -
\displaystyle{
\sum_{x \in {\cal X}}
} \, \displaystyle{
\sum_{(i,j)\in \mathcal{W}^{k,x}}
} \, \displaystyle{
\sum_{s\in\mathcal{D}^{k,x}}
} \, ( t_{s,i}^{k,x} + t_{ij}^{k,x} )
z^{k,x}_{s,ij} 
% + t_{ij}^{k,x} D_{ij}^{k,x} )  
\, \geq \, 0, \epc 
\forall \, k \in \mathcal{K} \\ [0.3in]
0 \, \leq \, \nu^k_{\rm AV} \, \perp \,
\mu^{\rm cap}_{\rm AV} N^k -
\displaystyle{
\sum_{(i,j)\in \mathcal{W}^{k,{\rm AV}}}
} \, \displaystyle{
\sum_{s\in\mathcal{D}^{k,{\rm AV}}}
} \, ( t_{s,i}^{k,{\rm AV}} +
t_{ij}^{k,{\rm AV}} ) z^{k,{\rm AV}}_{s,ij} 
% + t_{ij}^{k,{\rm AV}} D^{k,{\rm AV}}_{ij} 
\, \geq \, 0, \epc 
\forall \, k \in \mathcal{K} \\ [0.3in]
0 \, \leq \, D_{ij}^{\rm SV} \, \perp \, 
\alpha_1^{\rm SV}t_{ij}^{\rm SV} + 
\alpha^{\rm SV}_2 d_{ij}^{\, 0} - 
\sigma_{ij} \, \geq \, 0, \epc \forall \, 
(i,j) \in \mathcal{W} \\ [0.1in]
0 \, \leq \, D_{ij}^{k,x} \, \perp \, 
F_{ij}^{k,x} 
+ \alpha_{1}^{k,x} (t_{ij}^{k,x}- 
t_{ij}^{\, 0}) + \alpha_{2}^{k,x} 
d_{ij}^{\, 0}
+ \gamma_{1}^{k,x}t_{ij}^{k,x} +
\gamma_{2}^{k,x} w_{ij}^{k,x} \\ [0.1in]
\hspace{1in} - \, \sigma_{ij} + 
\textcolor{red}{
% \gamma_3^{k,x}
\lambda_{ij}^{k,x}} \, 
\geq \, 0, \hspace{0.3in} 
\forall \, (k, x) \in \mathcal{K} \times 
\mathcal{X}, \, (i,j) \in \mathcal{W}^{k,x} 
\\ [0.1in]
0 \, \leq \, \sigma_{ij} \, \perp \,   
D_{ij}^{\rm SV} + \displaystyle{
\sum_{k\in\mathcal{K}}
} \, \displaystyle{
\sum_{x\in\mathcal{X}}
} \, D_{ij}^{k,x} - D_{ij} \, \geq \, 0, \epc 
\forall \, (i,j) \in \mathcal{W}
\end{array} \]
}
% nonlinear complementarity problem 
% in the form 
% \begin{equation} \label{eq:MiCP AV}
% \begin{array}{l}
% 0 \, \leq \, \mathbf{x} \, \perp \,
% \Phi(\mathbf{x},\mathbf{y}) \, \geq \, 0 
% \\ [0.1in]
% \mathbf{y} \mbox{ free }, \epc
% \Psi(\mathbf{x},\mathbf{y}) \, = \, 0,
% \end{array} \end{equation} 
% and
% $\mathbf{y}$ contains all the free, i.e.,
% sign-unrestricted,
% variables and $\Psi(\mathbf{x},\mathbf{y})$
% is a corresponding function containing all
% the equality constraints.  Retaining some
% equality constraints, we obtain
% the mixed NCP formulation of the mixed-fleet
% transportation system as:
\[ \begin{array}{l}
% \end{array} \end{equation} 
% \begin{equation} 
% \label{eq:NCP mixed-fleet II}
% \begin{array}{l}
0 \, \leq \, \theta_{s,ij}^{k,x} \, \perp \, 
\theta_{s,ij}^{k,x} \displaystyle{
\sum_{s^{\prime} \in \mathcal{D}^{k,x}}
} z_{s^{\prime},ij}^{k,x} - 
z_{s,ij}^{k,x} + \zeta_{s,ij}^{k,x} 
\geq 0, \\ [0.2in] 
\hspace{1in} \forall \, k \in \mathcal{K}, \, 
x \in \mathcal{X}, \,
s \in \mathcal{D}^{k,x}, 
\, (i,j) \in \mathcal{W}^{k,x} \\[0.1in]
0 \, \leq \, \zeta_{s,ij}^{k,x} \, \perp \,
1 - \theta_{s,ij}^{k,x} \, \geq \, 0, \epc
\epc \forall \, k \in \mathcal{K}, \, 
x \in \mathcal{X}, \, 
s \in \mathcal{D}^{k,x}, 
\, (i,j) \in \mathcal{W}^{k,x} \\ [0.1in]
0 \, \leq \, t_{ij}^{\rm SV} \, \perp \,
\displaystyle{
\sum_{p \in \mathcal{P}_{ij}}
} \, h_p^{\rm SV} - D_{ij}^{\rm SV} 
\, \geq \, 0, \hspace{0.8in}  
\forall \, (i,j) \in {\cal W} \\ [0.25in]
0 \, \leq \, t_{ij}^{k,x} \, \perp \,
\displaystyle{
\sum_{p \in\mathcal{P}_{ij}}
} \, h_p^{k,x} - D_{ij}^{k,x} \, \geq \, 0,
\hspace{0.4in} \forall \, 
x \in \mathcal{X}, \, k \in \mathcal{K}, 
\, (i,j) \in {\cal W}^{k,x} \\ [0.25in]
0 \, \leq \, t_{s,i}^{k,x} \, \perp \,
\displaystyle{
\sum_{p \in\mathcal{P}_{si}}
} \, h_p^{k,x} - \displaystyle{
\sum_{j\in\mathcal{D}^{k,x}}
} \, z_{s,ij}^{k,x} \, \geq \, 0, \quad \, 
 \forall \, x \in \mathcal{X}, \, 
k \in \mathcal{K}, \, 
(s,i) \in \mathcal{D}^{k,x} \times 
\mathcal{O}^{k,x} \\ [0.25in]
0 \leq h_p^{\rm SV} \, \perp \, 
C_p( \boldsymbol{h} ) - t_{ij}^{\rm SV} 
\, \geq \, 0, 
\hspace{0.8in} \forall \, (i,j) \in {\cal W}, 
\, p \in\mathcal{P}_{ij} \\ [0.1in]
0 \leq h_p^{k,{\rm AV}} \, \perp \, 
\mu^{k,{\rm AV}} C_p(\boldsymbol{h}) - 
t_{ij}^{k,{\rm AV}} \geq 0,
\epc \forall \,  k\in\mathcal{K},\, (i,j) \in {\cal W}^{k,\rm AV}, \, 
p \in \mathcal{P}_{ij}
\\ [0.1in]
0 \, \leq \, h_p^{k,{\rm AV}} \, \perp \,  
C_p( \boldsymbol{h} ) - 
t_{s,i}^{k,{\rm AV}} \geq 0, \epc \forall \, k \in \mathcal{K},\,
(s,i) \in \mathcal{D}^{k,\rm AV}\times\mathcal{O}^{k,\rm AV}, 
\, p \in\mathcal{P}_{si} \\ [0.1in]
0 \, \leq \, h_p^{k,{\rm HV}} \, \perp \, 
C_p( \boldsymbol{h}) 
- t_{ij}^{k,{\rm HV}} \, \geq \, 0, 
\hspace{0.2in} 
\forall \, k\in\mathcal{K} ,\,(i,j) \in {\cal W}^{k,\rm HV}, \, 
p \in \mathcal{P}_{ij}
\\ [0.1in]
0 \, \leq \, h_p^{k,{\rm HV}} \, \perp \, 
\mu^{k,{\rm HV}} C_p( \boldsymbol{h}) - 
t_{s,i}^{k,{\rm HV}} \, \geq \, 0, 
\epc \forall \, k \in \mathcal{K},\,
(s,i) \in \mathcal{D}^{k,{\rm HV}}\times\mathcal{O}^{k,{\rm HV}}, 
\, p \in\mathcal{P}_{si}.
\end{array} \]
We define the
polyhedron ${\cal Z}$ consisting of 
nonnegative tuples 
$\{ z_{s,ij}^{k,x},D_{ij}^{k,x},
D_{ij}^{\rm SV} \}$ for all 
$(k,x)$ in ${\cal K} \times {\cal X}$,
$s \in {\cal D}^{k,x}$, and 
$(i,j) \in {\cal W}$ satisfying the following
flow conservation, fleet demand, and total demand constraints, respectively:
\[ \begin{array}{ll}
\displaystyle{
\sum_{(i,j)\in {\cal W}^{k,x}}
} \, z_{s,ij}^{k,x} \, \leq \, \displaystyle{
\sum_{i\in\mathcal{O}^{k,x}|(i,s)\in\mathcal{W}^{k,x}}
} \, D_{is}^{k,x}, & \forall k\in \mathcal{K},\, x \in \mathcal{X}, \ 
s \in \mathcal{D}^{k,x} \\ [0.3in]
\displaystyle{  
\sum_{s \in \mathcal{D}^{k,x}} 
} \, z_{s,ij}^{k,x} \, \geq \, D_{ij}^{k,x}, 
& \forall \,
k \in \mathcal{K}, \, x \, \in\mathcal{X}, 
\, (i, j) \in {\cal W}^{k,x} \\ [0.3in]
D_{ij}^{\rm SV} + \displaystyle{
\sum_{x\in\mathcal{X}}
} \, \displaystyle{
\sum_{k\in \mathcal{K}^x_{ij}}
} \ D_{ij}^{k,x} \, = \, D_{ij}, &
\forall \, (i,j) \in {\cal W}.
\end{array} \]
Note that the above inequalities must hold
as equalities 
for all tuples $\boldsymbol{z} \in {\cal Z}$
which are obviously bounded because of the 
last constraint.

\gap

\noindent {\bf Second step---bounds and a VI 
formulation:}  We establish upper bounds 
for the primary variables of the model that
will be used in the VI formulation; these
variables are:
\begin{equation} 
\label{eq:primary model variables}
\left\{ \, z_{s,ij}^{k,x}, \, 
D_{ij}^{k,x}, \, D_{ij}^{\rm SV}, \,
h_p^{k,x}, \, h_p^{\rm SV}, \, t_{ij}^{k,x}, 
\, t_{s,i}^{k,x}, \, t_{ij}^{\rm SV} \, 
\right\},
\end{equation}
with $\left\{ z_{s,ij}^{k,x}, 
D_{ij}^{k,x}, D_{ij}^{\rm SV} \right\}$
belonging to the polyhedron $\mathcal{Z}$,
whose elements we denote $\boldsymbol{z}$, and with 

\noindent $\left\{ h_p^{k,x}, h_p^{\rm SV}, 
t_{ij}^{k,x}, t_{s,i}^{k,x}, t_{ij}^{\rm SV}
\right\}$ satisfying 
(\ref{eq:modified traffic model}), or
equivalently, (\ref{eq:traffic model}).  
This
is clear for the former family of variables.
It is also clear for the path flow variables
because
\[
\displaystyle{
\sum_{p \in \mathcal{P}_{ij}}
} \, h_p^{\rm SV} \, = \, D_{ij}^{\rm SV},
\epc \displaystyle{
\sum_{p \in\mathcal{P}_{ij}}
} \, h_p^{k,x} \, = \, D_{ij}^{k,x},
\ \mbox{ and } \ 
\displaystyle{
\sum_{p \in {\cal P}_{si}}
} \, h_p^{k,x} \, = \, \displaystyle{
\sum_{j \in {\cal D}^{k,x}}
} \, z_{s,ij}^{k,x}.
\]
Indeed, by letting
\[
\bar{h} >\max\left\{ \displaystyle{
\max_{(i,j) \in {\cal W}}
} \, D_{ij},\, \displaystyle{
\max_{(k,x) \in {\cal K} \times {\cal X}}
} \ \left( \displaystyle{
\max_{(i,j) \in {\cal W}^{k,x}} 
} \displaystyle{
\max_{p \in {\cal P}_{ij}}
} h_p^{k,x}, \, \displaystyle{
\max_{(s,i) \in {\cal D}^{k,x} \times
{\cal O}^{k,x}}
} \ \displaystyle{
\max_{p \in {\cal P}_{si}}
} h_p^{k,x} \right), \ \displaystyle{
\max_{(i,j) \in {\cal W}}
} \displaystyle{
\max_{p \in {\cal P}_{ij}}
} h_p^{\rm {\rm SV}} \right\},
\]
it then follows that $\bar{h}$
is a (strict) upper bound of all tuples 
$\boldsymbol{z} \in {\cal Z}$ and all path
tuples $\{ h_p^{k,x},h_p^{\rm SV} \}$.
For the travel time variables, we 
have, as an illustration,
\[
t_{i^{\prime} j^{\prime}}^{
k^{\prime},x^{\prime}} \, \leq \, 
\displaystyle{
\max_{(k,x) \in {\cal K} \times {\cal X}}
} \, \displaystyle{
\max_{(i,j) \in {\cal W}^{k,x}}
} \, \displaystyle{
\max_{p \in {\cal P}_{ij}}
} \, \displaystyle{
\max_{0 \leq \boldsymbol{h} \leq \bar{h} \onebld}
} \, \mu^{k,x} \, C_p(\boldsymbol{h}). 
\]
Therefore, letting $\bar{t}$ be a 
(strict) upper bound of all the path costs (i.e., the right-hand 
maxima in the above expression), we obtain
\[
\bar{t} \, > \, \displaystyle{
\max\left\{ \, \displaystyle{
\max_{(k,x) \in {\cal K} \times {\cal X}}
} \, \left( \displaystyle{
\max_{(i,j) \in {\cal W}^{k,x}}
} \, t_{ij}^{k,x}, \, \displaystyle{
\max_{p \in {\cal P}_{si}}
} t_{s,i}^{k,x} \right), \, \displaystyle{
\max_{(i,j) \in {\cal W}}
} \, t_{ij}^{\rm SV} \, \right\}
} \, .
\]
Based on the above bounds, we may now define
the VI that will be shown to be equivalent 
to the NCP formulation of the traffic problem
and which is the cornerstone for proving the
existence of its solution.
The VI variables are the tuples 
(\ref{eq:primary model variables}) along
with the auxiliary variables 
$\boldsymbol{\theta} \triangleq 
\left\{ \, \{ \theta_{ij}^{\, k,x} \}_{
(i,j) \in {\cal W}^{k,x}} \, \right\}_{
k \in {\cal K}}^{x \in {\cal X}}$, and
$\boldsymbol{\nu} \triangleq \left\{
\nu_{\rm AV}^k, \, \nu^k \right\}_{
k \in {\cal K}}$, which altogether 
belong to the feasible set: 
$\boldsymbol{V} \triangleq {\cal Z} 
\times {\cal H} \times [ \, 0,1 \, ]^K \times 
\mathbb{R}_+^{2| {\cal K} |}$, where
${\cal H}$ consists of all nonnegative tuples
$\boldsymbol{h} \triangleq
\left\{ h_p^{k,x}, h_p^{\rm SV}, 
t_{ij}^{k,x}, t_{s,i}^{k,x}, t_{ij}^{\rm SV}
\right\}$ with upper bounds of $\bar{h}$ and
$\bar{t}$, respectively, and $K$ is the total
number of the $\theta$-variables.  Thus,
all the multipliers of the linear constraints
are hidden in the VI formulation, but
the multipliers $\nu^k$ and
$\nu^k_{\rm AV}$ for the (nonlinear) total
fleet capacity and AV capacity constrains
are kept explicitly. [This is a major 
departure from the analysis in 
\cite{ban2019general} where the nonlinear
constraints are penalized.]
We define the 
block partitioned function
$\boldsymbol{F}\left( \boldsymbol{z},
\boldsymbol{h},\boldsymbol{\theta},
\boldsymbol{\nu} \right)$, whose blocks are
arranged in the order consistent with its arguments,

\gap

\noindent $\bullet \left( \begin{array}{l}
\left( \begin{array}{l}
-\wt{R}_{s,ij}^{\, k,{\rm AV}} - 
\alpha_1^{k,{\rm AV}} t_{ij}^{k,{\rm AV}} + 
\beta_{1}^{k,{\rm AV}} ( t_{s,i}^{k,{\rm AV}}
+ t_{ij}^{k,{\rm AV}} ) \\ [0.1in] 
- \, \beta_3^{k,{\rm AV}}t_{s,i}^{k,{\cal AV}}
+ ( t_{s,i}^{k,{\rm AV}} + 
t_{ij}^{k,{\rm AV}} ) ( \nu^k + 
\nu^k_{\rm AV})
\end{array} \right)_{k \in {\cal K}; \, 
s \in {\cal D}^{k,{\, AV}},}^{ 
(i,j) \in {\cal W}^{(k,{\rm AV})}} \\ [0.3in]
\left( \begin{array}{l}
-\wt{R}_{s,ij}^{\, k,{\rm HV}} - 
\alpha_1^{k,{\rm HV}} t_{ij}^{k,{\rm HV}} + 
\beta_{1}^{k,{\rm HV}} ( t_{s,i}^{k,{\rm HV}}
+ t_{ij}^{k,{\rm HV}} ) \\ [0.1in] 
- \, \beta_3^{k,{\rm HV}}t_{s,i}^{k,{\cal HV}}
+ ( t_{s,i}^{k,{\rm HV}} + 
t_{ij}^{k,{\rm HV}} ) \nu^k
\end{array} \right)_{k \in {\cal K}; \, 
s \in {\cal D}^{k,{\, HV}},}^{ 
(i,j) \in {\cal W}^{k,{\rm HV}}}
\end{array} \right) \mbox{\begin{tabular}{l} 
the \\
$\left( z_{ij}^{k,x} \right)$ \\
block of \\
$\boldsymbol{F}\left( \boldsymbol{z},
\boldsymbol{h},\boldsymbol{\theta},
\boldsymbol{\nu} \right)$
\end{tabular}}$

\gap

\noindent $\bullet $ $\left( \begin{array}{l}
\left( \, \begin{array}{l} 
F_{ij}^{k,x} + 
\alpha_{1}^{k,x} (t_{ij}^{k,x}- 
t_{ij}^{\, 0}) + \alpha_{2}^{k,x} 
d_{ij}^{\, 0} \, + \\ [0.1in] 
\gamma_{1}^{k,x}t_{ij}^{k,x} +
\gamma_{2}^{k,x} w_{ij}^{k,x}
\end{array} \right)_{
(i,j) \in \mathcal{W}^{k,x}}^{
(k,x) \in \mathcal{K} \times {\cal X}} 
\\ [0.3in]
\left( \alpha_1^{\rm SV}t_{ij}^{\rm SV} + 
\alpha^{\rm SV}_2 d_{ij}^{\, 0} \right)_{
(i,j) \in {\cal W}}
\end{array} \right) \mbox{\begin{tabular}{l}
the \\ 
$\left( D_{ij}^{k,x} \right)$ \\
block of \\
$\boldsymbol{F}\left( \boldsymbol{z},
\boldsymbol{h},\boldsymbol{\theta},
\boldsymbol{\nu} \right)$
\end{tabular}}$

\gap

\noindent $\bullet \left( \, \begin{array}{l}
\left( \, \mu^{k,{\rm AV}} C_p(\boldsymbol{h}) - 
t_{ij}^{k,{\rm AV}} \, \right)_{
(i,j) \in {\cal W}}^{
p \in \mathcal{P}_{ij} \, k\in\mathcal{K}} 
\\ [0.15in]
\left( \, C_p( \boldsymbol{h} ) - 
t_{s,i}^{k,{\rm AV}} \, \right)_{
(s,i) \in \mathcal{D}\times\mathcal{O}}^{ 
p \in\mathcal{P}_{si}, \, 
k \in \mathcal{K}} \\ [0.15in]
\left( \, C_p( \boldsymbol{h}) 
- t_{ij}^{k,{\rm HV}} \, \right)_{
(i,j) \in {\cal W}}^{
p \in \mathcal{P}_{ij}, \, k\in\mathcal{K}}
\\ [0.15in]
\left( \, \mu^{k,{\rm HV}} C_p( \boldsymbol{h}) - 
t_{s,i}^{k,{\rm HV}} \, \right)_{(s,i) \in
{\cal D} \times {\cal O}}^{k \in {\cal K}} 
\\ [0.15in]
\left( \, C_p( \boldsymbol{h} ) - t_{ij}^{\rm SV} 
\, \right)_{(i,j) \in {\cal W}}^{
p \in {\cal P}_{ij}}
\end{array} \right) \mbox{\begin{tabular}{l}
the \\ 
$\left( h_p^{k,x} \right)$ \\
block of \\
$\boldsymbol{F}\left( \boldsymbol{z},
\boldsymbol{h},\boldsymbol{\theta},
\boldsymbol{\nu} \right)$
\end{tabular}}$

\noindent $\bullet \left( \begin{array}{l}
\left( \, \displaystyle{
\sum_{p \in\mathcal{P}_{ij}}
} \, h_p^{k,x} - D_{ij}^{k,x} \, \right)_{
(i,j) \in {\cal W}^{k,x}}^{
(k,x) \in \mathcal{K} \times \mathcal{X}} 
\\ [0.3in]
\left( \, \displaystyle{
\sum_{p \in\mathcal{P}_{si}}
} \, h_p^{k,x} - \displaystyle{
\sum_{j\in\mathcal{D}^{k,x}}
} \, z_{s,ij}^{k,x} \, \right)_{
(s,i) \in {\cal D}^{k,x} \times 
{\cal O}^{k,x}}^{
(k,x) \in {\mathcal K}\times \mathcal{X}} 
\\ [0.25in]
\left( \, \displaystyle{
\sum_{p \in \mathcal{P}_{ij}}
} \, h_p^{\rm SV} - D_{ij}^{\rm SV} 
\, \right)_{(i,j) \in {\cal W}} 
\end{array} \right) \mbox{\begin{tabular}{l}
the \\
$\left( t_{ij}^{k,x} \right)$ \\
block of \\
$\boldsymbol{F}\left( \boldsymbol{z},
\boldsymbol{h},\boldsymbol{\theta},
\boldsymbol{\nu} \right)$
\end{tabular}}$

\gap

\noindent $\bullet \left( \, \left( \, 
\theta_{s,ij}^{k,x} 
\displaystyle{
\sum_{s^{\prime} \in \mathcal{D}^{k,x}}
} z_{s^{\prime},ij}^{k,x} - z_{s,ij}^{k,x} 
\, \right)_{s \in {\cal D}^{k,x}}^{
(i,j) \in {\cal W}^{k,x}} \, \right)_{
k \in {\cal K}}^{x \in {\cal X}}  \epc
\mbox{\begin{tabular}{l}
the \\
$\left( \theta_{s,ij}^{k,x} \right)$ \\
block of \\
$\boldsymbol{F}\left( \boldsymbol{z},
\boldsymbol{h},\boldsymbol{\theta},
\boldsymbol{\nu} \right)$
\end{tabular}}$

\gap

\noindent $\bullet \left( \, \begin{array}{l}
\left( \, N^k - \displaystyle{
\sum_{x \in {\cal X}}
} \, \displaystyle{
\sum_{(i,j)\in \mathcal{W}^{k,x}}
} \, \displaystyle{
\sum_{s\in\mathcal{D}^{k,x}}
} \, ( t_{s,i}^{k,x} + t_{ij}^{k,x} )
z^{k,x}_{s,ij} \, \right)_{
k \in {\cal K}} \\ [0.25in]
\left( \, \mu^{\rm cap}_{\rm AV} N^k - 
\displaystyle{
\sum_{(i,j)\in \mathcal{W}^{k,{\rm AV}}}
} \, \displaystyle{
\sum_{s\in\mathcal{D}^{k,{\rm AV}}}
} \, ( t_{s,i}^{k,{\rm AV}} + 
t_{ij}^{k,{\rm AV}} ) z^{k,{\rm AV}}_{s,ij} 
\, \right)_{k \in {\cal K}} 
\end{array} \right) \mbox{\begin{tabular}{l}
the \\
$\left( \nu^k, \, \nu^k_{\rm AV} \right)$ \\
block of \\
$\boldsymbol{F}\left( \boldsymbol{z},
\boldsymbol{h},\boldsymbol{\theta},
\boldsymbol{\nu} \right)$
\end{tabular}}$

\gap

\noindent We have the following main 
result, which has two
parts: the first part is the equivalence of
the VI and the mixed NCP; and the second part
is the existence of a solution to the VI, and
thus a normalized equilibrium of the 
mixed-fleet transportation system.  It is
important to point out that this result
requires minimal assumptions; in particular,
there is
no restriction on the TNCs' available fleet
sizes $N^k$ except for their positivity.
This is a significant improvement
of the existence result compared to that of 
the previous model in
\cite{ban2019general} which has a 
restriction on such fleets; see 
Lemma~3 therein. Given in the
Appendix, the proof of the
theorem requires the technical
Proposition~\ref{pr:existence VI for traffic}
which is also proved in the Appendix.

\begin{theorem} 
\label{th:equivalence and existence} \rm
Let the travel demands $D_{ij}$ be positive
for all $(i,j) \in {\cal W}$.
% For any scalar $\gamma_3 > 0$, 
The following two statements hold:

\noindent (A) The VI defined by the pair
$\boldsymbol{F}$ and $\boldsymbol{V}$ is
equivalent to the (NCP)$_{\rm main}$
formulation 
% (\ref{eq:NCP mixed-fleet I})
% (\ref{eq:NCP mixed-fleet II}) 
of the mixed-fleet transportation system.

\gap

\noindent (B) Suppose that the path cost 
functions $C_p(\boldsymbol{h})$ are continuous, and 
nonnegative and satisfy the three
conditions: (\ref{eq:path cost conditions I}) 
(\ref{eq:path cost conditions II}), and
(\ref{eq:path cost conditions III}), and that
the customer
waiting times $w_{ij}^{k,x}$ are continuous
functions of the tuples 
$\{ \, z_{ij}^{k,x},D_{ij}^{k,x}, 
t_{s,i}^{k,x}, \theta_{s;ij}^{k,x} \, \}$.  
Then the VI has a solution.
\end{theorem}

\begin{proposition} 
\label{pr:existence VI for traffic} \rm
Let $F(\boldsymbol{x},\boldsymbol{y})
\triangleq \left( \begin{array}{cc}
\Phi(\boldsymbol{x},\boldsymbol{y}) 
\\ [3pt]
\Psi(\boldsymbol{x},\boldsymbol{y})
\end{array} \right)$ be a continuous function
from $\mathbb{R}^{n + m}$ into itself.  Let
$\boldsymbol{X}$ and $\boldsymbol{Y}$ be 
closed convex sets in 
$\mathbb{R}^n$ and $\mathbb{R}^m$, 
respectively, with $\boldsymbol{X}$ being
additionally bounded.  If there exists a 
vector 
$\boldsymbol{y}^{\rm ref} \in \boldsymbol{Y}$
such that the solutions of the VI defined by 
the function
$F^{\, \tau}(\boldsymbol{x},\boldsymbol{y})
\triangleq \left( \begin{array}{cc}
\Phi(\boldsymbol{x},\boldsymbol{y}) 
\\ [3pt]
\Psi(\boldsymbol{x},\boldsymbol{y}) + 
\tau ( \boldsymbol{y} - 
\boldsymbol{y}^{\rm ref} )
\end{array} \right)$ on the set 
$\boldsymbol{X} \times \boldsymbol{Y}$
over all scalars $\tau > 0$ are bounded,
then the VI $(F,\boldsymbol{X} \times 
\boldsymbol{Y})$ has a solution.
\end{proposition}

\section{Benchmark Numerical Results}\label{sec:benchmark_numerical_results}

Our model is implemented in GAMS using the PATH solver, which is designed to solve MiCPs~\cite{dirkse1995path}. The PATH algorithm employs a Newton-based approach to efficiently locate solutions that satisfy complementarity and feasibility conditions, making it particularly suitable for large-scale equilibrium problems in transportation systems~\cite{GAMS2025}. To evaluate the proposed model, two networks are used: a small network and the Sioux-Falls network~\cite{leblanc1975efficient}. The small network is used to validate the model by comparing its results with those reported in prior work \cite{ban2019general}. The Sioux-Falls network, representing a realistic urban-scale case, is used to conduct a comprehensive analysis of AV impacts from the perspectives of companies, the overall system, and travelers.

\begin{figure}[H]
\centering
\includegraphics[width=0.3\textwidth]{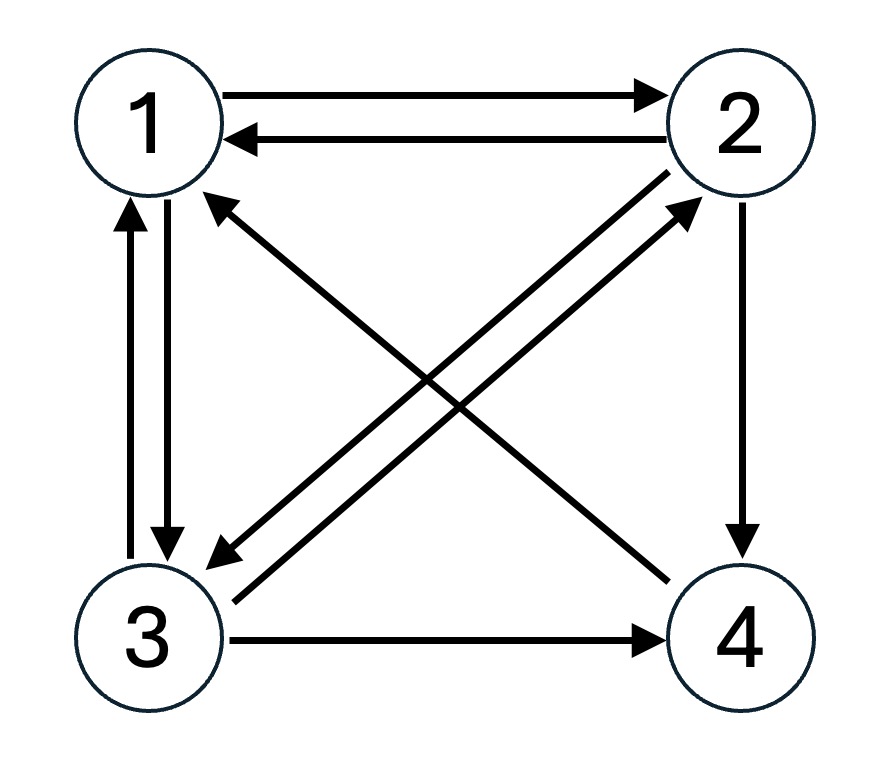}
\caption{Illustration of the Small Network}
\label{fig:geometry}
\end{figure}

The small network is a “4-node–9-link” network, as shown in Fig.~\ref{fig:geometry}. Node 1 serves as the origin, while nodes 2, 3, and 4 are destinations. The total travel demands are 50, 40, and 50, respectively. The free flow travel time, link length and link capacity of the small network are provided in Table~\ref{table:small_geometry}. Parameters related to fees and costs for each company are provided in Table~\ref{table:small_parameter_company}. Both the network geometry parameters and company parameters are the same as in prior work \cite{ban2019general}.

\begin{table}[H]
\small
\centering
\caption{Small Network Geometry Parameters}
\begin{tabular}{|c|c|c|c|c|c|}
\hline
\textbf{Links} & \textbf{From node} & \textbf{To node} & \textbf{Free flow} & \textbf{Length} & \textbf{Capacity} \\
& & & \textbf{travel time (h)} & \textbf{(mile)} & \textbf{(veh/h)}\\
\hline
1 & 1 & 2 & 0.3 & 10 & 40 \\
2 & 1 & 3 & 0.5 & 20 & 40 \\
3 & 2 & 3 & 0.4 & 20 & 60 \\
4 & 2 & 4 & 0.4 & 10 & 40 \\
5 & 3 & 4 & 0.3 & 20 & 40 \\
6 & 4 & 1 & 1.0 & 40 & 60 \\
7 & 2 & 1 & 0.4 & 15 & 50 \\
8 & 3 & 1 & 0.4 & 20 & 60 \\
9 & 3 & 2 & 0.5 & 20 & 40 \\
\hline
\end{tabular}
\label{table:small_geometry}
\end{table}

\begin{table}[H]
\small
\centering
\caption{Small Network Company Parameters}
\begin{tabular}{|c|c|c|c|}
\hline
\textbf{Parameters} & \textbf{Notation(x=HV)} & \textbf{Company 1} & \textbf{Company 2} \\
\hline
Fixed fare (\$) & $F^{k,x}$ & 3 & 2  \\
Time-based fare rate (\$/hr) & $\alpha_1^{k,x} $ & 20 & 15 \\
Distance-based fare rate (\$/mile) & $\alpha_2^{k,x} $ & 2 & 1.5 \\
Time-based conversion factor (\$/hr) & $\beta_1^{k,x}$ & 2 & 2 \\
Distance-based conversion factor (\$/mile) & $\beta_2^{k,x} $ & 0.55 & 0.9 \\
Waiting time conversion factor (\$/hr) & $\beta_3^{k,x} $ & 0.2 & 0.1 \\
Value of travel time of customer (\$/hr) & $\gamma_1^{k,x} $ & 7 & 18\\
Value of waiting time of customer (\$/hr) & $\gamma_2^{k,x} $ & 3 & 2\\
Conscience constraint & $\mu^{k,x}$ & 1.0 & 1.0\\
The number of vehicles & $N^{k}$ & 400 & 400\\
\hline
\end{tabular}
\label{table:small_parameter_company}
\end{table}

To validate our model and check its consistency
with the prior study \cite{ban2019general}, 
we first analyze the HV-only case on the small network. 
% Our model is built on a similar 
Our model is the same as the one in this reference except that it incorporates a different waiting time formulation. 
% and extends to mixed fleets. 
As such, we focus on comparing customer demand, VMT, VHT, and DHM under the same settings as the prior study, and we also compare the difference in waiting time. Since the prior study does not include AVs, the mixed-fleet part is not discussed here. The comparison results are shown in Table~\ref{table:compare}. For customer demand, VMT, VHT, and DHM, our model closely matches the results reported in \cite{ban2019general}. Regarding average waiting time, our capped formulation produces a shorter time compared with the multiplier-based formulation employed in the prior study. 

\begin{table}[H]
\small
\centering
\caption{Comparison of Our Model with the Prior Model ($x=HV, \mu^{k,x}=1.0$)}
\label{table:ten}
\begin{tabular}{|c|c|c|c|c|c|c|c|}
\hline
Model & Solo(\%) & Company 1(\%) & Company 2(\%) & VMT & VHT & DHM & Avg $w$ (min)\\
\hline

Ban et al. & 25.5 & 69.9 & 4.6 & 4943.307 & 1.456 & 3005.978 & \textbf{1.390} \\
Our Model & 29.8 & 64.3 & 6.7 & 4834.412 & 1.457 & 2940.647 & \textbf{0.714} \\

\hline
\end{tabular}
\label{table:compare}
\end{table}

The first three columns (Solo, Company 1, Company 2) in Table~\ref{table:compare} indicate the customer demand proportions for solo driving, selecting Company 1, and selecting Company 2. Customer demand is highly consistent between the two models, and the differences in total VMT and VHT are both within 10\%, indicating strong alignment in equilibrium outcomes. Importantly, our model achieves a substantially lower average waiting time of 0.714 minutes per order, representing nearly a 50\% reduction compared with prior work (1.39 minutes) \cite{ban2019general}. This improvement demonstrates that the capped customer waiting formulation introduced in our model effectively enhances system efficiency. These findings confirm the validity of our model and provide a reliable foundation for further analysis in mixed-autonomy scenarios.

\section{Case Study}\label{sec:case_study}

For large-scale analysis, we use the Sioux-Falls network shown in Figure~\ref{fig:sioux_falls_geometry}. The network contains 24 nodes and 76 directed links \cite{leblanc1975efficient}, and the dataset is openly available \cite{bstablerGitHub}. We select five nodes (1, 2, 4, 7, 9) as origins and five nodes (13, 19, 20, 23, 24) as destinations, resulting in a total of 25 OD pairs.

\begin{figure}[H]
\centering
\includegraphics[width=0.4\textwidth]{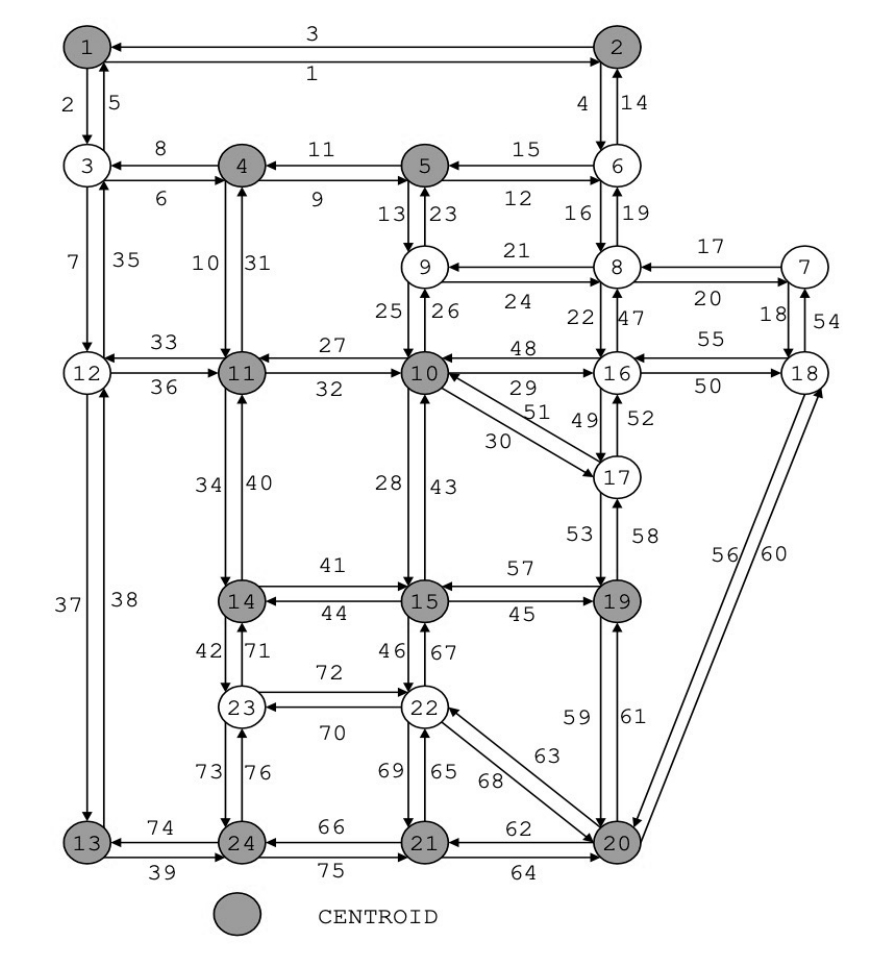}
\caption{Illustration of the Sioux-Falls Network \cite{bstablerGitHub}}
\label{fig:sioux_falls_geometry}
\end{figure}

To capture the heterogeneity in fare and cost strategies across ride-hailing companies, we define four representative company types in Figure~\ref{fig:company_strategy}, following Porter’s Generic Strategies framework~\cite{miller1986porter}. 
% Cost Leadership, Differentiation, focus

\gap 

\noindent $\bullet $ \textbf{Company 1 (Technology driven)}: adopts a differentiation strategy, emphasizing advanced IT capabilities and autonomous driving technologies. It offers medium–low fare levels to attract early adopters while maintaining low operational costs through automation and optimized fleet management. Company 1 is denoted as k1. 
  %The price and cost for AVs are set at 1.1 times those of HVs to reflect the higher technological investment in autonomous vehicles.

\gap 

\noindent $\bullet $ \textbf{Company 2 (Aggressive entrant)}: follows a cost-leadership strategy with an aggressive market-entry approach. It adopts low fares to rapidly gain market share, despite incurring high operating costs due to substantial capital expenditures and rapid fleet deployment. Company 2 is denoted as k2. 
  %($\beta_1^{k,x}$, $\beta_2^{k,x}$, $\beta_3^{k,x}$). Similar to Company 1, the AV price and cost are set at 1.1 times those of HVs to reflect the early-stage technological cost burden.

\gap 

\noindent $\bullet $ \textbf{Company 3 (Market leader)}: represents a scale-based dominance strategy. Leveraging its large user base and operational experience, it sets high prices while benefiting from medium–low costs, owing to economies of scale and optimized routing. Company 3 is denoted as k3. 
   %For AV operations, both prices and costs are reduced by a factor of 0.9 relative to HVs, assuming efficiency gains from automation and fleet coordination.

\gap 

\noindent $\bullet $ \textbf{Company 4 (Competitive co-player)}: acts as a focus-differentiation competitor, a smaller but strategically adaptive company operating alongside the market leader. It sets medium–high fare levels and experiences medium–high operating costs, reflecting its intermediate market position. Company 4 is denoted as k4. 
   %Similar to the market leader, its AV prices and costs are scaled by 0.9 relative to HVs, reflecting moderate efficiency improvements from automation.

\begin{figure}[H]
\centering
\includegraphics[width=0.5\textwidth]{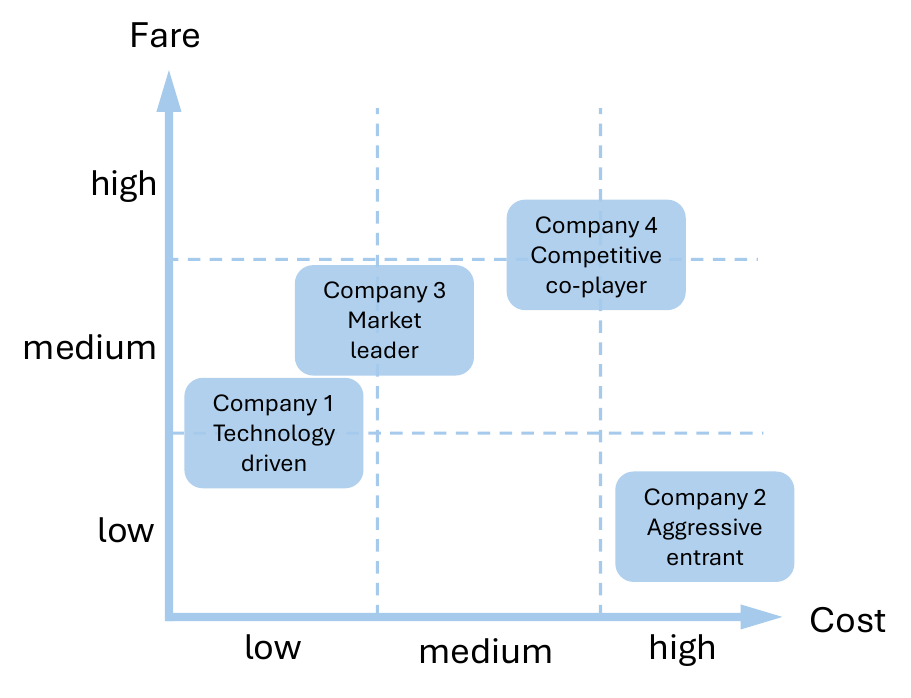}
\caption{Company Strategy}
\label{fig:company_strategy}
\end{figure}

\subsection{Company-focused Analysis}
In this subsection, we investigate how variations in AV penetration rate $\mu^{\rm cap}$ and AV relaxation parameter $\mu^{k,{\rm AV}}$ shape company-focused outcomes. Our analysis primarily examines their effects on company profits and the demand for company vehicles, while also exploring the competitive mechanisms that emerge among companies. In the following analysis, the HV relaxation parameter $\mu^{k,{\rm HV}}$ is set to 1.1, as HVs are not directly controllable.

\subsubsection{Impacts of AV Relaxation Parameter}

For cross-company comparisons of profitability, we analyze AV profitability under two distinct scenarios. \textbf{Homogeneous Relaxation:} all four companies adopt the same \(\mu^{k,{\rm AV}}\), which varies along the x-axis but remains identical across companies in Figure \ref{fig:av_profit_homo}. \textbf{Heterogeneous Relaxation:} company 1 adopts a more aggressive strategy with \(\mu^{1,AV} > 1\), which varies along the x-axis, while the remaining companies maintain \(\mu^{k,AV} = 1\) in Figure \ref{fig:av_profit_hetero}.

\begin{figure}[H]
\centering
\begin{subfigure}{0.32\textwidth}
    \includegraphics[width=\textwidth]{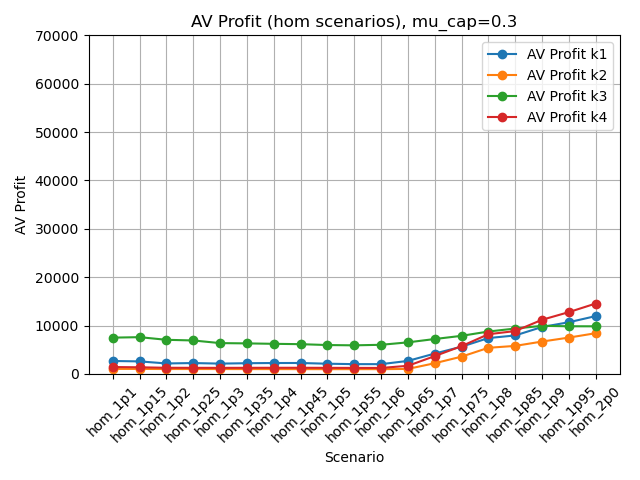}
    \caption{AV Profit ($\mu^{\rm cap}=0.3$)}
\end{subfigure}
\begin{subfigure}{0.32\textwidth}
    \includegraphics[width=\textwidth]{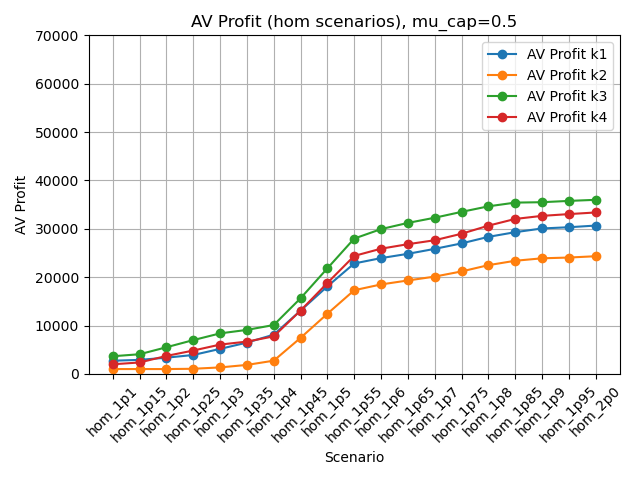}
    \caption{AV Profit ($\mu^{\rm cap}=0.5$)}
\end{subfigure}
\begin{subfigure}{0.32\textwidth}
    \includegraphics[width=\textwidth]{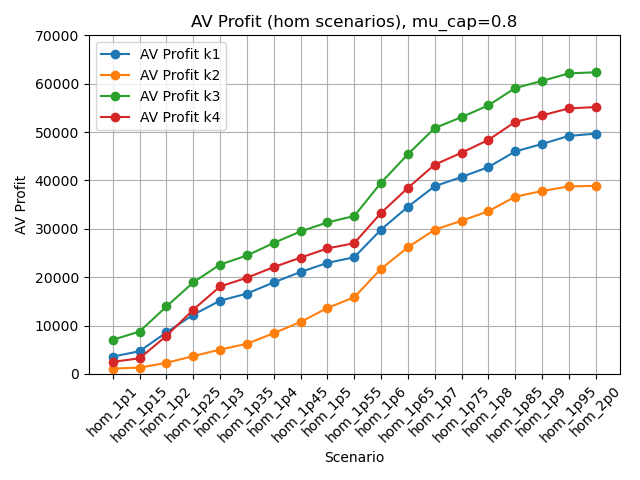}
    \caption{AV Profit ($\mu^{\rm cap}=0.8$)}
\end{subfigure}
\caption{AV Profit under Different $\mu^{\rm cap}$ (Homogeneous Scenarios)} 
\label{fig:av_profit_homo}
\end{figure}

\iffalse
\begin{figure}[H]
\centering
\begin{subfigure}{0.4\textwidth}
    \includegraphics[width=\textwidth]{figs/mucap03_v1/profit_HV_mucap03.png}
    \caption{profit HV ($\mu^{\rm cap}=0.3$)}
\end{subfigure}
\begin{subfigure}{0.4\textwidth}
    \includegraphics[width=\textwidth]{figs/mucap05_v1/profit_HV_mucap05.png}
    \caption{profit HV ($\mu^{\rm cap}=0.5$)}
\end{subfigure}
\begin{subfigure}{0.4\textwidth}
    \includegraphics[width=\textwidth]{figs/mucap08_v1/profit_HV_mucap08.png}
    \caption{profit HV ($\mu^{\rm cap}=0.5$)}
\end{subfigure}
\caption{HV Profit of Each Company (homo scenarios)} 
\label{fig:hv_profit_homo}
\end{figure}
\fi

Figure~\ref{fig:av_profit_homo} illustrates the relationship between the relaxation parameter \(\mu^{k,{\rm AV}}\) and companies’ AV profitability under homogeneous scenarios. Across companies, profits exhibit broadly similar increasing trends as \(\mu^{k,{\rm AV}}\) varies. This pattern can be explained by the fact that \(\mu^{k,{\rm AV}}\) reflects the degree of control companies exert over AVs: higher values allow companies to reassign AVs more flexibly, thereby generating additional revenue. 
Notably, the market leader (k3) consistently maintains its dominant position in most cases as \(\mu^{k,{\rm AV}}\) increases, reflecting the persistence of its competitive advantage in this setting. Furthermore, higher levels of AV penetration are associated with a faster growth rate of profits, suggesting that greater adoption of AVs amplifies the benefits realized by companies. This occurs because a larger share of AVs in the network allows companies to exert greater influence over system performance, thereby enhancing their capacity to increase profits.

\begin{figure}[H]
\centering
\begin{subfigure}{0.32\textwidth}
    \includegraphics[width=\textwidth]{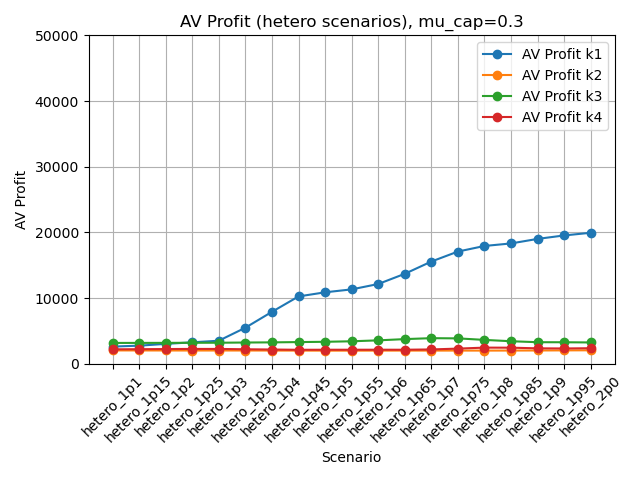}
    \caption{AV Profit ($\mu^{\rm cap}=0.3$)}
\end{subfigure}
\begin{subfigure}{0.32\textwidth}
    \includegraphics[width=\textwidth]{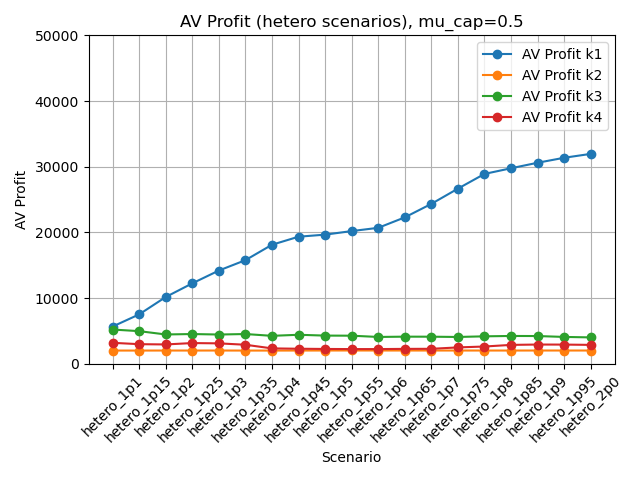}
    \caption{AV Profit ($\mu^{\rm cap}=0.5$)}
\end{subfigure}
\begin{subfigure}{0.32\textwidth}
    \includegraphics[width=\textwidth]{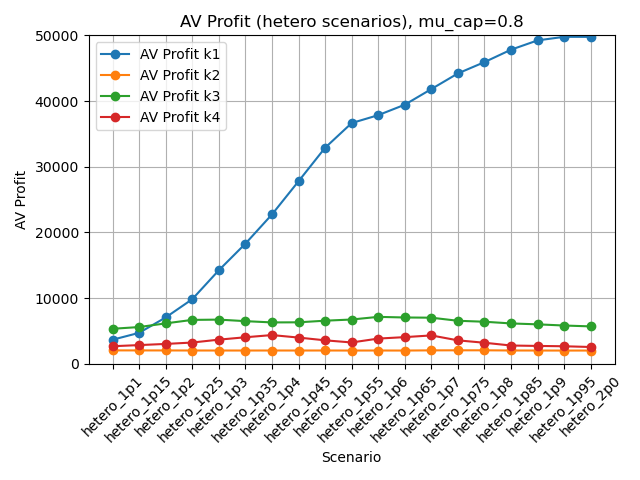}
    \caption{AV Profit ($\mu^{\rm cap}=0.8$)}
\end{subfigure}
\caption{AV Profit under Different $\mu^{\rm cap}$ (Heterogeneous Scenarios)} 
\label{fig:av_profit_hetero}
\end{figure}

\iffalse
\begin{figure}[H]
\centering
\begin{subfigure}{0.4\textwidth}
    \includegraphics[width=\textwidth]{figs/hetero_results/profit_HV_mucap03.png}
    \caption{profit HV $\mu^{\rm cap}=0.3$}
\end{subfigure}
\begin{subfigure}{0.4\textwidth}
    \includegraphics[width=\textwidth]{figs/hetero_results/profit_HV_mucap05.png}
    \caption{profit HV $\mu^{\rm cap}=0.5$}
\end{subfigure}
\begin{subfigure}{0.4\textwidth}
    \includegraphics[width=\textwidth]{figs/hetero_results/profit_HV_mucap08.png}
    \caption{profit HV $\mu^{\rm cap}=0.8$}
\end{subfigure}
\caption{HV Profit of Each Company (hetero scenarios)} 
\label{fig:hv_profit_hetero}
\end{figure}
\fi

Figure~\ref{fig:av_profit_hetero} presents the profitability outcomes under heterogeneous scenarios, where the \(x\)-axis represents \(\mu^{1,AV}\) and \(\mu^{k,{\rm AV}}=1\) for \(k = 2,3,4\). Results indicate that the technology-driven company (k1), which may exert more control on AVs, consistently secures substantially greater profits, producing a pronounced gap relative to its moderate competitors. For each individual company, adopting such a control-focused strategy compared with the rivals can therefore yield significant advantages. However, this also highlights a potential risk
% : without adequate regulatory oversight, 
that overly profit-pursuing strategies may undermine fair competition. Accordingly, this framework serves as a tool for system planners in regulating AV deviation behavior and AV adoption.

\begin{figure}[H]
\centering
\begin{subfigure}{0.35\textwidth}
    \includegraphics[width=\textwidth]{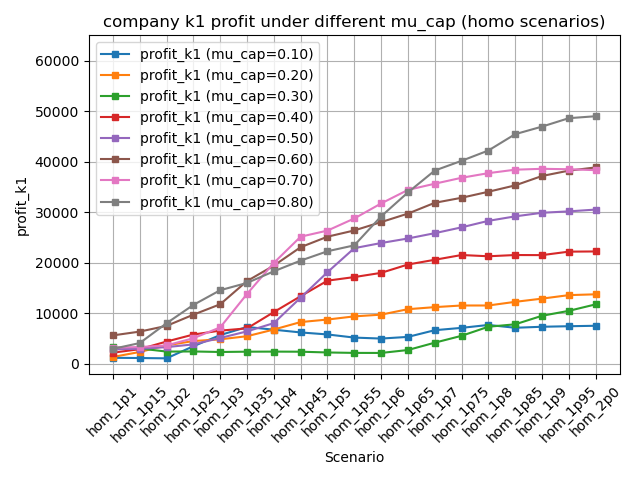}
    \caption{Company 1 Total Profit}
\end{subfigure}
\begin{subfigure}{0.35\textwidth}
    \includegraphics[width=\textwidth]{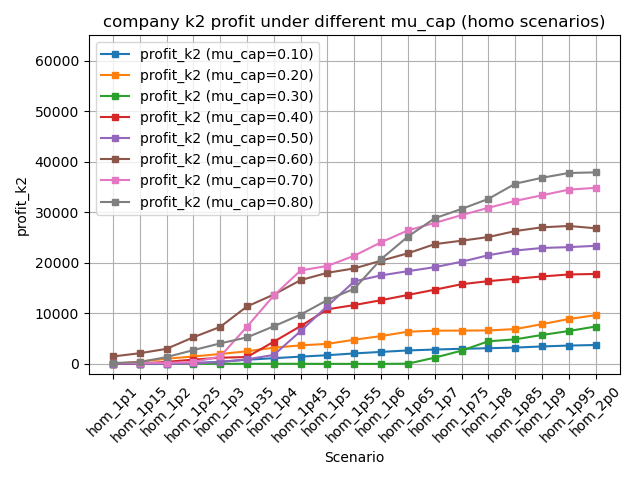}
    \caption{Company 2 Total Profit}
\end{subfigure}
\begin{subfigure}{0.35\textwidth}
    \includegraphics[width=\textwidth]{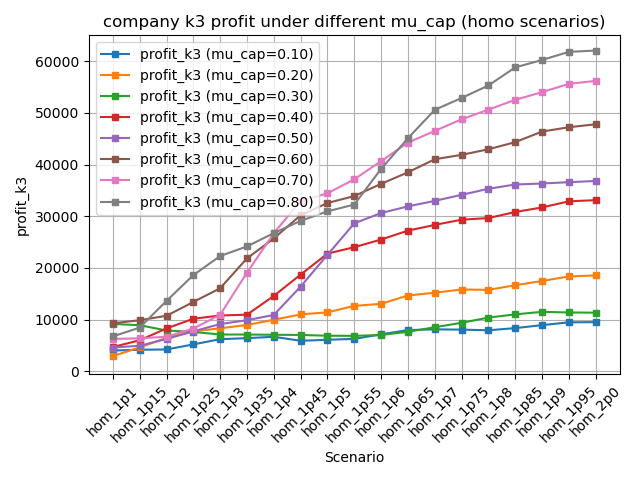}
    \caption{Company 3 Total Profit}
\end{subfigure}
\begin{subfigure}{0.35\textwidth}
    \includegraphics[width=\textwidth]{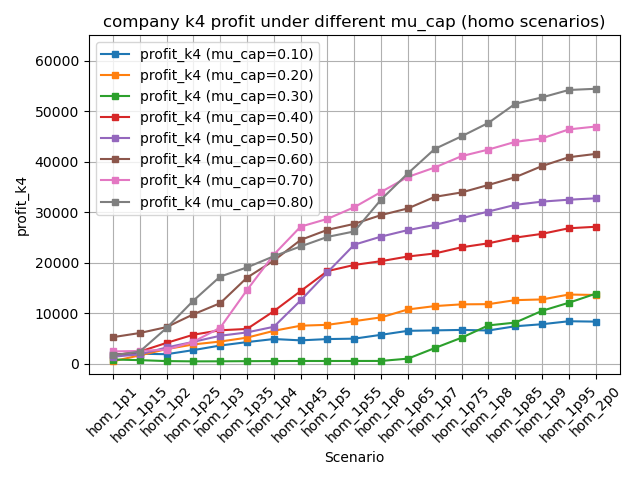}
    \caption{Company 4 Total Profit}
\end{subfigure}
\caption{Total Profit under Different $\mu^{\rm cap}$ (Homogeneous Scenario)} 
\label{fig:profit_diff_mucap}
\end{figure}

For the within-company analysis of profitability in Figure~\ref{fig:profit_diff_mucap}, the horizontal axis denotes \(\mu^{k,{\rm AV}}\). Each subplot corresponds to a single company, while the multiple curves within each plot represent different AV penetration rates \(\mu^{\rm cap}\). As \(\mu^{k,{\rm AV}}\) increases, companies with lower relaxation parameters (below 1.15) benefit most under moderate AV penetration (\(\mu^{\rm cap} = 0.6\)), suggesting that a balanced AV/HV market supports profitability at lower control levels. For moderate to high \(\mu^{k,{\rm AV}}\), profits approach a “saturation point", where further increments contribute little to profits. This indicates that a near-optimal level of AV routing control exists for market development, and excessively high AV shares provide limited additional benefits for companies compared with moderate levels in most cases. 

For the aggregate market share analysis in Figure \ref{figs:market_share}, subfigures (a) and (b) illustrate opposite trends in AV and HV as $\mu^{k,{\rm AV}}$ increases. This indicates that, when companies adopt higher AV coordination, travelers tend to shift from choosing AVs to HVs within ride-hailing services. Furthermore, as shown in subfigures (c) and (d), the companies’ total market share (AV and HV demand) declines as $\mu^{k,{\rm AV}}$ increases. This suggests that higher relaxation parameters generally weaken companies’ competitiveness, causing more travelers to choose solo driving.

\begin{figure}[H]
\centering
\begin{subfigure}{0.35\textwidth}
    \includegraphics[width=\textwidth]{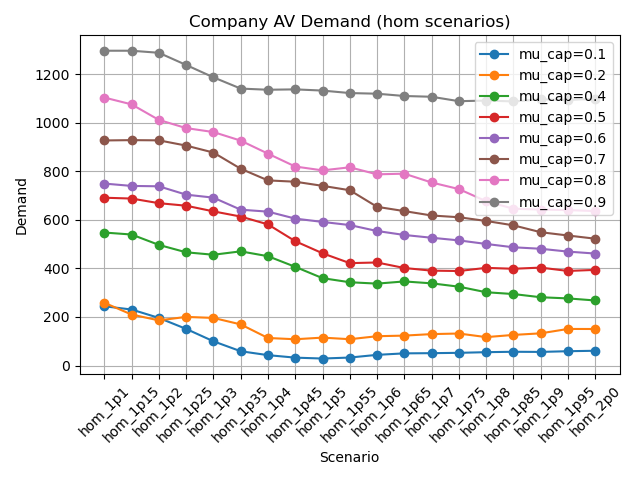}
    \caption{AV Demand for All Company}
\end{subfigure}
\begin{subfigure}{0.35\textwidth}
    \includegraphics[width=\textwidth]{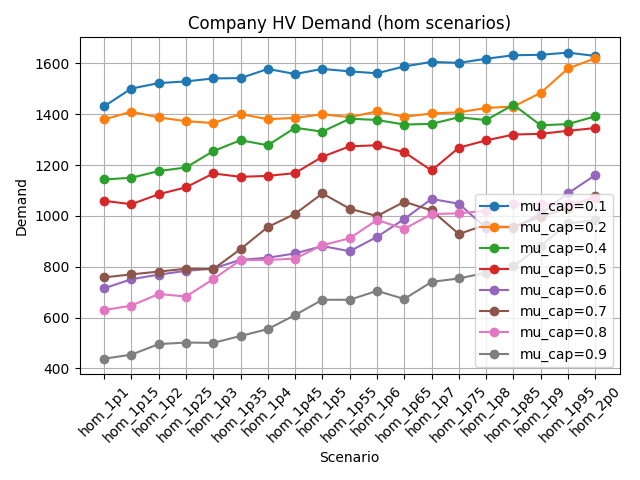}
    \caption{HV Demand for All Company}
\end{subfigure}
\begin{subfigure}{0.35\textwidth}
    \includegraphics[width=\textwidth]{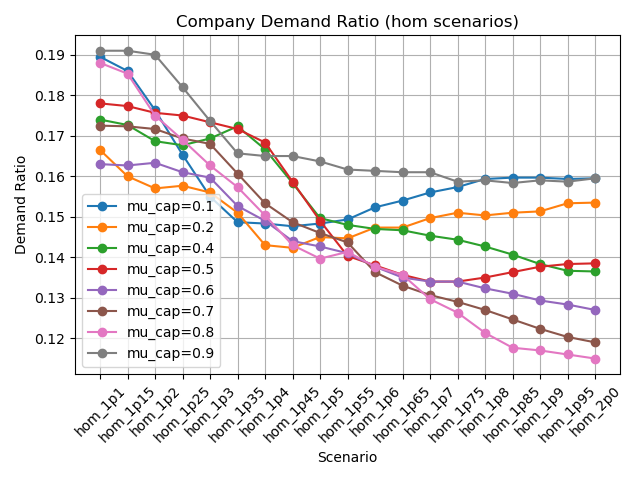}
    \caption{Company Demand (HV and AV) Ratio}
\end{subfigure}
\begin{subfigure}{0.35\textwidth}
    \includegraphics[width=\textwidth]{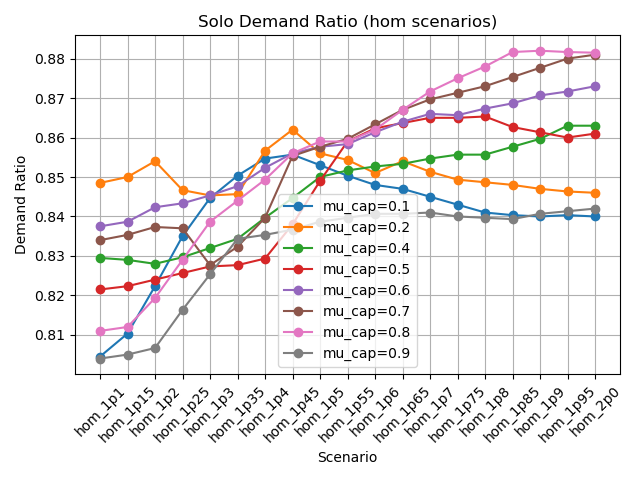}
    \caption{Solo Driver Demand Ratio}
\end{subfigure}
\caption{Market Share under Different $\mu^{k,{\rm AV}}$ (Homogeneous Scenario)}
\label{figs:market_share}
\end{figure}

\subsubsection{Impacts of AV Penetration Rate}

Figure~\ref{fig:profit_diff_muAV} illustrates the impact of the AV penetration rate $\mu^{\rm cap}$ on each company’s total profit. When $\mu^{\rm cap}$ is very low (below 0.2) or very high (above 0.8), the profit growth remains modest. In contrast, intermediate values of $\mu^{\rm cap}$ exhibit a notable increase in profit as the penetration rate rises. These findings suggest that a moderate level of AV penetration is most beneficial for companies, striking a balance between operational efficiency and market demand.
\begin{figure}[H]
\centering
\begin{subfigure}{0.34\textwidth}
    \includegraphics[width=\textwidth]{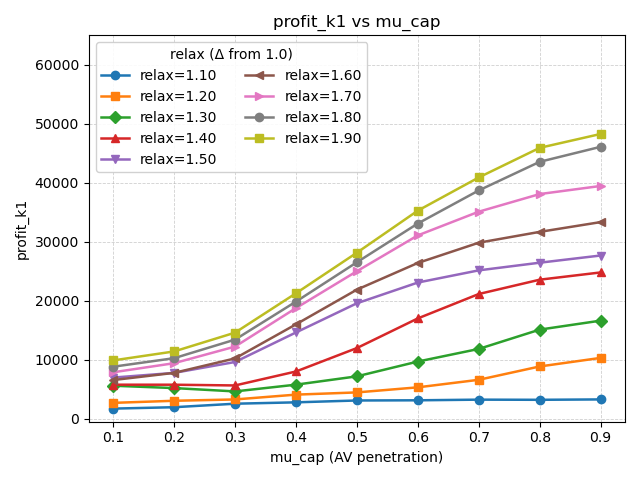}
    \caption{Company 1 Total Profit}
\end{subfigure}
\begin{subfigure}{0.34\textwidth}
    \includegraphics[width=\textwidth]{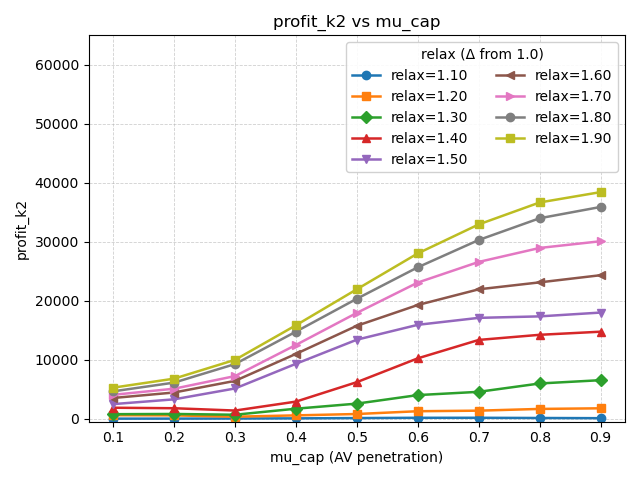}
    \caption{Company 2 Total Profit}
\end{subfigure}
\begin{subfigure}{0.34\textwidth}
    \includegraphics[width=\textwidth]{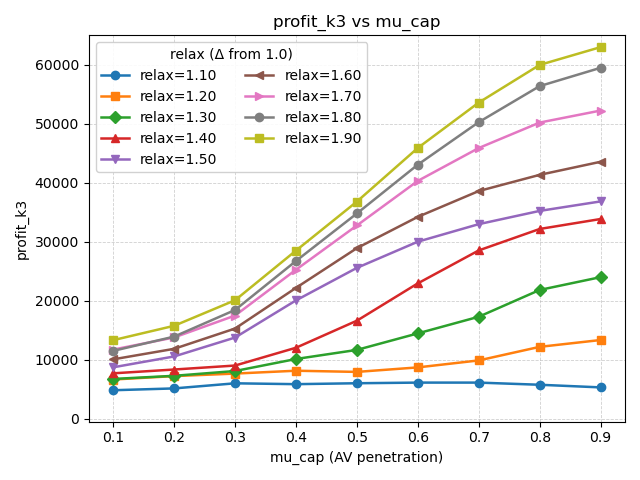}
    \caption{Company 3 Total Profit}
\end{subfigure}
\begin{subfigure}{0.34\textwidth}
    \includegraphics[width=\textwidth]{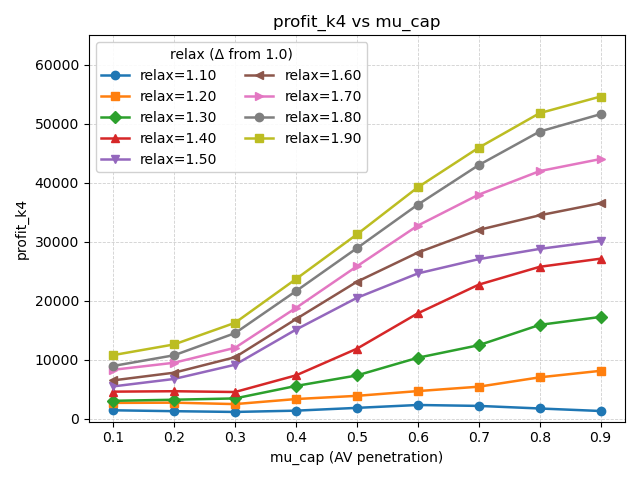}
    \caption{Company 4 Total Profit}
\end{subfigure}
\caption{Total Profit under Different $\mu^{k,{\rm AV}}$} 
\label{fig:profit_diff_muAV}
\end{figure}

\subsection{System and Traveler Analysis}

In the system-related figures, we present key system-level indicators, including VMT, VHT, average Wardrop travel time, and average fare. For VMT and VHT, Figure~\ref{fig:system_impact} shows that their totals generally decrease as the AV relaxation parameter ($\mu^{k,{\rm AV}}$) increases. This indicates that the introduction of AVs can help alleviate system congestion by improving overall network efficiency. Moreover, the results suggest that maintaining $\mu^{k,{\rm AV}}$ at approximately 0.3–0.7 could yield the most balanced system performance. These findings imply that policymakers and transportation planners could promote appropriately regulated AV operations to enhance system-level efficiency.

%In this set of figures, the horizontal axis represents the AV relaxation parameter ($\mu^{k,{\rm AV}}$), and each subplot corresponds to a different AV penetration rate ($\mu^{\rm cap}$). Subfigure (a) shows the total VMT, and Subfigure (b) displays the total VHT required to satisfy all travel demands. These figures illustrate the interplay between AV behavioral flexibility and fleet penetration, highlighting their combined influence on overall network efficiency.

\begin{figure}[H]
\centering
\begin{subfigure}{0.4\textwidth}
    \includegraphics[width=\textwidth]{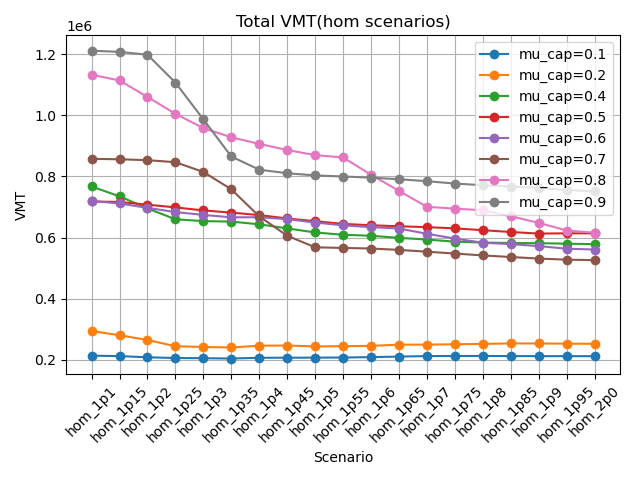}
    \caption{total VMT}
\end{subfigure}
\begin{subfigure}{0.4\textwidth}
    \includegraphics[width=\textwidth]{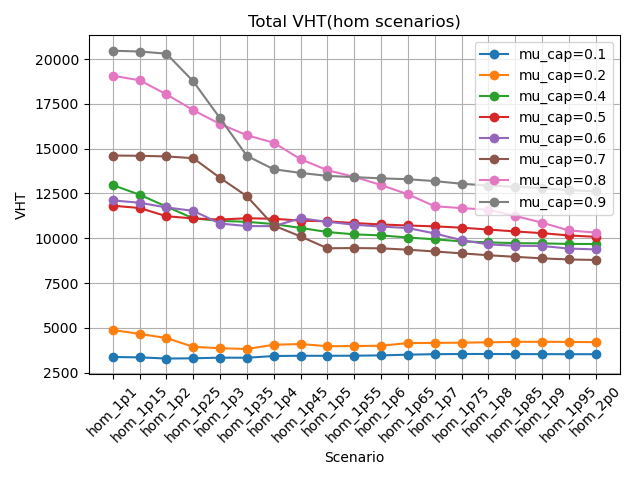}
    \caption{total VHT}
\end{subfigure}
\caption{System-level Impact under Different $\mu^{k,{\rm AV}}$} 
\label{fig:system_impact}
\end{figure}

For average Wardrop time (defined as the average minimum travel time across all OD pairs and travelers), Figure \ref{fig:avg_wardrop} shows that it is not significantly affected by variations in $\mu^{\rm cap}$ and $\mu^{k,{\rm AV}}$. However, subfigure (b) reveals a consistent pattern: the average Wardrop time achieves its minimum at a moderate AV penetration rate ($\mu^{\rm cap} \approx 0.5$) across different relaxation parameters. This suggests that an intermediate level of AV penetration can maximize its positive impact on network congestion.

\begin{figure}[H]
\centering
\begin{subfigure}{0.4\textwidth}
    \includegraphics[width=\textwidth]{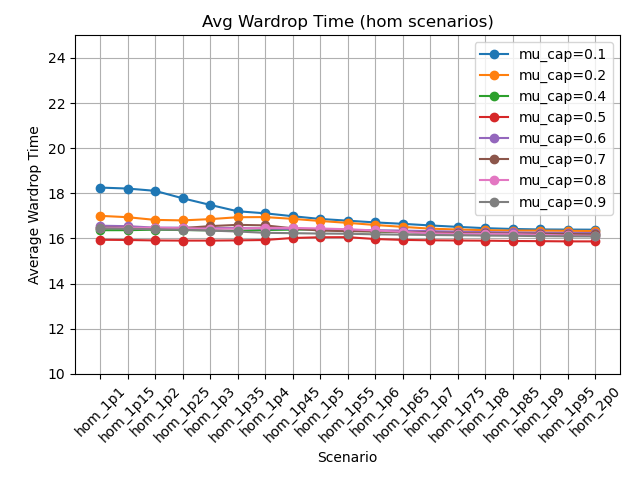}
    \caption{Average Wardrop Time via $\mu^{k,{\rm AV}}$}
\end{subfigure}
\begin{subfigure}{0.4\textwidth}
    \includegraphics[width=\textwidth]{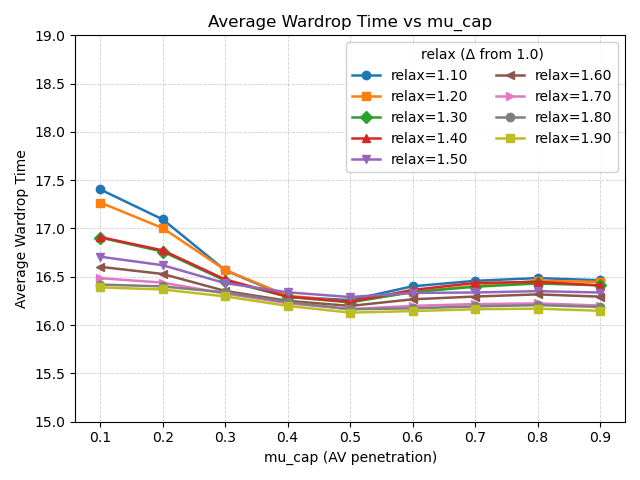}
    \caption{Average Wardrop Time via $\mu^{\rm cap}$}
\end{subfigure}
\caption{System-focus and Traveler-focus Time Impact under Different $\mu^{k,{\rm AV}}$ and $\mu^{\rm cap}$} 
\label{fig:avg_wardrop}
\end{figure}

For average fare, Figure~\ref{fig:avg_fare} shows that the average HV fare remains relatively stable while AV fare increases as $\mu^{k,{\rm AV}}$ rises. When $\mu^{\rm cap}>0.5$ and $\mu^{k,{\rm AV}}<1.4$, the average AV fare becomes lower than the average HV fare. This indicates that, under appropriately managed operations, AVs can provide cost advantages to travelers.

\begin{figure}[H]
\centering
\begin{subfigure}{0.4\textwidth}
    \includegraphics[width=\textwidth]{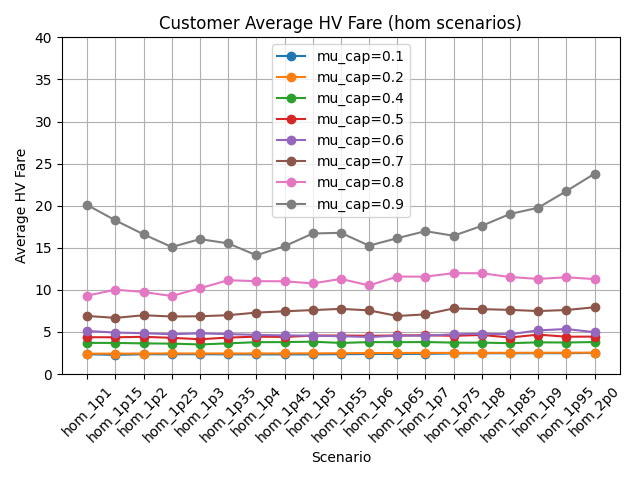}
    \caption{Average HV Fare}
\end{subfigure}
\begin{subfigure}{0.4\textwidth}
    \includegraphics[width=\textwidth]{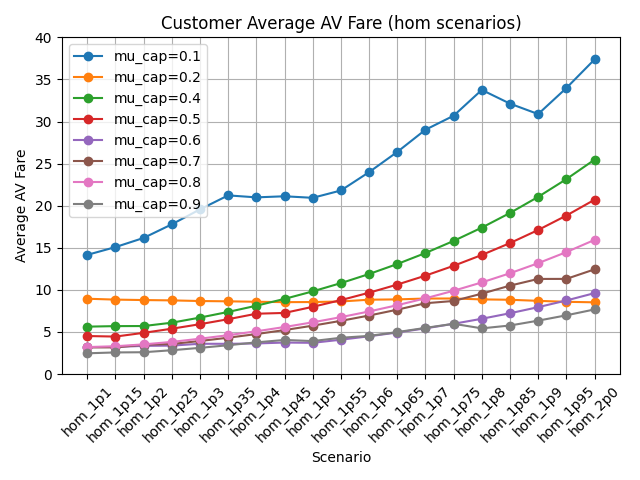}
    \caption{Average AV Fare}
\end{subfigure}
\caption{System-focus and Traveler-focus Fare Impact under Different $\mu^{k,{\rm AV}}$} 
\label{fig:avg_fare}
\end{figure}

For traveler-focused analysis, Figure~\ref{fig:avg_wardrop} shows that when the transportation system reaches a balanced composition of AVs and HVs, solo drivers can experience modest benefits, such as reduced travel times and lower congestion levels. Figure~\ref{fig:avg_fare} suggests that ride-hailing customers’ mode choices are sensitive to the prevailing AV penetration level, highlighting that market composition plays a key role in shaping travelers’ cost-based preferences.

\subsection{Summary}
This section reports the results in the solution
of a nonlinear complementarity problem, formulated from the KKT conditions of each subproblem, using the PATH solver in GAMS. On a small network, the model demonstrates the reliability and effectiveness of the proposed waiting-time formulation. On the Sioux Falls network, we further analyze the roles of $\mu^{\rm cap}$ and $\mu^{k,{\rm AV}}$ from company-level, system-level, and traveler-level perspectives. The results indicate that introducing a moderate proportion of AVs into the fleet can simultaneously enhance company profits, alleviate network congestion, reduce travel times for human drivers, and lower travel costs for passengers. 

\section{Conclusion} \label{sec:conclusion}
This paper proposes a unified equilibrium model 
for mixed-fleet ride-hailing systems, integrating 
the interactions among MiFleet TNCs, travelers, and 
traffic. The framework is flexible and can be 
extended to accommodate multiple heterogeneous fleet 
types, explicitly differentiating AV and HV behaviors 
across operational stages. It further introduces the
idea of customer waiting functions that include as a 
special case of a truncated, congestion-dependent, 
queue-based formulation for customer waiting time 
that endogenously links customer delay to network 
congestion.  We provide a rigorous proof of 
equilibrium existence under the mere continuity
of the model functions and the weak positivity of the
path costs. These theoretical advances enable a more
realistic and analytically tractable representation 
of customer waiting and traffic flow interactions 
in mixed-autonomy environments.

The proposed model effectively captures multilevel 
impacts across company, system, and traveler scales, 
uncovering key patterns in AV deployment. The 
numerical results indicate that a moderate or properly 
tuned AV penetration rate achieves the most balanced 
outcomes, enhancing company profitability, mitigating
congestion, and improving traveler welfare. 
In contrast, excessive automation introduces 
inefficiencies, such as system congestion, 
underscoring the importance of appropriate regulatory
oversight.

Beyond these immediate findings, the proposed
framework provides a foundation for future research 
on equilibrium modeling in intelligent transportation 
systems and ride-hailing company strategies, including
the analysis of market performance and design of
suitable pricing strategies under low and high AV penetration rates. By combining 
analytical rigor with practical interpretability, 
our study contributes to the broader vision of 
sustainable, coordinated, and human-centered 
automation, offering guidance to policymakers and 
platform operators navigating the evolving landscape 
of mixed-autonomy mobility.

\gap

\noindent {\bf Acknowledgements.}  The authors
are very grateful to Professors Yafeng Yin and 
Yueyue Fan for a very stimulating, insightful,
and inspiring discussion on the customer
waiting times and also for the idea of waiting-induced
travel demands as a form of elastic demands
in general.  In particular, Professor Yin shared
his perspective on the capping in the well-known
BPR travel time function as a consequence of a
queueing formalism.  While acknowledging some
similarities, our capped-queue waiting time
is different from the capping in the BPR 
and more general queue-based travel time functions.
Most importantly, the idea of capping in our contexts
captures the possibility for travelers not to wait
for TNC service that is a realistic feature in travel
choices with solo driving.  

%-------------------------------
% APPENDIX
%-------------------------------
\appendix
\section*{Appendix}

\renewcommand{\thesection}{A\arabic{section}}

\setcounter{section}{0}

\section{Proof of Theorem \ref{th:equivalence and existence}}
\begin{proof} We establish statement (A) by 
showing that the KKT conditions of
the VI whose feasible set $\boldsymbol{V}$ 
is polyhedral is equivalent to the said NCP.
The only difference 
between the former conditions and the latter
% (\ref{eq:NCP mixed-fleet I})--
% (\ref{eq:NCP mixed-fleet II}) 
is in the complementarity conditions of the 
$t$ and $h$-variables.
% (note: the travel demand equality 
% $D_{ij}^{\rm SV} + \displaystyle{
% \sum_{x\in\mathcal{X}}
% } \, \displaystyle{
% \sum_{k\in \mathcal{K}^x_{ij}}
% } \ D_{ij}^{k,x} \, = \, D_{ij}$ is equivalent
% to an inequality by 
% Lemma~\ref{lm:demand inequality})).
Specifically, the complementarity
conditions of these travel time and path
flow variables need to be modified to 
include the
multipliers---denoted by $u$ and $v$ 
below---of the bound constraints, 
resulting in the following relevant
modified constraints:
\[
\begin{array}{l}
\left. \begin{array}{l}
0 \, \leq \, t_{ij}^{\rm SV} \, \perp \,
\displaystyle{
\sum_{p \in\mathcal{P}_{ij}}
} \, h_p^{\rm SV} - D_{ij}^{\rm SV} + u_{ij}^{\rm SV}
\, \geq \, 0 \\ [0.2in]
0 \, \leq \, u_{ij}^{\rm SV} \, \perp \,
\bar{t} - t_{ij}^{\rm SV} \, \geq \, 0
\end{array} \right\} 
\epc \forall \, 
x \in \mathcal{X}, \, k \in \mathcal{K}, 
\, (i,j) \in {\cal W}^{k,x} \\ [0.35in]
\left. \begin{array}{l}
0 \, \leq \, t_{ij}^{k,x} \, \perp \,
\displaystyle{
\sum_{p \in\mathcal{P}_{ij}}
} \, h_p^{k,x} - D_{ij}^{k,x} + u_{ij}^{k,x}
\, \geq \, 0 \\ [0.2in]
0 \, \leq \, u_{ij}^{k,x} \, \perp \,
\bar{t} - t_{ij}^{k,x} \, \geq \, 0
\end{array} \right\} 
\epc \forall \, 
x \in \mathcal{X}, \, k \in \mathcal{K}, 
\, (i,j) \in {\cal W}^{k,x} \\ [0.35in]
\left. \begin{array}{l}
0 \, \leq \, t_{s,i}^{k,x} \, \perp \,
\displaystyle{
\sum_{p \in\mathcal{P}_{si}}
} \, h_p^{k,x} - \displaystyle{
\sum_{j\in\mathcal{D}^{k,x}}
} \, z_{s,ij}^{k,x} + u_{s,i}^{k,x} 
\, \geq \, 0 \\ [0.2in]
0 \, \leq \, u_{s,i}^{k,x} \, \perp \, 
\bar{t} - t_{s,i}^{k,x} \, \geq \, 0
\end{array} \! \right\} \,
\forall \, x \in \mathcal{X}, \, 
k \in \mathcal{K}, \, 
(s,i) \in \mathcal{D}^{k,x} \times 
\mathcal{O}^{k,x} \\ [0.35in]
\left. \begin{array}{l}
0 \, \leq \, h_p^{\rm SV} \, \perp \, 
C_p( \boldsymbol{h} ) - t_{ij}^{\rm SV} + 
v_p^{\rm SV} \, \geq \, 0 \\ [0.1in]
0 \, \leq \, v_p^{\rm SV} \, \perp \,
\bar{h} - h_{p}^{\rm SV} \, \geq \, 0
\end{array} \right\} \epc 
\forall \, (i,j) \in {\cal W}, 
\, p \in\mathcal{P}_{ij} \\ [0.3in]
\left. \begin{array}{l}
0 \, \leq \, h_p^{k,{\rm AV}} \, \perp \, 
\mu^{k,{\rm AV}} C_p(\boldsymbol{h}) - 
t_{ij}^{k,{\rm AV}} + v_p^{k,{\rm AV}}
\geq 0 \\ [0.1in]
0 \, \leq \, v_p^{k,{\rm AV}} \, \perp \,
\bar{h} - h_{p}^{k,{\rm AV}} \, \geq \, 0
\end{array} \right\}
\epc \forall \, (i,j) \in {\cal W}, \, 
p \in \mathcal{P}_{ij}, \, k\in\mathcal{K} 
\\ [0.3in]
\left. \begin{array}{l}
0 \, \leq \, h_p^{k,{\rm AV}} \, \perp \,  
C_p( \boldsymbol{h} ) - 
t_{s,i}^{k,{\rm AV}} + v_p^{k,{\rm AV}}
\geq 0 \\ [0.1in]
0 \leq v_p^{k,{\rm AV}} \perp 
\bar{h} - h_{p}^{k,{\rm AV}} \, \geq \, 0
\end{array} \right\} \epc \forall \, 
(s,i) \in \mathcal{D}\times\mathcal{O}, 
\, p \in\mathcal{P}_{si}, \, 
k \in \mathcal{K}
\end{array} \]
\[ \begin{array}{l}
% \\ [0.3in]
\left. \begin{array}{l}
0 \, \leq \, h_p^{k,{\rm HV}} \, \perp \, 
C_p( \boldsymbol{h}) 
- t_{ij}^{k,{\rm HV}} + v_p^{k,{\rm HV}}
\, \geq \, 0 \\ [0.1in]
0 \, \leq \, v_p^{k,{\rm HV}} \, \perp \, 
\bar{h} - h_p^{k,{\rm HV}} \, \geq \, 0
\end{array} \right\} \epc
\forall \, (i,j) \in {\cal W}, \, 
p \in \mathcal{P}_{ij}, \, k\in\mathcal{K} 
\\ [0.3in]
\left. \begin{array}{l}
0 \, \leq \, h_p^{k,{\rm HV}} \, \perp \, 
\mu^{k,{\rm HV}} C_p( \boldsymbol{h}) - 
t_{s,i}^{k,{\rm HV}} + v_p^{k,{\rm HV}}
\, \geq \, 0 \\ [0.1in]
0 \, \leq \, v_p^{k,{\rm HV}} \, \perp \,
\bar{h} - h_p^{k,{\rm HV}} \, \geq \, 0
\end{array} \! \right\} \forall \, 
(s,i) \in \mathcal{D}\times\mathcal{O}, 
\, p \in\mathcal{P}_{si}, \, 
k \in \mathcal{K}.
\end{array}  \]
We claim that all the multipliers $u$'s and 
$v$'s of the bound constraints are equal 
to zero.
In fact, suppose that $u_{ij}^{k,x} > 0$
for some $(k,x) \in {\cal K} \times {\cal X}$ 
and $(i,j) \in {\cal W}^{k,x}$.  We then have
$t_{ij}^{k,x} = \bar{t} > 0$.  Thus, by
complementarity, it follows that
\[
\bar{h} \, > \, D_{ij}^{k,x} \, = \, 
\displaystyle{
\sum_{p \in\mathcal{P}_{ij}}
} \, h_p^{k,x} + u_{ij}^{k,x} \, > \,
h_p^{k,x}, \epc \forall \, 
p \, \in \, {\cal P}_{ij}.
\]
Hence, by complementarity, $v_p^{k,x} = 0$
for all $p \in {\cal P}_{ij}$. This yields
\[
t_{ij}^{k,x} \, = \, 
\mu^{k,x} \, C_p(\boldsymbol{h}) \, < \, 
\bar{t}
\]
which contradicts the choice of $t_{ij}^{k,x}$.
Similarly, we can prove that
all the other bound multipliers are zero.

Statement (B) is an immediate consequence of 
Proposition~\ref{pr:existence VI for traffic}
under the identifications: 
$\boldsymbol{x}$ being the tuple
$\left( \boldsymbol{z},
\boldsymbol{h},\boldsymbol{\theta} \right)$,
$\boldsymbol{y} = \boldsymbol{\nu}$, 
$\boldsymbol{X} = {\cal Z} \times {\cal H} 
\times [ \, 0, \, 1 \, ]^K$,
$\boldsymbol{Y}$ being 
$\mathbb{R}_+^{2| {\cal K} |}$ and
$\boldsymbol{y}^{\rm ref}$ being the origin,
provided that we can show that the tuples
$\{ \boldsymbol{\nu}^{\tau} \}$ are bounded
for $\tau > 0$, where each 
$\boldsymbol{\nu}^{\tau}$ satisfies
\begin{equation} \label{eq:tau nu}
\begin{array}{l}
0 \, \leq \, \nu^{\tau;k}_{\rm AV} \, \perp \,
\tau \,  \nu^{\tau;k}_{\rm AV} + 
\mu^{\rm cap}_{\rm AV} N^k \, - \\ [0.1in]
\hspace{1in} \displaystyle{
\sum_{(i,j)\in \mathcal{W}^{k,{\rm AV}}}
} \, \displaystyle{
\sum_{s\in\mathcal{D}^{k,{\rm AV}}}
} \, ( t_{s,i}^{\tau;k,{\rm AV}} +
t_{ij}^{\tau;k,{\rm AV}} )
z^{\tau;k,{\rm AV}}_{s,ij}  
\, \geq \, 0, \epc 
\forall \, k \in \mathcal{K} \\ [0.3in]
0 \, \leq \, \nu^{\tau;k} \, \perp \,  
\tau \, \nu^{\tau;k} + N^k \, - \\ [0.1in]
\hspace{1in} \displaystyle{
\sum_{x \in {\cal X}}
} \, \displaystyle{
\sum_{(i,j)\in \mathcal{W}^{k,x}}
} \, \displaystyle{
\sum_{s\in\mathcal{D}^{k,x}}
} \, ( t_{s,i}^{\tau;k,x} + 
t_{ij}^{\tau;k,x} ) z_{ij}^{\tau;k,x}
\, \geq \, 0, \epc 
\forall \, k \in \mathcal{K} % \\ [0.3in]
\end{array} \end{equation}
for some (bounded) tuple
$\left( \boldsymbol{z}^{\tau},
\boldsymbol{h}^{\tau},
\boldsymbol{\theta}^{\, \tau} \right) \in 
{\cal Z} \times {\cal H} \times 
[ \, 0,1 \, ]^K$,
which along with suitable multipliers,
satisfy the NCP$_{\rm main}$.  For convenience
of reference, we
% (\ref{eq:NCP mixed-fleet I})
% and (\ref{eq:NCP mixed-fleet II}) 
re-write in full this NCP with all
inequalities along with their respective
(nonnegative) multipliers:
\[
% \begin{array}{l}
\begin{array}{l}
0 \, \leq \, z_{s,ij}^{\tau;k,{\rm AV}} 
\, \perp \, - \wt{R}_{s,ij}^{k,{\rm AV}} -
\alpha_1^{k,x} t_{ij}^{\tau;k,{\rm AV}} + 
\beta_{1}^{k,x} ( t_{s,i}^{\tau;k,{\rm AV}} + 
t_{ij}^{\tau;k,{\rm AV}} ) - 
\beta_3^{k,{\rm AV}}
t_{s,i}^{\tau;k,{\rm AV}}\, + \\ [0.1in]
\hspace{0.4in} \phi_s^{\tau;k,{\rm AV}} - 
\lambda_{ij}^{\tau;k,{\rm AV}} + 
( t_{s,i}^{\tau;k,{\rm AV}} + 
t_{ij}^{\tau;k,{\rm AV}} ) \,
( \nu^{\tau;k} + \nu^{\tau;k}_{\rm AV} )  
\, \geq \, 0, \\ [0.1in]
\hspace{2in} \forall \, k \in \mathcal{K}, \, 
s \in {\cal D}^{k,{\rm AV}}, 
\, (i,j) \in \mathcal{W}^{k,{\rm AV}}  
\\ [0.1in]
0 \, \leq \, z_{s,ij}^{\tau;k,{\rm HV}} 
\, \perp \, - \wt{R}_{s,ij}^{k,{\rm HV}} -
\alpha_1^{k,x} t_{ij}^{\tau;k,{\rm HV}} + 
\beta_{1}^{k,x} ( t_{s,i}^{\tau;k,{\rm HV}} + 
t_{ij}^{\tau;k,{\rm HV}} ) - 
\beta_3^{k,{\rm HV}}
t_{s,i}^{\tau;k,{\rm HV}}\, + \\ [0.1in]
\hspace{0.4in} \phi_s^{\tau;k,{\rm HV}} - 
\lambda_{ij}^{\tau;k,{\rm HV}} + 
( t_{s,i}^{\tau;k,{\rm HV}} + 
t_{ij}^{\tau;k,{\rm HV}} ) \, \nu^{\tau;k} 
\, \geq \, 0, \\ [0.1in]
\hspace{2in} \forall \, k \in \mathcal{K}, \, 
s \in {\cal D}^{k,{\rm HV}}, 
\, (i,j) \in \mathcal{W}^{k,{\rm HV}}  
\\ [0.2in]
0 \, \leq \, \phi_s^{\tau;k,x} \, \perp \, 
\displaystyle{
\sum_{i\in\mathcal{O}^{k,x}}
} \, D_{is}^{\tau;k,x} - \displaystyle{
\sum_{(i,j)\in \mathcal{W}^{k,x}}
} \, z_{s,ij}^{\tau;k,x} \, \geq \, 0, \ 
\forall \, k\in\mathcal{K}, \, 
x \in \mathcal{X}, 
\, s \in \mathcal{D}^{k,x} \\ [0.3in]
0 \, \leq \, \lambda_{ij}^{\tau;k,x} \, \perp \, 
\displaystyle{
\sum_{s \in \mathcal{D}^{k,x}}
} \, z_{s,ij}^{\tau;k,x} - D_{ij}^{\tau;k,x} 
\, \geq \, 0, \hspace{0.2in} 
\forall \, k \,\in \mathcal{K}, \, 
x \in \mathcal{X}, 
\, (i, j) \in \mathcal{W}^{k,x}
\end{array} \]
\[ \begin{array}{l}
% \\ [0.25in]
0 \, \leq \, D_{ij}^{\tau;k,x} \, \perp \, 
F_{ij}^{k,x} 
+ \alpha_{1}^{k,x} (t_{ij}^{\tau;k,x}- 
t_{ij}^{\, 0}) + \alpha_{2}^{k,x} 
d_{ij}^{\, 0}
+ \gamma_{1}^{k,x}t_{ij}^{\tau;k,x} +
\gamma_{2}^{k,x} w_{ij}^{\tau;k,x} \\ [0.1in]
\hspace{1in} - \, \sigma_{ij}^{\tau} + 
% \gamma_3^{k,x}
\lambda_{ij}^{\tau;k,x} \, 
\geq \, 0, \hspace{0.3in} 
\forall \, (k, x) \in \mathcal{K} \times 
\mathcal{X}, \, (i,j) \in \mathcal{W}^{k,x} 
\\ [0.1in]
0 \, \leq \, D_{ij}^{\tau;{\rm SV}} 
\, \perp \, 
\alpha_1^{\rm SV}t_{ij}^{\tau;{\rm SV}} + 
\alpha^{\rm SV}_2 d_{ij}^{\, 0} - 
\sigma_{ij}^{\tau} \, \geq \, 0, \epc 
\forall \, (i,j) \in \mathcal{W} \\ [0.2in]
0 \, \leq \sigma_{ij}^{\tau} \, \perp \,
D_{ij}^{\tau;\rm SV} + \displaystyle{
\sum_{k\in\mathcal{K}}
} \, \displaystyle{
\sum_{x\in\mathcal{X}}
} \, D_{ij}^{\tau;k,x} - D_{ij} 
\, \geq \, 0, \epc 
\forall \, (i,j) \in \mathcal{W} \\ [0.2in]
0 \, \leq \, t_{ij}^{\tau;{\rm SV}} 
\, \perp \, \displaystyle{
\sum_{p \in \mathcal{P}_{ij}}
} \, h_p^{\tau;{\rm SV}} - 
D_{ij}^{\tau;{\rm SV}}
\, \geq \, 0, \hspace{0.8in}  
\forall \, (i,j) \in {\cal W} \\ [0.25in]
0 \, \leq \, t_{ij}^{\tau;k,x} \, \perp \,
\displaystyle{
\sum_{p \in\mathcal{P}_{ij}}
} \, h_p^{\tau;k,x} - D_{ij}^{\tau;k,x} 
\, \geq \, 0,
\hspace{0.4in} \forall \, 
x \in \mathcal{X}, \, k \in \mathcal{K}, 
\, (i,j) \in {\cal W}^{k,x} \\ [0.25in]
0 \, \leq \, t_{s,i}^{\tau;k,x} \, \perp \,
\displaystyle{
\sum_{p \in\mathcal{P}_{si}}
} \, h_p^{\tau;k,x} - \displaystyle{
\sum_{j\in\mathcal{D}^{k,x}}
} \, z_{s,ij}^{\tau;k,x} \, \geq \, 0, 
\\ [0.2in]
\hspace{2in} \forall \, x \in \mathcal{X}, \, 
k \in \mathcal{K}, \, 
(s,i) \in \mathcal{D}^{k,x} \times 
\mathcal{O}^{k,x} \\ [0.1in]
0 \leq h_p^{\tau;\rm SV} \, \perp \, 
C_p( \boldsymbol{h}^{\tau} ) - 
t_{ij}^{\tau;{\rm SV}} \, \geq \, 0, 
\hspace{0.8in} \forall \, (i,j) \in {\cal W}, 
\, p \in\mathcal{P}_{ij} \\ [0.1in]
0 \leq h_p^{\tau;k,{\rm AV}} \, \perp \, 
\mu^{k,{\rm AV}} C_p(\boldsymbol{h}^\tau) - 
t_{ij}^{\tau;k,{\rm AV}} \geq 0,
\epc \forall \, (i,j) \in {\cal W}, \, 
p \in \mathcal{P}_{ij}, \, k\in\mathcal{K} 
\\ [0.1in]
0 \, \leq \, h_p^{\tau;k,{\rm AV}} \, 
\perp \,  C_p( \boldsymbol{h}^{\tau} ) - 
t_{s,i}^{\tau;k,{\rm AV}} \geq 0, \epc 
\forall \, 
(s,i) \in \mathcal{D} \times \mathcal{O}, 
\, p \in\mathcal{P}_{si}, \, 
k \in \mathcal{K} \\ [0.1in]
0 \, \leq \, h_p^{\tau;k,{\rm HV}} 
\, \perp \, 
C_p( \boldsymbol{h}^{\tau}) 
- t_{ij}^{\tau;k,{\rm HV}} \, \geq \, 0, 
\hspace{0.2in} 
\forall \, (i,j) \in {\cal W}, \, 
p \in \mathcal{P}_{ij}, \, k\in\mathcal{K} 
\\ [0.1in]
0 \, \leq \, h_p^{\tau;k,{\rm HV}} 
\, \perp \, 
\mu^{k,{\rm HV}} C_p( \boldsymbol{h}^{\tau}) - 
t_{s,i}^{\tau;k,{\rm HV}} \, \geq \, 0, 
\epc \forall \, 
(s,i) \in \mathcal{D}\times\mathcal{O}, 
\, p \in\mathcal{P}_{si}, \, 
k \in \mathcal{K}.
\end{array} \]
% there exists
% $\bar{\sigma}$ such that
% $\sigma_{ij}^{\tau} \leq \bar{\sigma}$ for
% $ all $(i,j) \in {\cal W}$.  
We sum up the complementarity 
conditions for a sequence of positive scalars
$\{ \tau_n \}$ up to the 
$\sigma_{ij}$-complementarities and mark 
all the terms that can be canceled when
these equations are added up (in what follows,
we use the superscript ``$n$'' as a short-hand
for $\tau_n$ in the variables):
\[ \begin{array}{l}
\displaystyle{
\sum_{k\in\mathcal{K}}
} \displaystyle{
\sum_{s \in \mathcal{D}^{k,{\rm AV}}}
} \displaystyle
\sum_{(i,j) \in \mathcal{W}^{k,{\rm AV}}
}  
% \gamma_3^{k,x}
\left\{ \! \begin{array}{l}
- \wt{R}_{s,ij}^{\, k,{\rm AV}}
z_{s,ij}^{n;k,{\rm AV}} + 
(\beta_1^{k,x} - \alpha_1^{k,{\rm AV}}) 
t_{ij}^{n;k,{\rm AV}} z_{s,ij}^{n;k,{\rm AV}} 
+ \\ [0.15in] 
(\beta_1^{k,{\rm AV}} - 
\beta_3^{k,{\rm AV}} ) 
t_{s,i}^{n;k,{\rm AV}}z_{s,ij}^{n;k,{\rm AV}} 
+ \color{green}\cancel{\phi_s^{n;k,{\rm AV}}
z_{s,ij}^{n;k,{\rm AV}}}  \\ [0.15in]
- \, {\color{red}\cancel{\lambda_{ij}^{n;k,{\rm AV}}
z_{s,ij}^{n;k,{\rm AV}}}} 
+ (t_{s,i}^{n;k,{\rm AV}} + 
t_{ij}^{n;k,{\rm AV}} ) ( \nu^{n;k}_{\rm AV} 
+ \nu^{n;k} ) z_{s,ij}^{n;k,{\rm AV}}
\end{array} \right\} + \\ [0.55in]
\displaystyle{
\sum_{k\in\mathcal{K}}
} \, \displaystyle{
\sum_{s \in \mathcal{D}^{k,{\rm HV}}}
} \, \displaystyle
\sum_{(i,j) \in \mathcal{W}^{k,{\rm HV}}
} \,  
% \gamma_3^{k,x}
\left\{ \begin{array}{l}
- \wt{R}_{s,ij}^{\, k,{\rm HV}}
z_{s,ij}^{n;k,{\rm HV}} + 
(\beta_{1}^{k,{\rm HV}} - 
\alpha_{1}^{k,{\rm HV}}) 
t_{ij}^{n;k,{\rm HV}} z_{s,ij}^{n;k,{\rm HV}} 
+ \\ [0.15in] 
(\beta_1^{k,{\rm HV}} - 
\beta_3^{k,{\rm HV}} ) 
t_{s,i}^{n;k,{\rm HV}}z_{s,ij}^{n;k,{\rm HV}} 
+ \color{green}\cancel{\phi_s^{n;k,{\rm HV}}}
z_{s,ij}^{n;k,{\rm HV}} \\ [0.15in]
- \, {\color{red}\cancel{\lambda_{ij}^{n;k,{\rm HV}}
z_{s,ij}^{n;k,{\rm HV}}}} + 
(t_{s,i}^{n;k,{\rm HV}} + 
t_{ij}^{n;k,{\rm HV}} ) \nu^{n;k}  
z_{s,ij}^{n;k,{\rm HV}}
\end{array} \right\} \, + 
\end{array} \]
\[ \begin{array}{l}  % \\ [0.5in]
\displaystyle{
\sum_{k\in\mathcal{K}}
} \, \displaystyle{
\sum_{x\in\mathcal{X}}
} \, \displaystyle{
\sum_{s\in\mathcal{D}^{k,x}}
} \,  
% \gamma_3^{k,x}
\left\{ \, \underbrace{\displaystyle{
\sum_{i\in\mathcal{O}^{k,x}}
} \, D_{is}^{n;k,x}\phi_s^{n;k,x}}_{\mbox{
nonnegative}} -
{\color{green}\cancel{\displaystyle{
\sum_{(i,j)\in \mathcal{W}^{k,x}} 
} \, z_{s,ij}^{n;k,x}\phi_s^{n;k,x}}} \right\} 
\, + \\ [0.4in]
\displaystyle{
\sum_{k\in\mathcal{K}}
} \, \displaystyle{
\sum_{x\in\mathcal{X}}
} \, \displaystyle{
\sum_{(i,j)\in\mathcal{W}^{k,x}} 
} \,  
% \gamma_3^{k,x}
\left\{ {\color{red}\cancel{
\displaystyle{
\sum_{s \in \mathcal{D}^{k,x}} 
} \, z_{s,ij}^{n;k,x}\lambda_{ij}^{n;k,x}}} - 
{\color{purple}\cancel{D_{ij}^{n;k,x}\lambda_{ij}^{n;k,x}}} 
\right\} \, + \,  \\ [0.3in]
\displaystyle{
\sum_{(i,j)\in\mathcal{W}}
} \, \displaystyle{
\sum_{x\in\mathcal{X}}
} \, \displaystyle{
\sum_{k\in\mathcal{K}_{ij}^x}
} \, \left\{\begin{array}{l} 
F_{ij}^{k,x}D_{ij}^{n;k,x} 
+ \alpha_{1}^{k,x} (t_{ij}^{n;k,x} - 
t_{{ij}}^0)D_{ij}^{n;k,x} \, + \\ [0.15in]
\alpha_{2}^{k,x} d_{ij}^{\, 0}D_{ij}^{n;k,x}
+ \gamma_{1}^{k,x}t_{ij}^{n;k,x}
D_{ij}^{n;k,x} + 
\gamma_{2}^{k,x} {w}_{ij}^{n;k,x}
D_{ij}^{n;k,x} \, - \\ [0.15in]
\cancel{\sigma_{ij}^nD_{ij}^{n;k,x}} + 
{\color{purple}\cancel{
% \gamma_3^{k,x}
\lambda_{ij}^{n;k,x} 
D_{ij}^{n;k,x}}} \end{array} \right\} +
\\ [0.5in]
\displaystyle{
\sum_{(i,j)\in\mathcal{W}}
} \, \left\{ 
\alpha_1^{\rm SV}t_{ij}^{n;\rm SV}
D_{ij}^{n;\rm SV} + 
\alpha^{\rm SV}_2 d_{ij}^{\, 0}
D_{ij}^{n;\rm SV} - 
\cancel{\sigma_{ij}^nD_{ij}^{n;{\rm SV}}} 
\right\} + \\ [0.1in] 
\displaystyle{
\sum_{(i,j)\in\mathcal{W}}
} \left\{ 
\cancel{D_{ij}^{n;\rm SV}\sigma_{ij}^n} + 
\cancel{{\displaystyle{
\sum_{k\in\mathcal{K}}
} \, \displaystyle{
\sum_{x\in\mathcal{X}}
} \, D_{ij}^{n;k,x}\sigma_{ij}^n}} - 
D_{ij}\sigma_{ij}^n \right\} = 0
\end{array} \]
After all the cancellations, the
only possibly negative unbounded
terms on the left-hand side is 
 $-D_{ij} \sigma_{ij}^n$. 
% with the former two carrying a positive sign, 
% and the latter a negative sign.
Since $\{ t_{ij}^{n;{\rm SV}} \}$ is 
bounded, it follows from
$\alpha_1^{\rm SV}t_{ij}^{n;{\rm SV}} + 
\alpha^{\rm SV}_2 d_{ij}^{\, 0} \geq 
\sigma_{ij}^n$ that 
$\{ \sigma_{ij}^n \}$ is bounded. 
Consequently, since the left-hand side sums up 
to zero, it follows that the (nonnegative) terms
\[
(t_{s,i}^{n;k,{\rm AV}} + 
t_{ij}^{n;k,{\rm AV}} ) \nu^{n;k}_{\rm AV} 
z_{s,ij}^{n;k,{\rm AV}} \ \mbox{ and } \
(t_{s,i}^{n;k,{\rm AV}} + 
t_{ij}^{n;k,{\rm AV}} ) \nu^{n;k} 
z_{s,ij}^{n;k}
\]
are bounded. 
% Moreover, the two complementarity conditions
% (\ref{eq:tau nu}) imply that 
% $\{ \tau_n \nu^{n;k} \}$ and 
% $\{ \tau_n \nu^{n;k}_{\rm AV} \}$ are both
% bounded.  
Suppose there is
a sequence of positive scalars $\{ \tau_n \}$
such that $\{ \nu^{n;k} \} \to \infty$ 
as $n \to \infty$ for some 
$k \in {\cal K}$. It follows from the
above summation that 
$\left\{ \, \displaystyle{
\sum_{x \in X}
} \, \displaystyle{
\sum_{(i,j) \in {\cal W}^{k,x}}
} \, \displaystyle{
\sum_{s \in {\cal D}^{k;x}}
} \, z_{s,ij}^{n;k,x} ( t_{s,i}^{n;k,x} +
t_{ij}^{n;k,x} ) \, \right\} \to 0$
as $n \to \infty$.  
Since
$\nu^{n;k} > 0$ for all $n$
sufficiently large, we have
\[
\tau_n \, \nu^{n;k} + N^k - \underbrace{
\displaystyle{
\sum_{x \in X}
} \, \displaystyle{
\sum_{(i,j) \in {\cal W}^{k,x}}
} \, \displaystyle{
\sum_{s \in {\cal D}^{k;x}}
} \, z_{s,ij}^{n;k,x} ( t_{s,i}^{n;k,x} +
t_{ij}^{n;k,x} )}_{\mbox{converges
to zero}} = 0
\]
which is a contradiction because $N^k > 0$.
In a similar way, we can also obtain a contradiction
if $\{ \nu^{n;k}_{\rm AV} \}$ is unbounded for some
$k$.  
\end{proof}

\section{Proof of Proposition \ref{pr:existence VI for traffic}}
\begin{proof}  We apply the homotopy 
invariance principle of the degree 
\cite{facchinei2003finite}[Definition~2.11, part~(A3)]
 of a continuous mapping
to the homotopy
\[
H(\boldsymbol{x},\boldsymbol{y},t) 
\, \triangleq \, \left( 
\begin{array}{l}
\boldsymbol{x} - \Pi_{\boldsymbol{X}}( 
\boldsymbol{x} - \Phi(\boldsymbol{x},
\boldsymbol{y}) )\\ [0.1in]
\boldsymbol{y} - \Pi_{\boldsymbol{Y}}( t 
(\boldsymbol{y} - 
\Psi(\boldsymbol{x},\boldsymbol{y})) + 
( 1 - t ) \boldsymbol{y}^{\rm ref})
\end{array} \right), \epc \mbox{for 
$t \in [ \, 0,1 \, ]$},
\]
where $\Pi_S$ is the Euclidean projector onto
a closed convex set $S$.  It suffices that 
the set: $\displaystyle{
\operatornamewithlimits{\bigcup}_{
t \in [ \, 0, 1 \, )}
} \, H(\bullet,\bullet,t)^{-1}(0)$ is bounded.
It is easy to see that a pair 
$(\boldsymbol{x},\boldsymbol{y}) 
\in H(\bullet,\bullet,t)^{-1}(0)$ for
some $t \in [ \, 0,1 \, )$ if and only if
$(\boldsymbol{x},\boldsymbol{y})$ is a 
solution of the VI $(F^{\, \tau},
\boldsymbol{X} \times \boldsymbol{Y})$ for 
$\tau = \displaystyle{
\frac{1 - t}{t}
}$.  By assumption, such pair 
$(\boldsymbol{x},\boldsymbol{y})$ is bounded.
\end{proof}

\bibliographystyle{elsarticle-harv}
\bibliography{references} 
\end{document}